\PassOptionsToPackage{table}{xcolor}
\documentclass[sigconf, nonacm]{acmart}

\newcommand\vldbdoi{XX.XX/XXX.XX}
\newcommand\vldbpages{X-X}
\newcommand\vldbvolume{19}
\newcommand\vldbissue{11}
\newcommand\vldbyear{2026}
\newcommand\vldbauthors{\authors}
\newcommand\vldbtitle{\shorttitle} 
\newcommand\vldbavailabilityurl{https://github.com/darroyue/bgps-temporal-graphs}
\newcommand\vldbpagestyle{plain} 

\usepackage{balance}
\usepackage{graphicx}
\usepackage{xcolor}
\usepackage{algorithm2e}
\usepackage{amsmath}
\usepackage{multirow}

\usepackage{amssymb}
\usepackage{enumitem}

\newcommand{\ti}{\mathit{ti}}
\newcommand{\tf}{\mathit{tf}\!}
\newcommand{\ts}{\mathit{ts}}
\newcommand{\te}{\mathit{te}}
\newcommand{\U}{\mathcal U}
\newcommand{\T}{\mathcal T}

\newcommand{\Vu}{\mathcal{V}_\mathrm{u}}
\newcommand{\Vt}{\mathcal{V}_\mathrm{t}}
\newcommand{\no}[1]{}

\newcommand{\Xall}[2]{_{[#1,#2)}}
\newcommand{\Xany}[2]{_{#1}^{#2}}
\newcommand{\Gpt}[1]{\hat{#1}}

\newtheorem{theorem}{Theorem}

\newtheorem{definition}{Definition}

\usepackage{pgfplots}
\pgfplotsset{compat=1.18}
\usepgfplotslibrary{statistics}

\newcommand{\new}[1]{\textcolor{black}{#1}}
\newcommand{\old}[1]{}
\newcommand{\tg}[2]{#2}

\newcommand{\leap}{\mathsf{leap}}

\begin{document}

\title{Worst-Case Optimal BGPs on Temporal Graphs}

% AH:putting back in country macros
% as omitting them generates errors; can find
% space elsewhere hopefully
\author{Diego Arroyuelo}
\affiliation{%
  \institution{Pontificia Universidad Cat{\'o}lica de Chile \& }}
  \affiliation{
  \institution{IMFD, Chile}}
\email{diego.arroyuelo@uc.cl}

\author{Aidan Hogan}
\affiliation{%
  \institution{DCC, Universidad de Chile \&}}
\affiliation{%  
    \institution{IMFD, Chile}}
\email{ahogan@dcc.uchile.cl}

\author{Gonzalo Navarro}
\affiliation{%
  \institution{DCC, Universidad de Chile \&}}
\affiliation{%  
    \institution{IMFD, Chile}}
\email{gnavarro@dcc.uchile.cl}

\author{Juan Reutter}
\affiliation{%
  \institution{Pontificia Universidad Cat{\'o}lica de Chile \&}}
  \affiliation{
  \institution{IMFD, Chile}}
\email{jreutter@ing.puc.cl}

\begin{abstract}
We study how to evaluate basic graph patterns (BGPs) in a worst-case-optimal (wco) manner over {\em temporal} labeled graphs, where edges have an interval of temporal validity. 
We adopt a flexible query language 
in which users specify $m$ quads of the form (subject, property, object, time), using constants or variables. The time component denotes the instant at which a particular edge is valid, and users may also include order relations between temporal constants or variables. The answer is the set of all valid variable assignments, including time.
We describe an index structure that, for a temporal graph with $N$ edges, requires $O(N)$ space and can evaluate extended BGPs in wco time $O(Q^* m \log N)$, where $Q^*$ represents the maximum number of solutions for query $Q$ over any temporal graph with the same number of instants of edge validity. 
We use our index to adapt Leapfrog Triejoin to the temporal graph setting under any variable evaluation ordering.  
Our index further yields wco guarantees for related query types, including snapshot evaluation, version queries, and other temporal variants.
Experiments on real-world datasets show that our approach answers realistic queries in milliseconds with low space overhead.
%Moreover, we show that the same index provides worst-case-optimal guarantees for several related query types, including non-temporal BGPs evaluated on a snapshot or time slice of the temporal graph, queries over graph versions, and other temporal variations. 
%Using a battery of experiments with real-world data we validate that our index structure can be used to answer realistic queries within a few milliseconds in most cases, and with little space overhead.

%We consider extending basic graph patterns (BGPs) to {\em temporal} labeled graphs, where the edges have an interval of temporal validity. 

\end{abstract}

\maketitle

%%% do not modify the following VLDB block %%
%%% VLDB block start %%%

\pagestyle{\vldbpagestyle}
\begingroup\small\noindent\raggedright\textbf{PVLDB Reference Format:}\\
\vldbauthors. \vldbtitle. PVLDB, \vldbvolume(\vldbissue): \vldbpages, \vldbyear.\\
\href{https://doi.org/\vldbdoi}{doi:\vldbdoi}
\endgroup
\begingroup
\renewcommand\thefootnote{}\footnote{\noindent
This work is licensed under the Creative Commons BY-NC-ND 4.0 International License. Visit \url{https://creativecommons.org/licenses/by-nc-nd/4.0/} to view a copy of this license. For any use beyond those covered by this license, obtain permission by emailing \href{mailto:info@vldb.org}{info@vldb.org}. Copyright is held by the owner/author(s). Publication rights licensed to the VLDB Endowment. \\
\raggedright Proceedings of the VLDB Endowment, Vol. \vldbvolume, No. \vldbissue\ %
ISSN 2150-8097. \\
\href{https://doi.org/\vldbdoi}{doi:\vldbdoi} \\
}\addtocounter{footnote}{-1}\endgroup
%%% VLDB block end %%%

%%% do not modify the following VLDB block %%
%%% VLDB block start %%%
\ifdefempty{\vldbavailabilityurl}{}{
\vspace{.3cm}
\begingroup\small\noindent\raggedright\textbf{PVLDB Artifact Availability:}\\
The source code, data, extended version and/or other artifacts have been made available at \url{https://github.com/darroyue/bgps-temporal-graphs} %\url{https://anonymous.4open.science/r/bgps-temporal-graphs-561D}.
\endgroup
}
%%% VLDB block end %%%

\section{Introduction}

A significant advance for solving Basic Graph Patterns (BGPs) was the discovery of the AGM bound \cite{AGM13} and worst-case-optimal (wco) joins. A graph database, in the simplest RDF-like format, is \new{a labeled graph, where the edges $s \stackrel{p}{\to} o$ are} seen as a set of triples $(s,p,o)$ for subject, predicate, and object. A BGP is a set of triple patterns $(s,p,o)$, where each component can be a constant or a variable. Solving a BGP returns all assignments of values to variables so that the triple patterns, when instantiated to triples, appear in the graph. A wco algorithm takes time proportional to the size of the output of the query on some graph, such as Leapfrog TrieJoin (LTJ) \cite{leapfrog}, which binds---i.e., finds all the possible values of---one variable, and for each value assigned to it, recurses on the remaining variables. 

Herein we consider {\em temporal graphs} with quads $(s,p,o,[\ti,\tf))$ denoting that %\old{triple $(s,p,o)$ in the graph} 
\new{the graph edge $s \stackrel{p}{\to} o$} was valid from time $\ti$ (inclusive) to time $\tf$ (exclusive). Temporal graphs give relations a temporal range of validity (e.g., when a researcher was affiliated with a university), or a temporal instant (e.g., when a conference took place), and support querying complex relations (i.e., subgraphs, paths) with temporal ranges and instants (e.g., which researchers were affiliated with a university when a given conference was held).  \new{Temporal graph querying arises naturally in real-world systems where relations evolve over time. Examples include event modeling~\cite{eventkg}, cybersecurity systems~\cite{cyber}, and AI applications~\cite{aitg1,aitg2}.}% These settings require efficient evaluation of complex graph patterns together with temporal constraints.}
%\textcolor{red}{ahora se ve medio raro el ejemplo acerca de India :-)}

We aim at solving BGPs on temporal graphs. 
The simplest form of such BGP queries is the \emph{return all time points} query, which looks to find all answers for the BGP that were valid at some point in time, together with the time where each of these answers exists in the database.  
A folklore way to solve this problem is a so-called {\em join-first} strategy: we first compute the answers of the BGP disregarding the time information, and then, for each solution to the BGP, we use the time information of the concrete edges to determine the points in time this solution is valid, if any. \new{An advantage of join-first is that one is free to solve the BGP with any wco strategy \cite{leapfrog,NPRR12,tutorialngo}, beyond-wco ones \cite{emptyheaded,geometric,abo2017shannon} and even combinations of wco and  non-wco algorithms \cite{graphflow,umbra,freejoin,adopt}. Alternatively,} Hu et al.~\cite{HSGAY22} introduce a so-called {\em time-first} strategy, which sweeps over the time component and recomputes the BGP \new{answer} each time an event changes the set of edges present in the graph. Though they show good theoretical results for restricted kinds of BGPs, in the general case the whole BGP \new{answer} is recomputed for every event. 

We argue that these solutions suffer from two main limitations. First, on the usability side, the \emph{return all time points} query may lack expressiveness in several applications. For example, the latest SQL 2011 standard recognizes the importance of returning answers that mix temporal constraints (e.g., joining two triples on the condition that one precedes a certain timestamp and the other comes after it) \cite{sql2011temporal}. 
Second, \new{while join-first strategies can handle such queries, they generate all solutions without considering time and only then filter by time, which is inefficient if the time filters are selective. Regarding LTJ, in particular,} both time- and join-first algorithms translate to fixing particular variable binding orderings, where it is well known that choosing a good variable binding ordering is crucial for efficient query resolution, even if all orderings are theoretically optimal \cite{leapfrog}. Hence, systems processing temporal graph queries would benefit from an \new {approach that can effectively integrate time constraints within the core of a wco (or beyond-wco!) algorithm.}

\begin{example} \label{ex:graph}
\new{As a motivating example, consider a graph describing the bus trips over a year in a big city, with quads $(a,\textsf{trip},b,[\ti,\tf))$ indicating that person $a$ was on bus $b$ during time interval $[\ti,\tf)$. Figure~\ref{fig:trips} shows a toy example (all the edges are labeled \textsf{trip}). A query looking for pairs of people that were together in a specific bus ``\textsf{B12}'' could ask to return all time points for the BGP $\{ (x,\textsf{trip},\textsf{B12}), (y,\textsf{trip},\textsf{B12}) \}$. A join-first strategy would find all the pairs of people $(x,y)$ that had ever traveled in that bus, and only then filter who did so at the same time. If $m$ people take that bus at some time, each taking it $v$ times during the year, join-first involves 
%an intermediate output size of $\Omega(m^2)$, and likely 
a time complexity of $O(m^2 v)$. A time-first approach would, instead, consider each of the $t$ time instants of the year (at a granularity of, say, minutes) and, for each of those, get the pairs of $b$ people on that bus at that time, for a total time complexity of %$\Omega(t)$, likely 
$O(tb^2)$. %Yes: \ah{Wouldn't it be $O(tb^2)$ since we could need to write pairs of people for each time instance?} 
A better strategy would be to take each of the $m$ persons $x$ that take that bus, for each of them the $v$ times they took it (this should be once or twice a day), and then find the other $b$ persons on that same bus at that time, for a total time complexity of $O(mvb)$, which should be close to the size of the final output and thus optimal. It is further expected that $b \ll m$ and $mv \ll t$ (since $mv$ is the total number of times bus \textsf{B12} was taken, a fraction of the universe).}

\new{As an example of the first limitation mentioned above, the time-first strategy would have trouble finding pairs of people that were together in the bus for over an hour, or people that took the bus \textsf{B12} before noon and then the bus \textsf{B6} in the afternoon (those are beyond a simple {\em return all time points} query). \qed}
\end{example}

\begin{figure}[t]
\includegraphics[width=0.4\textwidth]{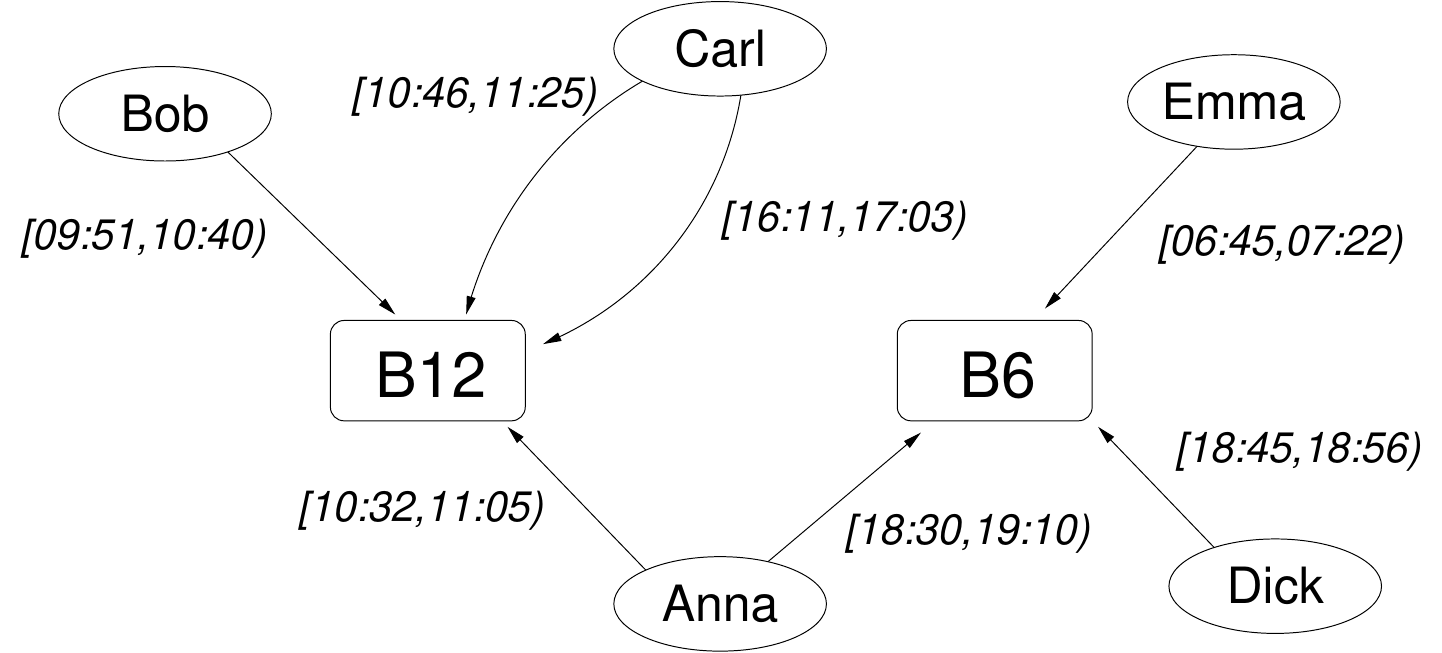}
%\vspace*{-5pt}
\caption{\new{A toy temporal graph about trips on buses. For simplicity the timestamps show only hours during some day.}}
\label{fig:trips}
\end{figure}

%\ah{Me parece muy claro el ejemplo. Puse un código B12 para la instancia de la clase bus, para distinguirlo de la clase. No eetendí a qué se hace referencia la second limitation? \textit{Quizás} valga destacar que intuitivamente, en este escenario, $b \ll m$ (el bus tiene capacidad limitada) y $v \ll t$ (los pasajeros, en general, están en el bus una fracción del tiempo).}

We introduce a linear-space temporal graph representation that overcomes both limitations. Our index structure stores temporal graphs in terms of intervals defined by their updates, using an improvement of versioned binary trees that avoids a space blowup by taking advantage of compressed representations. Using these structures, we obtain the following contributions. 

(1) Our structure enables retrieving all time points of any given BGP (via LTJ), allowing any possible variable binding order, with (at least) the same complexity guarantees that time-first or join-first algorithms offer. This already improves the current state of the art with respect to worst-case optimal joins over temporal graphs. 

(2) We show that we can answer much more general queries. \new{More precisely, we introduce \emph{temporal} BGPs (tBGPs), which are a temporal analog to BGPs, and show how our data structure can process any tBGP query. This language allows time-aware pattern matching, retreiving different time intervals for different parts of the query, and applying filters to restrict these intervals.  }
%These queries allow for retrieving different time intervals for different parts of the query, and to apply filters restricting these intervals. 

(3) We can answer point-in-time or interval queries that return all answers valid in a given timestamp or time period with stronger guarantees: we do it in wco time with respect to the graph restricted to that point or interval rather than the complete graph. This provides a further speedup for point-in-time or interval queries, which are key primitives of the versioning functionalities of the SQL 2011 standard and several graph database proposals~\cite{tsparql,glenda,temporalalgebragraph,versionedaidan}.
%\end{itemize}

%Juan: mejor ponerlo como (4) o no?
(4) We provide a prototype implementation and show that our improvements translate to practice, solving realistic temporal queries within milliseconds while using reasonable index space\new{. It also clearly outperforms both join-first and time-first strategies, as well as two prominent solutions from previous work \cite{HSGAY22,ZhuFY21}}.

\section{Related Work}

Temporal query languages have been extensively studied in relational and graph databases, leading to temporal extensions of SQL:2011~\cite{sql2011temporal} and to proposals for RDF and graph query languages such as T-SPARQL~\cite{tsparql}, stSPARQL~\cite{stSPARQL}, and SPARQL-ST~\cite{sparqlst}. Systems such as GLENDA~\cite{glenda} and RDF archives~\cite{temporalontology, versionedaidan} further support querying historical graph snapshots and tracking data evolution over time. Our work is complementary, focusing on efficient query evaluation rather than models or languages.

%Temporal query languages originated from research in relational databases. Early efforts such as TQuel \cite{tquel} and TSQL2 \cite{tsql2} introduced time constructs such as \texttt{OVERLAPS} and  \texttt{DURING} to retrieve data as it existed at a specific time point or interval. 
%These ideas led to the inclusion of temporal features in the SQL:2011 standard \cite{sql2011temporal}, and inspired  research for temporal query languages for RDF or graph databases. In terms of temporal graph languages, T-SPARQL \cite{tsparql} looks to adapt TSQL2 into SPARQL, incorporating similar temporal predicates and interval comparisons. More recent proposals include stSPARQL \cite{stSPARQL} and SPARQL-ST \cite{sparqlst}, which also support spatial reasoning, and further work on temporal path queries \cite{arenas2022temporal}. We also highlight work on RDF querying systems such as GLENDA~\cite{glenda} and RDF archives~\cite{temporalontology, versionedaidan}, which are more geared towards retrieving historical snapshots and tracking versions over time. 

Despite advances in temporal languages and data models, the body of work in algorithms supporting these query extensions is much thinner. Berberich et al. focus on temporal extensions to document databases (see e.g. \cite{timetext1,timetext2}), but these algorithms cannot be directly applied for graph patterns.  Work on temporal constructs or temporal algebras, for relations or graphs, can lead directly to algorithms~\cite{statementmodifiers,temporalalgebrarelational,temporalalgebragraph}, \new{and there is also body of work on extending traditional relational indexes to support temporal constructs such as duration or range (see e.g. \cite{kriegel2000managing,ceccarello2023indexing}),}
but all of these are based on traditional techniques and not wco joins. %In addition, all of these contributions emphasize the need for tailored evaluation strategies considering both the temporal constructs and the structural complexity of the data. 
Relatedly, Khamis et al.~\cite{abo2022complexity} study the complexity of Boolean conjunctive queries with \emph{intersection joins}, which conceptually align with temporal constraints that intersect time domains. Their techniques can be adapted  to temporal joins, typically incurring an additional logarithmic indexing factor, while our work targets general temporal graph patterns in linear space. Finally, \new{Khurana and Deshpande~\cite{khurana2016efficient} and Hou et al.~\cite{hou2025efficient} propose fully-functional systems with support for evolving databases by storing over-time differences. However, their query engines are designed to exploit other components of relational systems such as data partitioning or cache optimization (under the assumption that  temporal queries in practice often focus on recent data), which cannot be directly applied to wco strategies. } 

%\new{Several works study algorithmic support for querying evolving graphs through snapshot-based systems~\cite{khurana2016efficient}, temporal indexing techniques based on interval-oriented indexing structures~\cite{kriegel2000managing,ceccarello2023indexing}, and dynamic graph query processing~\cite{hou2025efficient} based on over-time differences. These approaches typically focus on accelerating specific classes of temporal queries, such as historical snapshot retrieval, temporal reachability, or interval-overlap queries, maintaining query answers under updates, or organizing temporal edges for efficient historical access. Our work is complementary: we study worst-case-optimal evaluation algorithms for arbitrary temporal graph patterns under flexible temporal constraints.}

Recent work on algorithms for temporal graph query processing focuses on evaluation under specific data-time constraints, 
such as retrieving BGPs that occur at the same time within an interval~\cite{ZhuFY21}, answers that exist for a long period of time~\cite{paranjape2017motifs,li2022durable}, for the longest period~\cite{semertzidis2018top}, or answers where edges respect time constraints, in the sense that every match for $(s,p,o)$ must satisfy that $s$ exists before $o$~\cite{gao2020time,huan2025tematch}.  
%
%An example is to retrieve those answers that exist for a long period of time, called \emph{durable} ~\cite{paranjape2017motifs,li2022durable}, or even to look for answers that exist for the longest period~\cite{semertzidis2018top}.
%Another example is to look for answers where edges respect time constraints, in the sense that every match for a triple $(s,p,o)$ must satisfy that $s$ exists before $o$~\cite{gao2020time,huan2025tematch}. 
Yet, none of these algorithms can be applied directly in our case because they are designed for specific time constraints. In contrast, our work focuses on answering arbitrary temporal graph patterns, which is a more general problem: all of these specific time constraints can be expressed as temporal graph patterns, whose expressive power goes well beyond these constraints. 

Our work is closest to that of Hu et al.~\cite{HSGAY22}, addressing join-first and time-first strategies for wco joins over temporal graphs. In comparison, our approach handles any variable binding ordering, more general queries, and interval or point queries with better guarantees. Their algorithm sweeps over the time domain, which takes \new{$\Theta(T)$ time over $T$ time instants}. This corresponds, in our work, to a time-first strategy. Though we support all variable binding plans (including time-first), since Hu et al.\ first sweep over time and then process the query in two independent steps, they can use more sophisticated algorithms for the (non-temporal) query part, such as EmptyHeaded \cite{emptyheaded}, which uses the Generalized Hypertree Decomposition (GHD) of queries. Moreover, as they show in their paper, their idea of only obtaining the solutions where the disappearing edge participates yields better results on hierarchical queries. Finally, their  paper also proposes a cross between time-first and join-first combined with GHD-based algorithms, which again may work better or worse, depending on the instance, than our query plans. We discuss in the Conclusions how our algorithms can be adapted to work under GHD-like time guarantees as well.

\section{Model}

%\begin{figure*}[t]
%\includegraphics[width=\textwidth]{tgraph.pdf}
%\caption{A temporal graph $G$, its point-based representation $\hat{G}$, and two snapshots at times $t=1$ and $t=3$. Solid and dashed arrows are labeled with two different predicates: \textsf{follows} and \textsf{subscriber}, respectively.}
%\label{fig:tgraph}
%\end{figure*}

Here we describe temporal graphs, queries, and their AGM bound. 

\subsection{Temporal graphs and queries}
As is usual in the literature, we regard a temporal graph $G$ as a set of labeled edges $s \stackrel{p}{\rightarrow} o$, or triples $(s,p,o)$, 
that exist during a given time interval $[\ti,\tf\,)$ (see, e.g., \cite{tsparql,sparqlst}). 

\begin{definition}
Let $\U$ be \new{any} set of values and $(\T,\leq)$ a totally ordered {\em time domain}, disjoint from $\U$. 
A {\em temporal graph $G$} is a set of $N := |G|$ tuples $(s,p,o,[\ti,\tf\,))$, 
where $s,p,o$ are elements from $\U$, and 
$\ti,\tf$ are elements from $\T$, with 
$\ti < \tf$, so that $[\ti,\tf\,)$ represents an interval in the time domain.
We also define, respectively \\
$\U_G := \{s,p,o \mid \exists \ti,\tf,\, (s,p,o,[\ti,\tf\,)) \in G\}$ and \\
$\T_G := \{ \ti,\tf \mid \exists s,p,o,\, (s,p,o,[\ti,\tf\,)) \in G\}$ \\
as the set of elements and time instants mentioned in tuples of $G$. %\textcolor{red}{ojo, introduje $\U_G$ y $\T_G$}
\end{definition}

A tuple $(s,p,o,[\ti,\tf\,))$ represents that the triple $(s,p,o)$ is valid in any time $t$ such that $\ti \leq t < \tf$. We assume that temporal graphs do not have redundant tuples: if $G$ has 
tuples $(s,p,o,[\ti,\tf\,))$ and $(s,p,o,[\ti\,',\tf\,'))$, then \new{$[\ti,\tf)$ and $[\ti\,',\tf\,')$ must be disjoint}. 
This enforces a unique way to represent temporal graphs. Furthermore, since $\U \cap \T = \emptyset$, 
node identifiers and labels cannot be time instants.

%Figure~\ref{fig:tgraph} (left) shows an example temporal graph $G$ with $\U_G = \{1,\ldots,10\}$ and $\T_G = \{1,\ldots,4\}$. It describes the evolution of a social network where nodes $1$ to $7$ are persons and node $8$ is a group. %for legibility we name the labels $\mathsf{subscriber}=9$ and $\mathsf{follows}=10$. 
%The graph contains edges 
%%$x \stackrel{\mathsf{follows}}{\longrightarrow} y$ among persons 
%$\mathsf{follows}$
%(solid arrows) and edges 
%%$8 \stackrel{\mathsf{subscriber}}{\longrightarrow} x$ 
%%from a  group to persons 
%$\mathsf{subscriber}$ (dashed arrows). Arrows show their interval of validity $[\ti,\tf)$. Ignore for now the second graph, $\hat{G}$. To its right, we show the snapshots $G_1$ and $G_3$ of the valid edges at instants $1$ and $3$.

We focus on graph patterns, which are the foundation of most query languages used in industry, like SPARQL, CYPHER and GQL \cite{survey,GQL}. \new{Graph patterns are sets of tuples with constants and variables, which must be found in $G$.} To support temporal operations, we add temporal variables to patterns, and temporal filters. 

\begin{definition} \label{def:bgp-seg}
Let $\Vu$, $\Vt$ be two disjoint sets of variable symbols, additionally disjoint from $\U$ and $\T$. 
%In the segregated model, 
A {\em temporal basic graph pattern (tBGP)} is a set of tuple patterns $(x,y,z,w)$, where each $x$, $y$, and $z$ are in $\U_G \cup \Vu$, and $w$ is in $\T_G \cup \Vt$, plus optional {\em comparison clauses} of the form $w_1 \le w_2$, where $w_1,w_2 \in \T_G \cup \Vt$. 
\end{definition}

We note, in particular, that tBGPs can only mention time instants that exist in $\T_G$. This is not really a restriction since the edges of $G$ at a time $t \in \T$ are exactly the same as those at time $t'$, where $t'$ is the predecessor of $t$ in $\T_G$. Queries can then be translated into valid tBGPs: clauses $w_1 \le w_2$ stay valid after mapping values of $\T$ to their predecessors in $\T_G$. We also assume that those clauses are satisfiable, and do not form cycles (like $w_1 \le w_2 \le w_1$, as those are expressed by collapsing the variables in the cycle into one).

\begin{example}
\new{In Figure~\ref{fig:trips}, $\U_G = \{ \textsf{Anna}, \textsf{Bob}, \textsf{Carl}, \textsf{Dick}, \textsf{Emma}, \textsf{B6},\\ \textsf{B12}, \textsf{trip}\}$ and $\T_G {=} \{ \textsf{06:45}, \textsf{07:22}, \textsf{09:51}, \textsf{10:32}, \textsf{10:40}, \textsf{10:46}, \textsf{11:05}, \textsf{11:25}, \\ \textsf{16:11}, \textsf{17:03}, \textsf{18:30}, \textsf{19:45}, \textsf{18:56}, \textsf{19:10} \}$. We cannot write a pattern $(x,\textsf{trip},\textsf{B12},\textsf{10:50})$ because \textsf{10:50} is not in $\T_G$, but the answer is the same as for $(x,\textsf{trip},\textsf{B12},\textsf{10:46})$, where $x$ can be \textsf{Anna} or \textsf{Carl}. Only at \textsf{11:05} the answer changes: \textsf{Anna} is no longer a solution. \qed}

\end{example}
The semantics of tBGPs is given by assignments, called \emph{solutions}\new{: values given to the tBGP variables so that all tuples occur in $G$}. 

\begin{definition}
A {\em solution} to a tBGP with variables $\Vu$ and $\Vt$ % strictly speaking, we need to limit the variables to only those used in the tBGP; otherwise the number of solutions is infinite
over a temporal graph $G$ is an assignment \new{$f : \Vu \rightarrow
\U_G$ and $f : \Vt \rightarrow \T_G$.}
%, such that $f(u) = u$ if $u \in \U_G$ and $f(t) = t$ if $t \in \T_G$.
Furthermore, \new{assuming $f(c)=c$ for $c\in \U_G\cup\T_G$,} for each tuple pattern $(x,y,z,w)$ in the tBGP, there must exist a tuple
$(f(x),f(y),f(z),[\ti,\tf\,)) \in G$ where $\ti \le f(w) < \tf$. 
Finally, for each clause $w_1 \le w_2$ it must hold $f(w_1) \le f(w_2)$.
\end{definition}

Note that the time component of our solutions must also belong to $\T_G$. This is a form of compacting the (possibly infinite) set of all the solutions in $\T$: a solution with time $t \in \T_G$ stays valid for all the times $t''\! \in \T$ such that $t \le t''\! < t'$, where $t'$ follows $t$ in $\T_G$.

\begin{example}
\new{A tBGP for the first query of Example~\ref{ex:graph} is $Q = \{ (x,\textsf{trip},$ $\textsf{B12},t), (y,\textsf{trip},\textsf{B12},t) \}$, with $\Vu= \{x,y\}$ and $\Vt=\{t\}$. Two solutions are $\langle f(x)=\textsf{Anna},f(y)=\textsf{Bob},f(t)=\textsf{10:32}\rangle$ and $\langle f(x)=\textsf{Anna},f(y)=\textsf{Carl},f(t)=\textsf{10:46}\rangle$. The first is valid until \textsf{10:40} and the second until \textsf{11:05} (not inclusive). A tBGP for the last query is $Q= \{ (x,\textsf{trip},\textsf{B12},t_1), (x,\textsf{trip},\textsf{B6},t_2),$ $t_1 \le \textsf{12:00}, \textsf{12:00} \le t_2 \}$. A solution is $\langle f(x)=\textsf{Anna}, f(t_1)=\textsf{10:12}, f(t_2)=\textsf{18:30}\rangle$. \qed}    
\end{example}

%We investigate the problem of solving a tBGP, defined as follows.

\begin{definition}
Given a temporal graph $G$, the problem of {\em solving} a tBGP $Q$ is that of computing the set $Q(G)$ of all the solutions to $Q$.   
\end{definition}

% \marginpar{\tiny \old{Eliminé cosas aquí, las pueden ver comentadas}}
%Intuitively, each solution $f \in Q(G)$ represents an assignment that maps each tuple pattern $(x,y,z,w)$ in the tBGP to a triple $(f(x),f(y),f(z))$ in $G$ such that this triple is valid in time $f(w)$. 
%For example, 
%%for some $t \in \T_G$, the query $(x,y,z,t)$ returns all triples $(s,p,o)$ that are valid in time $t$. For 
%for $t_1,t_2$ in $\T_G$, query $(x,y,z,w), t_1 \le w, w \le t_2$ returns all triples $(s,p,o)$ that are valid in any time point between $t_1$ and $t_2$. 
%In this case, the answers are given as a set $(s,p,o,t'_1)$, $(s,p,o,t'_2)$, $\dots,$ $(s,p,o,t'_\ell)$, where $t'_1,\dots, t'_\ell \in \T_G$ are all time points between $t_1$ and $t_2$ in which $(s,p,o)$ exists. One can always pack the answers into time intervals if required. %by the user 

%In Figure~\ref{fig:tgraph}, a tBGP $Q = \{ (x,\mathsf{follows},z,w),$ $(y,\mathsf{follows},z,w),$ $(8,\mathsf{subscriber},z,w)\}$ retrieves pairs $(x,y)$ of followers of the same person $z$ that is subscribed to group $8$, all happening at the same time $w$. Two answers from $Q(G)$ are $\langle x=3, y=6, z=4, w=2\rangle$ and $\langle x=3,y=6,z=4,w=3\rangle$. In packed form, they can be written as $\langle x=3, y=6, z=4, w \in [2,4)\rangle$.

The selection of tBGPs as our language is justified as the key primitive of temporal relational calculus~\cite{croker1989completeness}, the algebraic building block of most temporal database languages (see e.g. \cite{tquel,tsparql,sparqlst}). Hence, tBPGs are a fundamental querying functionality of temporal database languages per BGPs for (non-temporal) graph query languages. We compare tBGPs to other query languages in Section~\ref{sec:expressive}.

\subsection{Temporal worst-case optimal algorithms}

The AGM bound refers to cardinalities of relations. To extend the concept to temporal graphs, we use the following notion.
%: the AGM bound of a query $Q$ over a graph is defined as the maximum size of the answer of the evaluation of $Q$ over any graph whose point-wise representation has the same number of tuples.  

\begin{definition} \label{def:Gpoint}
The {\em point-based representation} of $G$ is   
$$\Gpt{G} ~=~ \{(s,p,o,t) \mid \exists \ti,\tf,~(s,p,o,[\ti,\tf\,)) \in G,~t \in [\ti,\tf\,) \cap \T_G \}.$$
\end{definition}

\begin{figure}[t]
    \centering
    \includegraphics[width=0.35\textwidth]{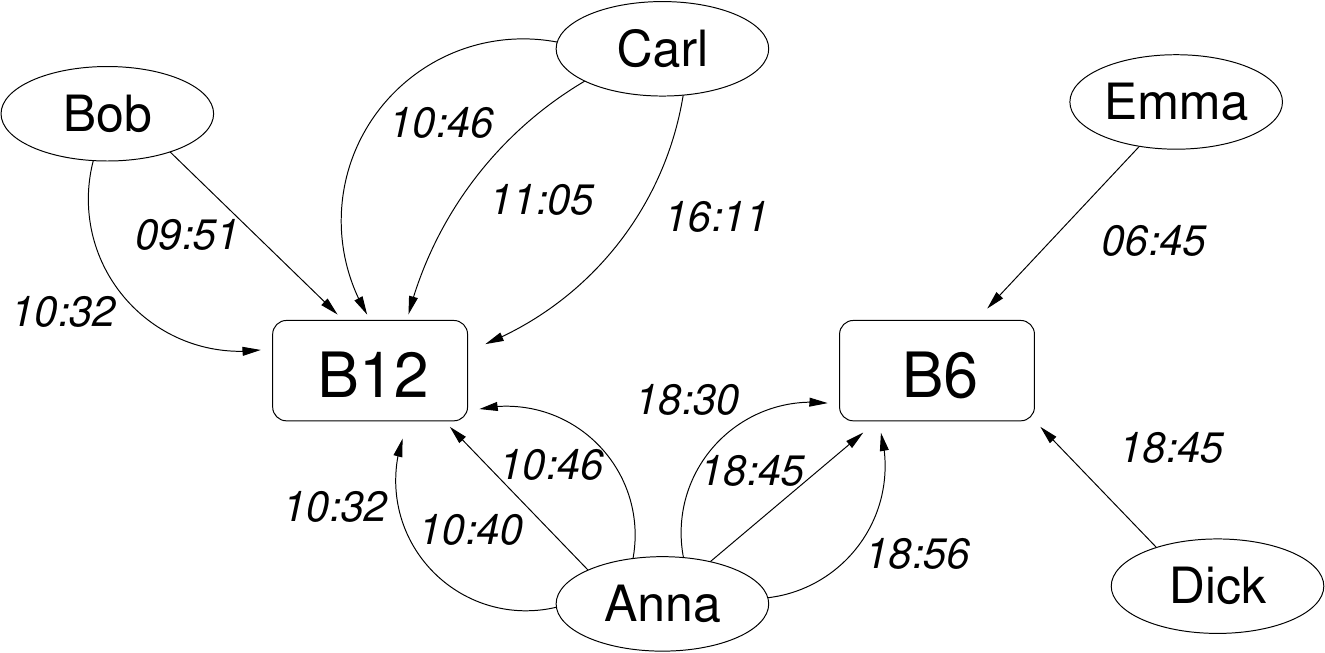}
    %\vspace*{-5pt}
    \caption{\new{Point-based representation of the graph of Figure~\ref{fig:trips}.}} %\ah{sugiero revisar 10:32 de Bob. Anna debería tener 19:09 para marcar fin?}\textcolor{red}{Arreglado Bob. Anna: no, 19:09 no existe, existe su predecesor, 18:56.}}}
    \label{fig:tripshat}
\end{figure}

The times $t$ present in $\Gpt{G}$ are only those constants $\ti, \tf \in \T$ that appear as interval extremes in $G$ (possibly on another triple); see Figure~\ref{fig:tripshat}. As explained, we need not consider other times in the domain where $G$ does not change. As a result, we have the bound $|\Gpt{G}| \le |G| \cdot |\T_G| \in O(N^2)$, which is tight when the edges are valid for long time intervals and many different time instants exist in $G$. %\textcolor{red}{Gonzalo: lo traje aquí, justo después de su def. Se necesita ya que pronto se habla de cuadrático vs lineal.}

\new{We define the AGM bound of query $Q$ over temporal graph $G$ as the maximum number of solutions of $Q$ over any temporal graph $G'$ whose point-wise representation has as many tuples as $\hat{G}$.}

\begin{definition}
The AGM bound $Q^*$ of a tBGP $Q$ for a temporal graph $G$ is $Q^* := \max \{ |Q(G')|, |\Gpt{G'}| {\le} |\Gpt{G}| \}$, over temporal graphs $G'$. An algorithm computing $Q(G)$ is {\em worst-case optimal \new{(wco)}} if it takes time $O(Q^*)$, possibly with data-agnostic and polylog factors of $|G|$.
\end{definition}

%we introduce the ``point graph'' $\Gpt{G}$ of $G$, which represents the same time validity intervals of $G$ with edges that refer to individual time {\em instants} (not ranges).

%\begin{definition} \label{def:Gpoint}
%Given a temporal graph $G$, we define 
%$$\Gpt{G} = \{ (s,p,o,t) \mid \exists \ti,\tf,~(s,p,o,[\ti,\tf\,]) \in G,~t \in [\ti,\tf\,] \cap \T \}.$$
%\end{definition}

%\textcolor{red}{Nosotros podemos lograr ser wco, creo, si hacemos join-first, porque no se nos cortan los intervalos. Pero tenemos que demostrar que LTJ en las otras variables y luego intersecciones para cada variable de tiempo, es wco en esta interpretación. Podremos decir que, sin ser wco, otros órdenes pueden darnos mejores tiempos aunque piquen el intervalo. Es una mala idea presentarlo así?}

The following theorem then states our main result. %, which holds for both models.

\begin{theorem} \label{thm:main}
Let $G$ be a temporal graph with $N$ tuples. Then, there exists a data
structure using $O(N)$ space that can compute $Q(G)$ for every tBGP $Q$ with $m$ tuples in \new{wco} time, $O(Q^* m \log N)$.
\end{theorem}

%\textcolor{red}{En términos de wco, vale la pena todo esto? No habrá siempre un grafo $G'$ que me obligue a picar todo el resultado? Podría incorporar optimalidad de otro modo, exigiendo que sea wco c/r a la parte no temporal, y óptima c/r a los intervalos? Vale la pena, dado que después me obligará a decir que sólo join-first es óptimo?}

The wco algorithm in Theorem~\ref{thm:main} can be easily obtained by running a wco algorithm directly on $\Gpt{G}$, interpreting the tuples of the tBGP $Q$ as quad-patterns
%and running any standard wco algorithm for the resulting BGP on $\Gpt{G}$ 
(let us disregard the time clauses for simplicity in this discussion). The disadvantage is, of course, that $\Gpt{G}$ is a bloated representation of $G$, which can require quadratic space. We could, alternatively, build $\Gpt{G}$ from $G$ and index it on the fly at query time, but this translates the $O(N^2)$ factor to the query time and working space. This is unacceptable in most practical cases. %  where the output is smaller than the database, + bien que N^2.

Our representation of Theorem~\ref{thm:main} actually {\em simulates} the use of the LTJ algorithm \cite{leapfrog} on $\Gpt{G}$, but uses $G$ in native form within $O(N)$ space, without converting it to $\Gpt{G}$. %, and therefore without any of the space or time penalizations we discussed.
Our algorithm is not only wco, but can also simulate LTJ with any desired variable elimination order (VEO), exactly as if run on $\Gpt{G}$. Choosing specific VEOs is key for practical performance even if all VEOs are wco in theory \cite{leapfrog}.

\subsection{Point \& interval queries: better guarantees}

We obtain even stronger bounds for point-in-time and point-in-interval queries, which ask for all answers of a BGP valid at a given time point $t$ or at some point in the interval $[\ti,\tf\,)$. We can also query for answers valid along a whole time interval. To state these results, we define some (non-temporal) labeled graphs, derived from slicing a temporal graph across the time domain. We use $|G|$ to denote the number of triples in a non-temporal graph. 

\begin{definition} \label{def:otherG}
For any $t \in \T_G$, $G_t$ is the set of triples
$(s,p,o)$ such that there exists a tuple $(s,p,o,[\ti,\tf\,))$ in $G$ with
$\ti \le t < \tf$. For any $t_1 < t_2$ in $\T_G$, $G\Xall{t_1}{t_2}$ is the
set of triples $(s,p,o)$ such that there exists a tuple $(s,p,o)$ in
$G_{t}$ for every $t \in \T_G$ with $t_1 \le t < t_2$, and \new{$G\Xany{t_1}{t_2}$} is the
set of triples $(s,p,o)$ such that there exists a tuple $(s,p,o)$ in
$G_{t}$ for some $t \in \T_G$ with $t_1 \le t < t_2$.
\end{definition}

\begin{example}
\new{Let $G$ be the graph of Figure~\ref{fig:trips}. Then $G_{\textsf{10:32}}$ contains the edges $(\textsf{Anna},\textsf{trip},\textsf{B12})$ and $(\textsf{Bob},\textsf{trip},\textsf{B12})$, $G\Xall{\textsf{10:32}}{\textsf{10:46}}$ contains only the edge $(\textsf{Anna},\textsf{trip},\textsf{B12})$, and $G\Xany{\textsf{10:32}}{\textsf{10:46}}$ contains the edges $(\textsf{Anna},\textsf{trip},\textsf{B12})$, $(\textsf{Bob},\textsf{trip},\textsf{B12})$, and $(\textsf{Carl},\textsf{trip},\textsf{B12})$. \qed}    
\end{example}

A point-in-time query amounts to solving classic BGP pattern matching on graph $G_t$, and a point-in-interval query requires matching BGPs on $G\Xany{t_1}{t_2}$, with $t,t_1,t_2$ given at query time. 
Our data structure can evaluate these queries in wco time with respect to the sizes of $G_t$ or $G\Xany{t_1}{t_2}$, i.e., the best one can hope for: the algorithm is wco on the labeled graph that has exactly the triples we need to consider. This improves the time of Theorem~\ref{thm:main} since all such graphs are smaller than $\hat G$. 
Concretely, we show the following in Section~\ref{sec:points-and-intervals}.

\begin{theorem} \label{thm:GtGanyt1t2-intro}
Let $G$ be a temporal graph with $N$ tuples. Then, there is a data
structure using $O(N)$ space that allows the following: 
\begin{itemize}
\item Answer BGPs $Q$ with $m$
tuples on the labeled graph $G_t$ for any time instant $t$ given with $Q$, in
$O(Q^*_t\, m \log N)$ time, where $Q^*_t$ is the maximum number of solutions for $Q$ in
some labeled graph with at most $|G_t|$ triples, and 
\item Answer BGPs $Q$ with $m$
tuples on the labeled graph $G\Xany{t_1}{t_2}$ for any time interval $[t_1,t_2)$ given with $Q$, in
$O(Q\Xany{t_1}{t_2*}\, m \log N)$ time, where $Q\Xany{t_1}{t_2*}$ is the maximum number of solutions for $Q$ in
some labeled graph with at most $|G\Xany{t_1}{t_2}|$ triples.
%\item Solve BGPs $Q$ with $m$
%tuples on the labeled graph $G\Xall{t_1}{t_2}$ for any time interval $[t_1,t_2)$ given with $Q$, in time
%$O(Q^*_{[t_1,t_2)}\, m \log N)$, where $Q^*_{[t_1,t_2)} = \min_{t_1 \le t < t_2} Q^*_t$.
\end{itemize}
\end{theorem}

In Section \ref{sec:points-and-intervals} we also show that our data structure leads to instance-optimal algorithms when matching a single triple pattern, and can also be used to \emph{list} all time intervals where these matches existed.

Next are duration queries, with two flavors. First, we can specifically input a BGP and time points $t_1,t_2$ at query time, and require all tuples that existed during the whole interval, or, in other words, to query 
graph $G\Xall{t_1}{t_2}$. We can also input a duration $\delta$ instead of the interval, in which case we want to look for answers that are valid for at least $\delta$ time. We show the following in Section~\ref{sec:duration-queries}.

\begin{theorem} \label{thm:Gallt1t2duration-intro}
Let $G$ be a temporal graph with $N$ tuples. Then, there is a data
structure using $O(N)$ space that allows the following: 
\begin{itemize}
\item Solve BGPs $Q$ with $m$
triples on the labeled graph $G\Xall{t_1}{t_2}$ for any time interval $[t_1,t_2)$ given with $Q$, in
$O(Q^*\Xall{t_1}{t_2}\, m \log N)$ time, where $Q^*\Xall{t_1}{t_2} = \min_{t_1 \le t < t_2} Q^*_t$.
\item Solve BGPs $Q$ with $m$
triples, reporting all the time intervals of length at least $\delta$ where $Q$ holds in $G$ for any duration $\delta$ given with $Q$, in time
$O(Q^*_{\delta}\, m \log N)$, where $Q^*_{\delta}$ is the maximum number of solutions for this query in the subgraph of $G$ formed by all tuples with duration $\delta$ or more. 
\end{itemize}
\end{theorem}

%We show how the same representation of Theorem~\ref{thm:main} allows solving classic BGPs on graphs $G_t$, $G\Xall{t_1}{t_2}$ and $G\Xany{t_1}{t_2}$ with $t$, $t_1$ and $t_2$ defined at query time. Variants of those solutions yield efficient algorithms for finding all the changes that occurred in $G$ between two time instants, or to run queries in wco time considering only edges with long enough durations. \textcolor{red}{más color?}

%\textcolor{red}{Vamos a hablar de GHD, o lo dejamos para las conclusiones? O para la comparación con Hu et al., que tal vez podría traerse acá.}

%\textcolor{magenta}{JUAN: esto me queda un poco volando, tiene mas que ver con la representación física que la lógica o no?}
%\textcolor{red}{Gonzalo: Sí, déjame pensarlo, tal vez hubo que decirlo antes para poder ahora hablar de cosas cuadráticas en $N$.}

\subsection{On the choice of temporal BGPs}
\label{sec:expressive}

Temporal BGPs play a role for temporal graphs analogous to conjunctive queries (CQs) for relational data: they form the  pattern-matching core underlying richer temporal query languages~\cite{croker1989completeness}.

The study of tBPGs is further justified by the fact that algorithms for standard conjunctive query evaluation cannot be directly translated to the temporal setting, because the language of tBGPs can be more expressive than CQs over a standard relational encoding of temporal graphs (see Appendix~\ref{app:expressive} in the extended version). \new{The appendix also shows that,} if $G$ is expanded into its point-based representation $\hat G$, then every tBGP corresponds to a CQ with inequalities evaluated over $\hat G$.
%, again, Appendix~\ref{app:expressive}
%(see the extended version).  
This further justifies tBGPs as the CQ analog in our setting. Since the size of $\hat G$, for a graph $G$ with $N$ tuples, can be of order $N\cdot |\T_G|$, materializing $\hat G$ to process tBGPs as relational CQs is infeasible. Notice that our index structures can correctly evaluate tBGPs over $\hat G$ using only $O(N)$ space.

\section{Overview of our solution}

Our solution is to build a data structure for temporal graphs in which we can run the Leapfrog TrieJoin algorithm with any  variable ordering to compute joins in wco time, and where the total space is linear with respect to the number of tuples of the temporal graph. 

For a temporal graph $G$, we replace the $\ti$ and $\tf$ components in $G$ by their rank in 
$\mathcal T_G$, turning them into integers in $[0,T)$ where $T=|\T_G|$ (we also maintain a dictionary to enable switching between actual time constants and integers in $[0,T)$). We further map the 
set of node identifiers and labels to the integer interval $[0,U)$, where $U = |\U_G|$. We handle those numeric identifiers using binary tries.

We now define the Leapfrog TrieJoin algorithm and (versioned) binary tries: the two building blocks of our solution. In Section \ref{sec:representation}, we construct a data structure from these two building blocks that achieves the desired wco optimality, but with $O(N\log N)$ space. Section \ref{sec:linear} explains how to further reduce this to a structure using linear space, and how to use this structure for querying. 

\subsection{Leapfrog Triejoin}

We describe the Leapfrog Triejoin (LTJ) algorithm \cite{leapfrog}, in the version that is adapted for solving BGPs on labeled graphs \cite{HRRSiswc19}. 

LTJ requires that the triples $(s,p,o)$ are represented as tries, which we will call {\em LTJ tries}, in the 6 possible orders of the components \textsc{s}, \textsc{p}, and \textsc{o}. Tries are labeled trees where no two children of a node have the same label, and represent all the strings that can be read by concatenating the labels of their root-to-leaf paths. 

Each LTJ trie stores the components \textsc{s,p,o} in some order, as strings of length 3. We call these orders \textsc{spo}, \textsc{sop}, \textsc{pos}, \textsc{pso}, \textsc{osp}, and \textsc{ops}. For example, the trie for the order \textsc{pos} has one root-to-leaf path with consecutive labels $p,o,s$ for each triple $(s,p,o)$ in the graph (see Figure~\ref{fig:postrie}).
Each LTJ trie has height 3 and exactly $N$ leaves.

\begin{figure}[t]
\includegraphics[width=0.40\textwidth]{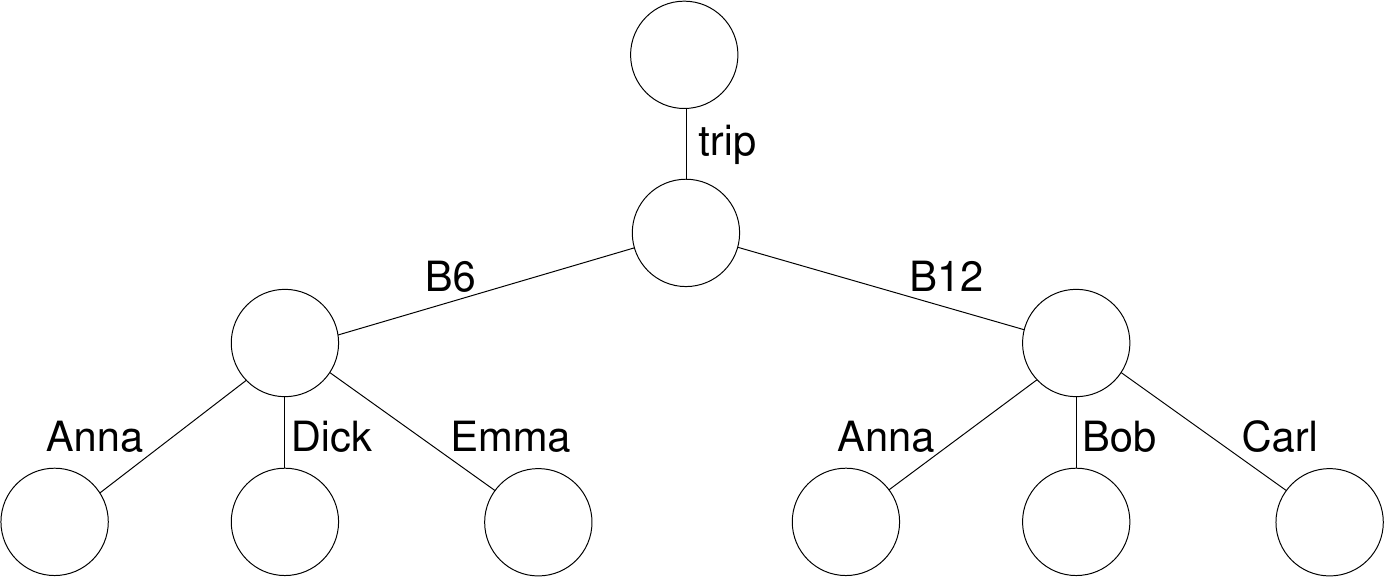}

\caption{\new{The LTJ trie of order \textsc{pos} for the graph of Figure~\ref{fig:trips} devoided of temporal annotations.}}
\label{fig:postrie}
\end{figure}

Let $Q = \{t_1, \ldots, t_m\}$ be a BGP (i.e., a set of triple patterns) and $\{x_1, \ldots, x_v\}$ its set of variables. LTJ carries out $v$ iterations, ``eliminating'' one variable at a time. The order LTJ chooses to eliminate the variables is known as the {\em variable elimination order (VEO)}. 

Each triple pattern $t_i$ is associated with one of six tries, say $\tau_i$.
A level $\ell$ of $\tau_i$ corresponds to a constant $c$ (or variable $x$) if $t_i$ contains $c$ (or $x$) at the position corresponding to level $\ell$ in $\tau_i$. To be a valid trie for $t_i$, the first levels of $\tau_i$ (i.e., levels closest to the root) must correspond to constants of $t_i$, and the subsequent levels must correspond to the variables of $t_i$, consistently with the VEO. 

\begin{example}
\new{For illustration, consider the graph of Figure~\ref{fig:trips} devoid of temporal annotations (and thus with only one edge from \textsf{Carl} to \textsf{B12}). If $t_i=(x,\textsf{trip},z)$ and the VEO eliminates $z$ and then $x$, then $\tau_i$ must be the trie of the order \textsc{pos} shown in Figure~\ref{fig:postrie} (this is why LTJ needs the tries in various orders). Its first level corresponds to \textsc{p} $=\mathsf{trip}$, the second to \textsc{o} $=z$, and the third to \textsc{s} $=x$. \qed} 
\end{example}

The first step of LTJ is to descend, from the root of each trie $\tau_i$, by the constants of $t_i$. It then starts the variable elimination step. Say that the chosen VEO is $x_1,\ldots,x_v$. LTJ then finds each value $c$ that appears as a child of the current node in every trie $\tau_i$ whose next level corresponds to variable $x_1$. For each such $c$, LTJ {\em binds} $x_1 := c$, descends by $c$ in all the corresponding tries $\tau_i$, and goes on recursively with $x_2$.  
 Once we have bound all variables in this way, the set of bindings for $x_1,\ldots,x_v$ is a new solution for $Q$. 

The values $c$ are found by intersecting the children of the current nodes $v_i$ in all the suitable tries $\tau_i$; see Algorithm~\ref{alg:leapfrog} \new{in Appendix~\ref{ap:pseudo} of the extended version. LTJ cycles over all the involved tries looking for the smallest $x' \ge x$, where $x$ is its next candidate, within the children of the current node $v_i$ in each trie $\tau_{i}$. When all tries find the same $x$, this is the next output of the intersection. The key primitive used, $\leap(v_i,x)$, finds that smallest $x'\geq x$ within the children of $v_i$.}
If leap is executed in $O(\log N)$ time (e.g., using binary search) then LTJ obtains the wco time complexity $O(Q^* m\log N)$.

\subsection{Binary Tries}
\label{sec:vbt}

%\new{In the integrated model, we will replace the values of $G$ that belong to the time domain by their rank in $\T$, assigning them integer values in $[0,T)$. The other values, in $\U \setminus T$, are assigned integer values in $[T,U)$. In the segregated model, we} \old{We}
%We will replace the $\ti$ and $\tf$ components in $G$ by their rank in 
%$\mathcal T_G$, turning them into integers in $[0,T)$ where $T=|\T_G|$. We will similarly map the 
%set of node identifiers and labels to the integer interval $[0,U)$, where $U = |\U_G|$. We will
%handle those numeric identifiers using binary tries.

Assume that we aim to represent each $G_t$ as a labeled graph, for the consecutive values $t \in [0,T)$. We now describe a technique inspired by versioned data structures, beginning with binary tries. 

%\subsection{Binary Tries}

A {\em binary trie (BT)} is a binary tree that represents a set of binary 
strings, each string being a root-to-leaf path in the trie. If the sequence of 
left-right turns from to root to the leaf are interpreted as $0$ and $1$, 
respectively, we obtain the string represented by the leaf. We assume no binary
string is a prefix of another. 

We use BTs to encode the sets of children of nodes of the LTJ tries. Those children are numbers in $[0,U)$ interpreted as binary 
strings of length $\ell=\lceil \log_2 U \rceil$, reading from most to least significant digit. Figure~\ref{fig:vbt} (left) illustrates the BT for \new{a toy LTJ trie node.}

%\begin{figure}[t]
%\begin{center}
%\includegraphics[width=0.45\textwidth]{bt.pdf}
%\end{center}
%\caption{On the left, an LTJ trie node $v$ with three children: values $2$, 
%$3$, and $6$. On the right, the BT for that set of children, 
%using $\ell=3$ bits. 
%The BT represents the binary strings $2 = 010_2$, $3 = 011_2$, and $6 = 110_2$.
%The root of the BT is $v$ and its leaves are the corresponding children of $v$
%in the LTJ trie.}
%\label{fig:bt}
%\end{figure}

\begin{figure}[t] %\begin{figure*}[t]
\begin{center}
\includegraphics[width=0.49\textwidth]{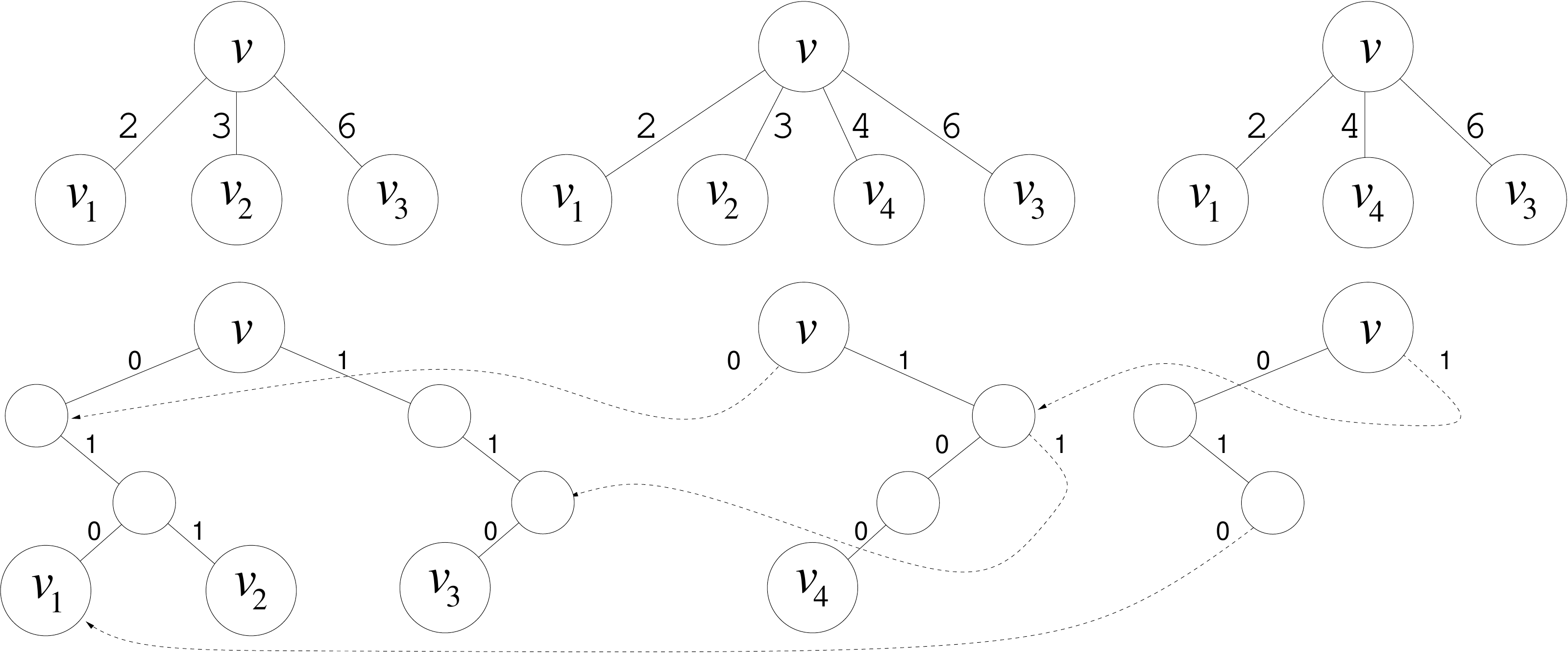}
\end{center}
%\vspace*{-5pt}
\caption{At the top, three consecutive versions of an LTJ trie node $v$. The leftmost one has three children: values $2$, $3$, and $6$. It is represented (below) with a BT using $\ell=3$ bits, which stores the binary strings $2 = 010_2$, $3 = 011_2$, and $6 = 110_2$. The root of the BT is $v$ and its leaves are the corresponding children of $v$ in the LTJ trie. The center and right trie are successive versions of the leftmost one: the value $4$ appears
in the second version, whereas $3$ disappears in the third. Both are represented (below) as VBTs: the pointers to preceding versions are shown with dashed lines; note in particular
the pointer from the third to the first version.}
\label{fig:vbt}
\end{figure} %\end{figure*}

BTs can simulate the search for the first child of $v$ with value $\ge x$ in
$O(\ell) = O(\log N)$ time, as follows. We start at the BT root, with height
$h := \ell$. If the $h$-th highest bit of $x$ (an $\ell$-bit number) is 
$1$, we continue recursively by the right child of the BT, decrementing
$h$. Otherwise, we proceed as follows. First, we try to find the answer on the left 
branch, also decrementing $h$. If that search returns an answer, we return it 
too. If not, we return the leftmost leaf of the right child of the node. When 
we arrive at a leaf, its string represents the first value $\ge x$ in the trie.
This takes $O(\ell)$ time because it can descend by both children of a node 
only once in the whole process. %Algorithm~\ref{alg:BTseek} in Appendix~\ref{ap:pseudo} shows the pseudocode.
For pseudocode see the extended version (Algorithm~\ref{alg:BTseek} in Appendix~\ref{ap:pseudo}).

BTs then support the LTJ leap operation within the same $O(\log N)$ time
factor penalty of binary searches. As a consequence, we can use a BT as a local
structure to represent the set of children of every LTJ trie node (instead of
just a plain array of increasing identifiers), and retain the same wco
time $O(Q^* \log N)$ on labeled graphs.

\subsection{Versioned Binary Tries}

We aim to represent all the graphs $G_t$ with the same tries.
When representing $G_t$, we will have LTJ trie nodes $v$ that will
be very similar to the node $v$ for $G_{t-1}$: the node $v$ in $G_t$ may have a few inserted or deleted children with respect to the node $v$ in $G_{t-1}$. To efficiently represent the children of those LTJ trie nodes $v$, we use {\em versioned binary tries (VBTs)}. A VBT represents a BT as a set of 
updates with respect to a {\em reference} BT. Instead of explicitly storing all
the nodes of the BT it represents, the VBT records only the root-to-leaf paths 
that change with respect to the reference BT. Changes refer to newly inserted or
deleted binary strings. The BT subtrees that do not change, instead of being 
duplicated, are pointed to from the corresponding nodes in the newly created paths
to the reference BT. Figure~\ref{fig:vbt} illustrates successive versions of the leftmost subtrie.

Note that a top-down traversal on a VBT is identical to that on a standard BT,
so we can run the LTJ intersections on the VBTs as well, in $O(\log N)$ time
per leap. Note how, in Figure~\ref{fig:vbt}, we can traverse 
the VBT of the third version exactly as if it were a BT. 

We use VBTs to represent a sequence of versions of a node's children.
The first version in time is represented as a standard BT. Each new 
version is a VBT encoded with respect to the preceding BT. From the third
VBT onwards, the preceding BT is also represented as a VBT, so pointing to 
a subtree of the preceding BT may actually correspond to pointing to an earlier
BT pointed to by the preceding VBT. 
This is shown on the third version in 
Figure~\ref{fig:vbt}.

Since the paths are of length $O(\log N)$, it follows that 
VBTs require $O(\log N)$ space per update they record with respect to a
previous BT. VBTs will be part of our solution, as described next.
% Too early: They then require in total $O(|G_0|\log N+T\log N)=O(N \log N)$ space. Within this space, we are representing $\Gpt{G}$ (i.e., every $G_t$ as if it were independent), without using $\Theta(N^2)$ space. In particular, this representation would suffice to solve BGPs on any $G_t$ in wco time, but we aim at more.

\section{Our Representation}
\label{sec:representation}

In abstract terms, we regard each tuple $(s,p,o,[\ti,\tf\,))$ of $G$ as if it were $\tf-\ti$ standard quads $(s,p,o,t)$, one per integer $t \in [\ti,\tf\,)$; recall that we have mapped $\T$ to $[0,T)$. This is $\Gpt{G}$, which as explained may have $\Theta(N^2)$ quads, but we will manage to represent them all within $O(N\log N)$ space (later, we will reduce the space to $O(N)$).

With the quads model, the tBGPs can be solved directly by using the LTJ 
algorithm on the LTJ tries storing the 4! permutations of $\{s,p,o,t\}$. The 
algorithm is then wco with respect to this model. In particular, consider a tBGP
where the time component is a single constant $t$ for all the triple patterns.
Then, using an LTJ trie with an order that starts with $t$, we descend by the
child $t$ of the root and can solve the tBGP in time $O(Q^*_t \log N)$, where
$Q^*_t$ is the AGM bound for the corresponding BGP on the graph $G_t$. As another example, consider the same tBGP where now $t$ is a variable. By solving
it on an LTJ trie that starts with the $t$ component, we can mimic the so-called
``time-first'' approach in previous work \cite{HSGAY22}; by using a 
trie that ends with the $t$ component, we mimic the so-called ``join-first''
strategy. We can use more general strategies, however, by putting the time
component elsewhere in the order.

In this scenario, the clauses $w_1 \le w_2$ are handled as follows. If $w_1$ is bound before $w_2$, then it will have assigned a value $f(w_1) = t_1$ each time $w_2$ is bound. At this moment we restrict the LTJ intersections for $w_2$ so that the values stay within $[t_1,T)$. If, instead, $w_2$ is bound to $f(w_2)=t_2$ before $w_1$, we enforce that the values of $w_1$ be within $[0,t_2]$. 
This is still wco, as it is equivalent to materializing a table with all the pairs $(w_1,w_2)$ with $w_1 \le w_2$ and including it in the LTJ algorithm. Since all the time instants appear in both columns, LTJ will not restrict the first bound variable; the second one will be restricted exactly as described.

\begin{figure*}[t]
\includegraphics[width=\textwidth]{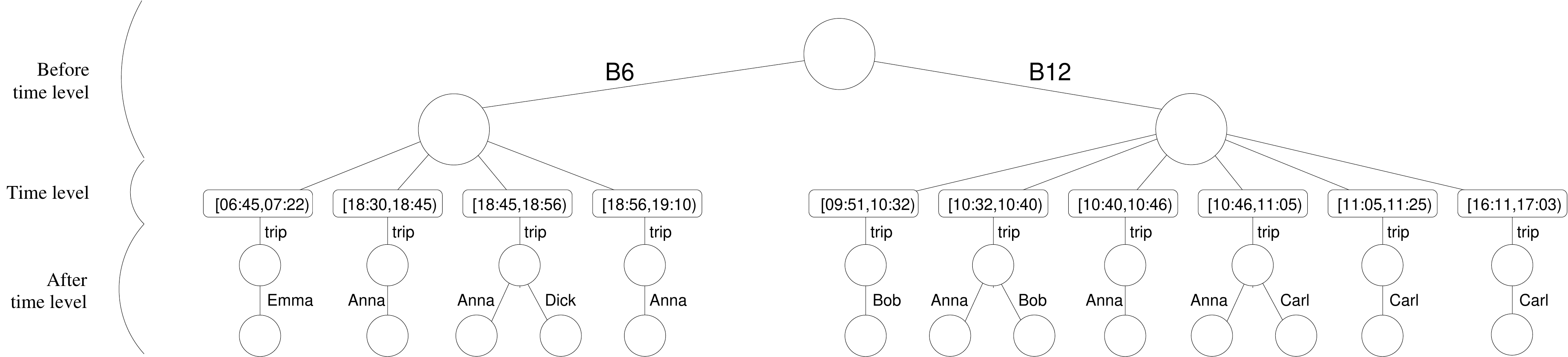}
%\vspace*{-5pt}
\caption{\new{The \textsc{otps} trie corresponding to the graph of Figure~\ref{fig:trips}.}}
\label{fig:otpstrie}
\end{figure*}

To do all this within $O(N\log N)$ space, 
we introduce a special representation for the LTJ tries. Their abstract
form will still be a trie of height 4, with one level per tuple component.
The difference is in the way we represent the children of the nodes.
Given a particular permutation, say $(s,t,p,o)$, we distinguish the levels
before the time component (in our example, the first level, $s$), the
time level, and the levels after the time component (levels 3 and 4, 
$p$ and $o$, in our example). Note that which levels are before and after the time
level depends on the permutation we are storing.

The levels before the time level can be implemented in traditional form, 
e.g., as an array of increasing child values. We now describe our
implementation of the time level and its subsequent levels.

\subsection{The time level} \label{sec:timelevel}

The time level will be represented in a form closer to the tuples of $G$ (with the times already converted to integers in $[0,T)$) than to the set of quads of $\Gpt{G}$. This 
level will be an ordered sequence of disjoint time intervals, stored 
in classical form (e.g., an array). Let $v$ be an LTJ trie node whose children belong to the time level, and $[\ti_1,\tf_1), [\ti_2,\tf_2), \ldots$ 
be the time intervals of all the tuples stored in the LTJ subtrie of $v$ (this is a subset of all the time intervals in $G$ and they can overlap). Then consider the 
list $\langle t_1,t_2,\ldots \rangle$ containing all the values 
$\ti_k$ and $\tf_k$, sorted in increasing order and with equal values 
removed. The intervals forming the children of $v$ are then in principle
$[t_1,t_2), [t_2,t_3),\ldots$, which covers the whole universe $[0,T)$.
We remove, however, intervals with no descending tuples.
The surviving intervals are then of the form $[\ts_1,\te_1), [\ts_2,\te_2), 
\ldots$, where each $[\ts_i,\te_i)$ is equal to some $[t_k,t_{k+1})$.

\begin{example} \label{ex:otps}
\new{Consider the \textsc{otps} trie for Figure~\ref{fig:trips}. The intervals $[\ti_i,\tf_i)$ for the node $\textsf{B12}$ are $\{\textsf{[9:51,10:40),[10:32,11:05),[10:46,11:25),}$ $\textsf{[16:11,17:03)}\}$. The corresponding sorted list of times $t_i$ is then $\langle \textsf{9:51,10:32,10:40,10:46,11:05,11:25,16:11,17:03}\rangle$. A consecutive pair forms an interval $[\ts_i,\te_i)$, except for $[\textsf{11:25,}$ $\textsf{16:11})$, in which there are no tuples. Figure~\ref{fig:otpstrie} illustrates the different levels in this trie. \qed}
\end{example}

% this is in section The Intersections
%To run $\leap(v,x)$ on the time level, we use binary search on the values 
%$\ts_i$. After finding $i$ such that $\ts_i \le x < t_{i+1}$, we return $x$ if
%$x < \te_i$, otherwise we return $t_{i+1}$.

\subsection{Levels after the time component}

We will use BTs and VBTs to represent the children
of all LTJ trie nodes below the time level. Consider an LTJ trie node $v$
whose children, at the time level, are $[\ts_1,\te_1), [\ts_2,\te_2), \ldots$ 
We will {\em concatenate} all the binary strings corresponding to all the 
levels that follow the time level. \new{In the \textsc{otps} trie of Figure~\ref{fig:otpstrie}, the descendants of each node $[\ts_i,\te_i)$ will be
arranged in a single BT (or VBT) holding binary strings of length $2\ell$, formed by
concatenating the $p$ and the $s$ values (i.e., $p:s := 2^\ell p+s$)} of all the
tuples that must be stored below $v$ and exist at time $\ts_i$. Per our
construction, this set of tuples does not vary within the time interval 
$[\ts_i,\te_i)$. 

%Note that every tuple in $G$ produces at most $2$ entries below some node $v$ in the time level of each LTJ trie, so the total number of nodes in this level is at most $2N$ per LTJ trie. 

Figure~\ref{fig:trie} \new{shows hypothetical values $p:s$ descending from some node $v$.} Note that concatenating the components is almost immaterial with respect
to using different BTs for the children of every LTJ trie node. The BT node representing each prefix $p$ (at depth $\ell$) becomes the root
of the subtrie representing each component $s$ of the concatenations $p:s$.
We can then interpret the node of $p$ as the LTJ trie node obtained by 
descending by $p$ from $v$, and the subtrie with the $s$ components as the BT 
of its children. Our arrangement, however, is more convenient for the 
versioning we require next.

\begin{figure*}[t]
\begin{center}
\includegraphics[width=\textwidth]{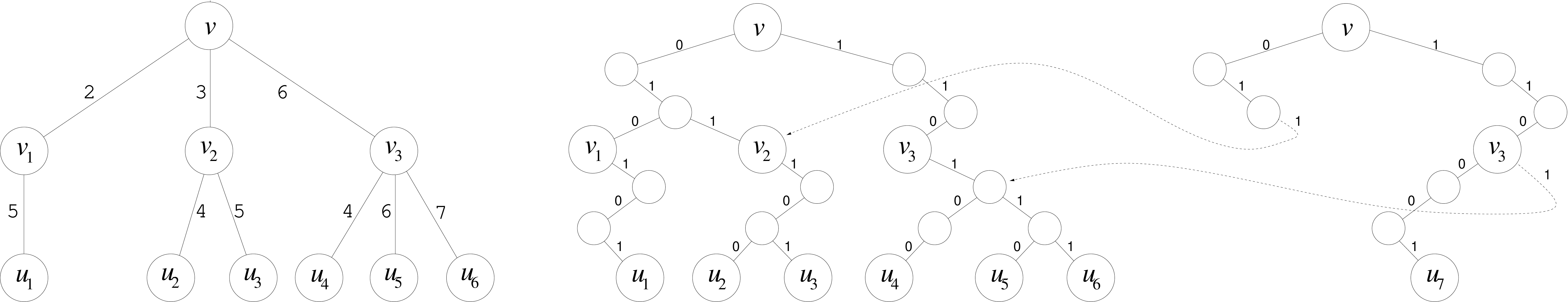}
\end{center}

\caption{On the left, the two final levels of an LTJ trie node rooted at $v$.
In the middle, the BT representation of the corresponding children. Note that
we can also regard all the BTs as a single BT on the concatenation of the path
labels, $2:5$, $3:4$, $3:5$, $6:4$, $6:6$, and $6:7$. On the right, a VBT
representing two edits: the removal of $2:5$ and the insertion of $6:1$.}
\label{fig:trie}
\end{figure*}

Let us consider how we store the elements formed by concatenating the remaining
attributes of the tuples below $v$. We first build the basic BT for the node 
$[\ts_1,\te_1)$, containing all the elements whose tuples exist in time $\ts_1$.
This set of tuples does not vary until time $\te_1$. This BT is the child of 
the node $[\ts_1,\te_1)$ of the time level.

Now consider the next interval, $[\ts_2,\te_2)$. If $\ts_2 > \te_1+1$, then 
the intervals spanning $[\te_1,\ts_2)$ were empty, and we simply build
a new BT for $[\ts_2,\te_2)$ containing only the new elements that appear at 
time $\ts_2$. Otherwise, there exists a nonempty BT that precedes the BT for 
$[\ts_2,\te_2)$, so we represent this BT as a VBT, whose reference is the BT 
for $[\ts_1,\te_1)$. This VBT must represent two events: 
\begin{enumerate}
\item Consider a tuple $(s,p,o,[\ti,\tf\,))$ of $G$ where $(s,p,o)$ must be 
inserted below the LTJ trie node $v$ and such that $\ti = \ts_2$. This tuple 
must be added to the VBT.
\item Similarly, each tuple $(s,p,o,[\ti,\tf\,))$ belonging to the subtrie of
$v$ and such that $\tf = \ts_2$ must be removed from the VBT. 
\end{enumerate}

\begin{example}
\new{In Figure~\ref{fig:otpstrie}, below node \textsf{B6}, the interval \textsf{[07:22,18:30)]} disappears as it contains no tuples, and thus we start a new BT with root \textsf{[18:30,18:45)} containing \textsf{Anna}. The next interval, \textsf{[18:45,18:56)}, modifies the previous one by adding \textsf{Dick}, represented as a VBT with respect to the previous one. The last interval, \textsf{[18:56,19:10)}, removes \textsf{Dick} again, and is again represented by a VBT. \qed}
\end{example}

We create new paths in the VBT of $[\ts_2,\te_2)$ as described in 
Section~\ref{sec:vbt}, to account for those events. Each event requires
creating $O(\log N)$ nodes in the BTs or the VBTs.
Note that each graph tuple $(s,p,o,[\ti,\tf\,))$ generates two
events: the insertion of $(s,p,o)$ at $\ti$ and its deletion at $\tf$. 
Therefore, the trie represents at most $2N$ events, each of which induces
$O(\log N)$ new nodes in the VBTs. The total space is thus $O(N\log N)$.

\begin{example}
\new{In the trie \textsc{otps} of Example~\ref{ex:otps}, the tuple $(\textsf{Anna},\textsf{trip},$ $\textsf{B12},\textsf{[10:32,11:05)})$ will generate, below the node \textsf{B12}, a time $\ti=\textsf{10:32}$ where we will insert $\textsf{trip:Anna}$, and a time $\tf=\textsf{11:05}$ where we will delete $\textsf{trip:Anna}$. Because we use VBTs for intermediate intervals (in this case, for the intervals \textsf{[10:40,10:46)} and \textsf{[10:46,11:05)}), $\textsf{trip:Anna}$ will not be explicitly represented in those. \qed}
\end{example}

\subsection{The intersections} \label{sec:inters}

Intersections work exactly as in LTJ over BTs, except at the time level. At this level, each node represents an interval, which must be handled during intersection: each interval $[\ts_i,\te_i)$ stands for all the time instants $\ts_i \le t < \te_i$. We modify the list intersection algorithm of LTJ so as to assume that all those
time instants $t$ are explicitly represented and that copies of the BT of 
$[\ts_i,\te_i)$ descend from all those implicit time instants $t$. To find the 
first time $t^* \ge t_0$, we look for $t_0$ and, (1) if we find 
some $\ts_i \le t_0 < \te_i$, we answer $t^* := t_0$; (2) if we find some 
$\te_i \le t_0 < \ts_{i+1}$, we answer $t^* := \ts_{i+1}$.

We can avoid recomputing the same answers for all the 
time instants $t$ of an interval: Rather than returning each time instant, we 
return the range $[t_0, \te_i)$ in case (1) above, and $[\ts_{i+1}, \te_{i+1})$
in case (2). We then modify the intersection algorithm to record 
the maximum interval starting at $t^*$ that is included in all the intersected
lists. The instantiation of the corresponding time 
variable is then a whole interval where the instantiated BGP will not change.
%Algorithm~\ref{alg:time-leapfrog} in Appendix~\ref{ap:pseudo} shows pseudocode.
See pseudocode in the extended version (Algorithm~\ref{alg:time-leapfrog} in Appendix~\ref{ap:pseudo}).

Thus we can return solutions to BGPs with ranges of values for the
time-bound variables. This is a form of compacting the output, but we cannot ensure we return it in the optimally compacted form.

%\section{BGPs with Minimum Duration}
%
%One extension to the query language is to set a minimum duration for the solutions that we report. This can be expressed as tuple patterns of the form $(x,y,z,[w,w+d))$ for a nonnegative constant $d$, meaning that in every solution $(s,p,o,t)$, the tuple must exist at least during the interval $[t,t+d)$.

%Pensé que la podía resolver, pero ahora no estoy seguro. Creo que sólo funciona con t en el último nivel. Ponemos los [\ts_i,\tf_i] como puntos (x,y) = (\ts_i,\tf_i-\ts_i) en una grilla 2D y para buscar el siguiente \ge t_0, buscamos el siguiente punto en x en [\t0,+\infty] \times [d,\infty]. Pero si t está en los niveles anteriores, pueden particionar esos intervalos por otros eventos.

\section{Linear Space} \label{sec:linear}

The base version we have described implements $\leap$ in $O(\log N)$ time and $O(N\log N)$ space. This space is worrisome for large graphs. We now describe a more sophisticated storage mechanism that achieves linear space with no penalty in time complexity. 

\subsection{Data Structure}

Consider the sequence of updates that occur in the time-level children $[\ts_1,\te_1), [\ts_2,\te_2),\ldots$ of a particular LTJ trie node. Instead of creating a sequence of VBTs $v_1,v_2,\ldots$, we create a {\em single} BT $V$ where we insert all the paths corresponding to those updates; each VBT (or BT) $v_l$ corresponds to a sequence of insertions/deletions of tuples (just insertions in case of BTs). If there are $b$ levels after the time level, for $1 \le b \le 3$, each such path is of length $b\ell$. Updates consisting of insertions correspond to a path of length $b\ell$. Updates consisting of deletions are also paths of length $b\ell$, plus a deletion mark. Note this differs from the way we represented deletions in VBTs; see Figure~\ref{fig:trie}. If a given tuple is inserted and deleted several times, its leaf will correspond to several updates. 

We will record the timestamps of those updates. If $V$ represents $L$ updates, they are numbered $1$ to $L$. Each node $[\ts_l,\te_l)$ stores the value $p_l \in [1,L]$ for its last update. This means that its BT corresponds to executing the updates $[1,p_l]$ on an empty trie. We store the timestamps, however, in a way that will enable fast navigation.

Each node $v$ of $V$ will conceptually store the subsequence $v.T$ of $[1,L]$ corresponding to the updates that occur below $v$. The subsequence is $v.T = \langle 1,2,\ldots,L\rangle$ if $v$ is the root of $V$. If $v$ is a leaf, then $v.T$ contains the timestamps where its particular tuple was created or removed. Figure~\ref{fig:wtree} shows an example (see only the top and bottom-left parts for now).

\begin{figure*}[t]
\begin{center}
\includegraphics[width=\textwidth]{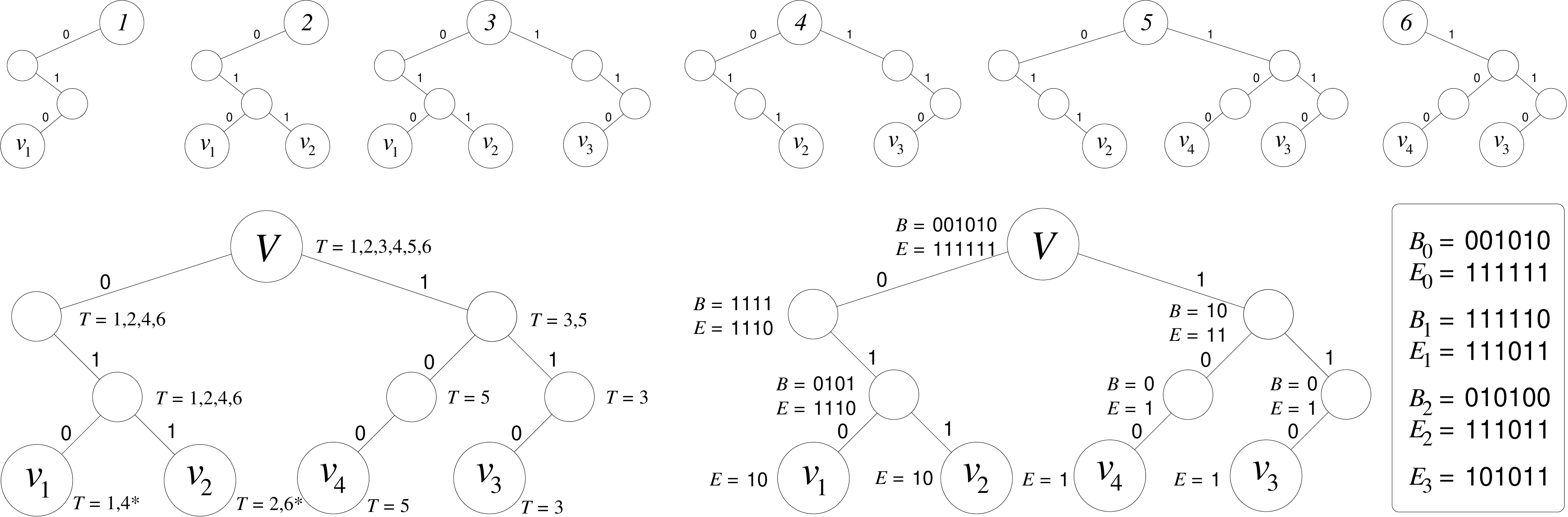}
\end{center}
\caption{On the top, the BTs of six consecutive timestamps, corresponding to inserting $v_1$ with value $010_2=2$, then $v_2$ with value $011_2 = 3$, and $v_3$ with value $110_2=6$. The leaf $v_1$ is then removed, $v_4$ is inserted with value $100_2=4$, and finally $v_2$ is removed. On the bottom left, our conceptual trie $V$ representing all the events, with the $v.T$ sequences of timestamps besides each node $v$; the deletions are marked with a $*$ at the leaves. On the right, the actual compact representation of $V$, first still with the tree topology and then just like the bitvectors $B_d$ and $E_d$. Bitvector $E_3$ can be made virtual, as explained.}
\label{fig:wtree}
\end{figure*}

Instead of storing $v.T[1,L_v]$ explicitly, we will store only a bitvector $v.B[1,L_v]$, where $v.B[i]=0$ iff $v.T[i]$ appears below the left child of $v$, and $v.B[i]=1$ if it appears below the right child. To efficiently support operation $\leap$, we also store a bitvector $v.E[1,L_v]$, where $v.E[i]=1$ iff at timestamp $v.T[i]$ there exists some tuple below $v$ whose last update was an insertion (i.e., the tuple exists at timestamp $i$). The leaves of $V$ store only their bitvector $v.E$, not $v.B$. 

\begin{example}
\new{At the bottom of Figure~\ref{fig:wtree}, the root's left child has $E[4]=0$ since the node does not exist at timestamp $T[4]=6$. Its right child has $B=0101$ since timestamps $T[1]=1$ and $T[3]=4$ go to its left child, whereas $T[2]=2$ and $T[4]=6$ go to the right.} \qed
\end{example}

Note that every update induces $2b\ell+1$ bits in $V$: one in the bitvectors $v.B$ of the proper ancestors $v$ of its leaf and one in the bitvectors $v.E$ of those nodes and the leaf itself. We can actually avoid storing $v.E$ at leaves $v$ because it is always an alternating sequence of $1$s and $0$s. Therefore, the total number of bits in $v.B$ and $v.E$ bitvectors of $V$ is $L\cdot 2b\ell = O(L\log N)$. The machine word must hold $\Omega(\log N)$ bits if it can address $\Theta(N)$ tuples in constant time; therefore this number of bits amounts to $O(L)$ words of space. Since the sum of the lengths $L$ of all the nodes below the time level amounts to the $2N$ updates in the graph, those $O(L)$ words amount to $O(N)$ over the whole LTJ trie.

We still have $O(L\log N)$ space to store the topology of $V$, however, as it stores $L$ paths of length $b\ell = O(\log N)$. We avoid the need to store those pointers by concatenating all the bitvectors $v.B$ (and, similarly, $v.E$) levelwise, left to right. Note that the total length of the bitvectors $v.B$ or $v.E$ in a given level is always $L$, so we have $b\ell$ bitvectors $B_d[1,L]$, for $d \in [0,b\ell-1]$ that can be stored as a large concatenated bitvector of length $Lb\ell$ (similarly, $E_d$ for $d \in [0,b\ell]$, where the level $d=b\ell$ can be made implicit, as explained). Those bitvectors are shown at the bottom-right of Figure~\ref{fig:wtree}.

\subsection{Navigation} % \textcolor{red}{al apéndice, o una parte?}} 
\label{sec:navig}

To navigate $V$ without pointers, we must know the ranges $B_d[s_v,e_v]$ and $E_d[s_v,e_v]$ where the bitvectors $B_v[1,L_v]$ and $E_v[1,L_v]$ are stored, for any node $v$ of depth $d$. For the root, these are just $B_0[1,L]$ and $E_0[1,L]$. For the rest, we make use of the function
$rank(B,i,j)$, which counts the number of $1$s in $B[i,j]$. This operation can be computed in constant time by storing just $o(|B|)$ additional bits~\cite{Cla96,Mun96}. Assume $B_d[s_v,e_v]$ is the area of $B_d$ corresponding to $v.B$. Then, letting $r := rank(B_d,s_v,e_v)$, the area for the left child of $v.B$
is $B_{d+1}[s_v,e_v-r]$, and the area for its right child is $B_{d+1}[e_v-r+1,e_v]$. We know we have reached a null pointer $v$ in $V$ when $s_v > e_v$. 

\begin{example}
\new{Let us start at the root in Figure~\ref{fig:wtree}, with $B_0[1,6]=001010$. To go to its left child, we compute $rank(B_0,1,6)=2$; thus its left child corresponds to $B_1[1,4]=1111$. If we want to go left again, we compute $rank(B_1,1,4)=4$, and find that its left child is null because it corresponds to $B_2[1,0]$.} \qed
\end{example}

Until now, our structure is a variant of the wavelet matrix \cite{CNO15}, but the bitvectors $E_d$ extend it so that we can traverse $V$ as if it were the VBT at some timestamp $p_l \in [1,L]$. We start at the root of $V$ (i.e., $d:=0$ and $[s_v,e_v] := [1,L]$) with the local offset $p:=p_l-s_v$ for $p_l$. If $E_d[s_v+p]=0$, then the VBT node was null at this time instant. Otherwise, we can go left or right. The timestamp $p$ becomes $p-rank(B_d,s_v,s_v+p)$ on the left child and $rank(B_d,s_v,s_v+p)-1$ on the right child. Note that $p$ cannot become negative, since that would mean that there have been no updates on the subtree of $V$ since the beginning, and thus the node would be null at timestamp $p$ (which we detect before entering the node).

\begin{example}
\new{Assume we want to go to the left child of the root at timestamp $p_l=6$ in Figure~\ref{fig:wtree}. Though $V$ has such a node, which is $B_1[1,4]$, we see that
it did not exist at timestamp $6$ since $E_1[1+3]=0$. This position $3$ is where we map the relative value $p=p_l-1=5$ at the root to relative value $p-rank(B_0,1,6) = 3$ at its left child.} \qed
\end{example}

With those operations, we can traverse any VBT top-down in constant time per movement to children, while the total space used by $V$ is $O(L)$ and the total space is $O(N)$ per LTJ trie. 

To support $\leap$ on this representation, recall that $\leap$ is applied on an LTJ trie node $[\ts_l,\te_l)$ of the time level, and aims to traverse its VBT (or BT) from its root $v$. 
Previously we describe this 
%It was described in Algorithm~\ref{alg:BTseek} 
on a VBT of height $\ell$, as we always perform $\leap$ on a single attribute, but now we have concatenated $b$ attributes in VBTs of height $b\ell$. 
However, since we can start at any node whose depth is a multiple of $\ell$, the procedure remains unchanged over blocks of $\ell$ elements. 
%We now describe how to carry out the operation $\leap$ on this representation. In Algorithm~\ref{alg:BTseek}, $\leap$ is applied on an LTJ trie node $[\ts_l,\te_l)$ of the time level, and aims to traverse its VBT (or BT) from its root $v$. It was described in Algorithm~\ref{alg:BTseek} on a VBT of height $\ell$, as we always perform $\leap$ on a single attribute. We have now concatenated $b$ attributes in VBTs of height $b\ell$. 
%However, because we can start at any node whose depth is a multiple of $\ell$, the procedure remains unchanged over blocks of $\ell$ elements. 
%
When run on $V$, we are only interested in the timestamps $[1,p_l]$, which represent the VBT of $[\ts_l,\te_l)$. We use $v.E$ as before to determine if a node is null at timestamp $p_l$. %See the pseudocode in Algorithm~\ref{alg:wtree} of Appendix~\ref{ap:pseudo}. 
See pseudocode in the extended version (Algorithm~\ref{alg:wtree}).
%
%The resulting pseudocode is given in Algorithm~\ref{alg:wtree} of Appendix~\ref{ap:pseudo}. 
This completes the proof of Theorem~\ref{thm:main}.

\section{Querying for points and intervals} \label{sec:points-and-intervals}

Definition~\ref{def:otherG} introduces several labeled graphs derived from a temporal graph $G$, in which it is of interest to answer standard BGPs. We now show how our linear-space data structure on $G$ can be used to answer BGPs over these graphs. We start with point-in-time and point-in-interval queries, corresponding to querying graphs $G_t$ and $G\Xany{t_1}{t_2}$, where $t$, $t_1$, and $t_2$ are given at  query time 
(Theorem \ref{thm:GtGanyt1t2-intro}). 
%Next, we show even stronger results when considering single triple patterns queries. 
We finish with $G\Xall{t_1}{t_2}$ and queries with duration. (Theorem \ref{thm:Gallt1t2duration-intro}).

\subsection{Point-in queries}

To answer a query $Q$ over $G_t$, we transform it into a temporal BGP by adding the constant $t$ as the fourth component in all the triples. Interestingly, if we force the variable order to descend first by the attribute $\textsc{t}=t$ we can show that the resulting algorithm takes exactly the same amount of steps that the normal LTJ algorithm on $Q$ would take, immediately giving us worst-case optimality over $G_t$. This proves the first part of Thm.~\ref{thm:GtGanyt1t2-intro}.

%\marginpar{\tiny \textcolor{red}{No se entenderá si mandamos la sec 6.2 al apéndice}}
For the case of $G\Xany{t_1}{t_2}$, we must extend the navigation of $G_t$ using VBTs, described in Section~\ref{sec:navig}, to intervals $[t_1,t_2)$. The key idea is that if $t_1$ and $t_2$ are represented by $p_1$ and $p_2+1$ at some node $[s_v,e_v]$, then $rank(E_d,s_v+p_1,s_v+p_2)=0$ iff the node $s_v$ did not exist along the whole period between the local offsets $p_1$ and $p_2$, that is, it did not exist in $G\Xany{t_1}{t_2}$.
Hence, we can quickly find the time intervals that overlap or are contained in $[t_1,t_2)$, and continue the LTJ algorithm while keeping track of these intervals. 
This yields the second part of Thm.~\ref{thm:GtGanyt1t2-intro}; 
details in 
the extended version (Appendix~\ref{ap:point-in-queries}).

We note that the so-called ``join-first'' strategy, which solves $Q$ on the labeled graph $G_\mathrm{all} = G\Xany{1}{T}$ of all the tuples that ever existed and then filters the results by time, can offer the AGM bound only on $|G_\mathrm{all}|$, whereas our result bounded by $|G\Xany{t_1}{t_2}|$ is the best one can hope for: the algorithm is wco on the labeled graph that has exactly the triples we want to consider. 

Given a single triple pattern $(s,p,o)$, our algorithm lists all matches that existed during the interval  $[t_1,t_2)$ in instance-optimal time: $O(\log N)$ per reported triple. We can also \emph{list} all the maximal intervals where each such match existed during $[t_1,t_2)$, each in $O(\log N)$ time. This also permits tracking the differences between $G_{t_1-1}$ and $G_{t_2-1}$: an important operation on versioned graphs. 
Neither ``join-first'' nor ``time-first'' strategies can handle these queries near-optimally. 
Full proof in the extended version (Appendix~\ref{ap:point-in-queries}). 
%For space reasons we leave details to Appendix~\ref{ap:point-in-queries}. 

\begin{theorem}
Let $G$ be a temporal graph with $N$ tuples. Then, there is a data
structure using $O(N)$ space that can report all the $occ$ occurrences of a single-triple-pattern query $Q=\{(s,p,o)\}$ in the labeled graph $G\Xany{t_1}{t_2}$, for any time interval $[t_1,t_2)$ given with $Q$, in time $O((1+occ)\log N)$.
It can also list each maximal time interval where each occurrence appears in time $O(\log N)$.
\end{theorem}

\subsection{Queries with duration} \label{sec:duration-queries}

In some cases we are interested in requiring that the solutions to the queries last for some time. We consider two cases of such queries. The first establishes a time interval and looks for solutions that always hold during this interval. According to our definitions, this amounts to running the query over $G\Xall{t_1}{t_2}$; recall Definition~\ref{def:otherG}. A second case does not fix the exact times, but sets a minimum duration $\delta$ along which the reported solutions must hold \cite{HSGAY22}.

In order to query $G\Xall{t_1}{t_2}$ in optimal time, we can reuse some of the ideas used to query $G\Xany{t_1}{t_2}$. Indeed, we have that 
%\marginpar{\tiny \textcolor{red}{Tb necesita algo de la sec 6.2 para entenderse}} 
$rank(E_d,s_v+p_1,s_v+p_2)=p_2-p_1+1$ iff the node $s_v$ exists throughout the whole period between the local offsets $p_1$ and $p_2$. Therefore, we can use Algorithm~\ref{alg:wtreeRange} with a single change to require that all the bits in the area of $E_d$ are 1s, not just one of them. This may not be optimal in the size of $G\Xall{t_1}{t_2}$, because we may be stuck exploring a BT node that existed all along $[t_1,t_2)$, while none of its children did. However, the algorithm is still wco with respect to the size of $G_t$ for {\em any} $t \in [t_1,t_2)$: if the subtree of a node $v$ at depth $d$ does not exist in $G_t$, then the corresponding bit $E_d$ will be zero and the algorithm will not attempt to explore it. Instead of satisfying the AGM bound on $|G\Xall{t_1}{t_2}| = |\cap_{t \in [t_1,t_2)} G_t|$, we satisfy it in terms of $\min_{t \in [t_1,t_2)} |G_t|$. This proves the first part of Thm.~\ref{thm:Gallt1t2duration-intro}.

For the case where we look for solutions of a query $Q$  that hold over a minimum duration $\delta$, we first solve the query on $G$, adding the same variable $t$ as the fourth component of all the triple patterns, and use the LTJ tries where the time component {\sc t} is at the end, as in join-first approaches. For the time level we use a more sophisticated data structure, based on geometric grids, that allows us to quickly discard time intervals smaller than $\delta$ while at the same time storing the children of each LTJ trie node $v$ in an efficient way; see the extended version (Appendix~\ref{ap:duration-queries}).  This is the second part of Thm.~\ref{thm:Gallt1t2duration-intro}.

\section{Experiments}

Appendix~\ref{ap:implem}, in the extended version, gives details on the implementation of our index.
In this section we evaluate its \new{empirical} performance over \new{real} temporal graphs\new{, in two scenarios. 
First, we evaluate our index on a temporal graph and real queries from Wikidata, aiming to assess its practicality and the advantage of being able to bind the time at any point, not just at the beginning (time-first) or at the end (join-first). Second, we compare our index with two other systems that offer comparable functionality \cite{HSGAY22,ZhuFY21} over real temporal graphs with synthetic queries.}

We ran our experiments on a 16-core Intel Xeon Silver 4110 with Debian 5.10.127-2, clocked at 2.1 GHz, with 768 GB of RAM and 10 MB cache. Our C++ code was compiled with \texttt{-Ofast}. 
Our benchmark, code and datasets are available on Github. %\textcolor{red}{OJO datasets tb} %~\cite{ourgithub}.
%\textcolor{red}{falta algo? uname -a dice "Linux shannon 5.10.0-0.deb10.16-amd64 $\#1$ SMP Debian 5.10.127-2~bpo10+1 (2022-07-28) x86$\_$64 GNU/Linux", no lo sé resumir :-)}

%We focus on three questions. The first evaluates whether the asymptotic guarantees of our index result in an implementation that is efficient in practice, rather than merely optimal in theory.
%\begin{itemize}
%\item[\textbf{Q1}:] How do the theoretical time and space bounds of our index structure translate into practical performance?
%\end{itemize}
%An aspect overlooked in the theory of worst-case optimal algorithms is variable binding order, which does not affect theoretical bounds, but does affect performance in practice. Hence we ask:
%\begin{itemize}
%\item[\textbf{Q2}:] How do different variable elimination orders affect the performance of our algorithm? Which ordering strategy provides the best results in terms of performance?
%\end{itemize}
%A key advantage of our index structure is its support for variable orderings that bind time variables across graph variables. Our third question evaluates the practical impact of this flexibility.
%\begin{itemize}
%\item[\textbf{Q3}:] Can we design efficient variable ordering strategies $%beyond time-first and join-first? %How important are such strategies for the overall performance of our algorithms?
%\end{itemize}

\subsection{\new{Wikidata with real query logs}}

\paragraph{Benchmark}

We perform experiments over the Wikidata knowledge graph~\cite{VrandecicK14}. We choose Wikidata as it provides a (1) large-scale, (2) real-world, (3) diverse knowledge graph featuring (4) temporal annotations, further providing a comprehensive log of (5) real-world queries~\cite{MalyshevKGGB18}. As per many real-world temporal graphs, some relations include temporal annotations, while others are asserted without temporal qualification (and considered universally valid).

\begin{table}[t]
\caption{\new{The datasets we use, including space in bytes per tuple and construction time in microseconds per edge.}}
\setlength{\tabcolsep}{4pt}
%\vspace*{-3mm}
\begin{tabular}{lrrrrr}
\toprule
Dataset & \multicolumn{1}{c}{$N$} & \multicolumn{1}{c}{$|\U_G|$} & \multicolumn{1}{c}{$|\T_G|$} & Space & Time \\
\midrule
Wiki  & 844,172,299 & 116,145,285 & 155,956 & 241 & 38.0 \\
WikiT & 15,943,638 & 9,384,937  & 155,956 & 230 & 34.7 \\
\midrule
Divvy & 21,243,344 & 7,158     & 20,673,270 & 190 & 25.0 \\
Yellow & 30,003,832 & 8,552 & 20,797,780 & 194 & 29.1 \\
%Stack & 63,497,050 & 2,601,977 & 52,423,237 & 256 & 42.5 \\
Caida & 15,792,089 %80,355,707 
      & 110,704   & 336        & 141 & 20.8 \\
\bottomrule
\end{tabular}

\label{tab:datasets}
\end{table}

We extract the temporal graph from Wikidata considering all \new{$844$ million triples (see Table~\ref{tab:datasets})} % $N=844{,}172{,}299$ triples 
consisting of items and properties. We extract temporal qualifiers on Wikidata statements, where available, for start time ($\ti$), end time ($\tf$), and point in time ($\ti$ and $\tf$). 
%Since temporal values have different levels of precision, we map them to a uniform precision of date (i.e., day-level). 
In the resulting temporal graph, \new{about 16 million} %15,943,638 
triples have non-trivial temporal annotations. 
\tg{We consider two graphs: the \textit{full graph} includes all %844,172,299 
triples including both universally-valid and temporally-restricted relations, while the \textit{transient graph} \new{(WikiT in Table~\ref{tab:datasets})} only includes the \new{16 million} %15,943,638 
triples with non-trivial temporal annotations. \new{The full graph has 12,894 distinct predicates; the transient one 1,161.}}{\new{The graph features 12,894 distinct predicates.}}
%Both graphs include $|\T_G| = 155{,}956$ unique time instances. %The temporal graph contains $|\U_G| = 116{,}158{,}179$ identifiers, coming from 116,145,285 nodes and 12,894 unique predicates, and $|\T_G| = 155{,}956$ unique time instances.

We extract 1000 BGPs from the real-world query logs published for Wikidata~\cite{MalyshevKGGB18}; more details are given in the extended version (Appendix~\ref{app:bgps}). We filter queries whose constants do not occur in the \tg{full and transient graphs, resulting in 974 and 718 BGPs, respectively}{graph, resulting in 974 queries}. We run queries with a result limit of 1,000 and a timeout of 600 seconds.

%We filter disconnected and duplicate BGPs (modulo graph isomorphism). In order to ensure that the BGPs we use for experiments are likely to touch the non-trivial temporal annotations, we define a score to prioritize BGPs with more triple patterns whose predicates have a higher ratio of triples with non-trivial temporal annotations. More specifically, for each predicate $p_i$, in case it is a constant, we define $r_i$ as the ratio of triples with predicate $p_i$ that have non-trivial temporal annotations in the graph; otherwise, if $p_i$ is a variable, we define $r_i$ as 0. For each BGP extracted from the log of the form $Q = \{(s_1,p_1,o_1), \ldots, (s_m,p_m,o_m)\}$, we compute the score $\sum_{i = 1}^m r_i$. We then extract the top 1,000 BGPs from the query logs per this score. Further filtering BGPs with constants not appearing in the graph, we arrive at 972 queries.  
%Finally, we extend each triple pattern with a shared temporal variable $v \in \Vt$ to generate $Q' = \{(s_1,p_1,o_1,v), \ldots, (s_m,p_m,o_m,v)\}$. 

\paragraph{Indexing} \new{Table~\ref{tab:datasets} shows that our index \tg{built on the full graph}{} (Wiki) uses 241 bytes per edge, or 60} %240.7 bytes per edge, or 60.2 
32-bit words. This is 12 times the space needed to store the edges in plain form (as five 32-bit integers). For comparison, a standard index on non-temporal graphs uses 157 bytes per edge~\cite{cltj}. From the total space, 84\% is used by the 6 tries starting with the time component, with \textsc{tspo}, \textsc{tpos}, and \textsc{tosp} being the largest tries  (see the extended version, Appendix~\ref{ap:implem}).
%Figure~\ref{fig:metatrie} in Appendix~\ref{app:bgps}).

Indexing \tg{the full graph} took 9 hours,
%: 38 %\textcolor{red}{8.99 hours: 37.91 
%microseconds per edge. It grows 
growing linearly with the data size (verified by indexing increasing subsets). About 9\% of the time was for building the 12 regular tries; the rest is used to build
the 6 tries starting with the time component: building the linear-space data structure we designed in Section~\ref{sec:linear} requires sorting the data $\lceil \log_2 |\U_G| \rceil = 28$ times per level, not just once per the standard tries (those use faster radix-sorts, however). Our construction runs essentially in-place, using no more space than the final index.

%\paragraph{Queries} For this proof-of-concept experiment, we illustrate the difference between solving the queries with a pure time-first approach, a pure join-first approach, and a general strategy that can bind the time variables at any point. We choose the next variable to bind according to its estimated selectivity: we bind first the variables whose trie nodes have the least children, as in previous work \cite{HRRSiswc19} (we also force that the next variable to bind is connected to a previously bound one, when possible). .....

\paragraph{Variable elimination} \new{A standard LTJ strategy \cite{HRRSiswc19} is to choose as the next variable the one whose trie nodes have the least amount of children among those (if possible) that are connected to a previously bound variable. Lonely variables (i.e., those appearing once in the BGP) are bound at the end. For temporal graphs we modify this strategy in three ways: join-first (\textbf{JF}) leaves the temporal variables to the end, time-first (\textbf{TF}) binds the temporal variables first, and least-children (\textbf{LC}) treats the time as any other variable (the ability to do this  distinguishes our index from previous work). For nodes at the level {\sc t}, we divide the total temporal span divided by their number of children, to estimate the number of children of its nodes.} 
%%\begin{itemize}
%\begin{description}%[leftmargin=0cm, labelwidth=1.2cm, style=nextline]
%\item[\textbf{JF}] (Join-First): The order is decided first as if the query did not have a time component, hence doing a standard LTJ join, as in e.g. \cite{AGHNRRStods24,HRRSiswc19}, and then we process all time variables. 
%\item[\textbf{TF}] (Time-first): First we process all time variables, and then proceed with a non-temporal LTJ ordering. 
%\item[\textbf{LC}] (Least children): A general strategy used in LTJ \cite{HRRSiswc19} in which the next variable is the one whose trie nodes have the least amount of children, while requiring that the next variable to bind is connected to a previously bound one, when possible. % \textcolor{red}{Diego: explicar la estrategia que usas ahora}
%%This strategy was used before in an LTJ context \cite{HRRSiswc19}.
%%\item[\textbf{BEST}]: This strategy simulates an oracle that always chooses the best ordering among \textbf{JF}, \textbf{TF} and \textbf{LC}. 
%\end{description}

%\end{itemize}

\paragraph{Results}

% \begin{table}[t]
% \caption{Median and average times (seconds) for TF, JF, LC, and BEST over all queries and queries with no results ($\emptyset$}
% \centering
% \begin{tabular}{lrrrr}
% \hline
%  & \textbf{median} & \textbf{average} & \textbf{median} ($\emptyset$) & \textbf{average} ($\emptyset$) \\
% \hline
% TF & 2.710 & 1670.15 & 0.021 & 1244.000 \\
% JF & 0.058 & 7.34 & 0.019 & 0.259 \\
% LC & 0.226 & 1631.64 & 0.018 & 1217.000 \\
% BEST & 0.058 & 7.05 & 0.017 & 0,257 \\
% \hline
% \end{tabular}
% \label{tab:full}
% \end{table}

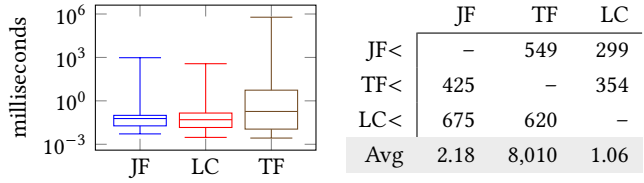
\begin{figure}
\begin{tikzpicture}[baseline=(current bounding box.center)]
\begin{axis}[
    ymode=log, % recommended here, since your values range up to 6000
    boxplot/draw direction=y,
    height=3.5cm,
    width=4.5cm,
    xtick={1,2,3,4},
    xticklabels={
        JF,
        LC,
        TF,
        BEST
    },
    %x tick label style={rotate=20, anchor=east},
    ylabel={milliseconds},
    ytick={1e-3,1e0,1e3,1e6},
    yticklabels={$10^{-3}$,$10^{0}$,$10^{3}$,$10^{6}$}
]

% 	Join-first	Time-first	General VEO
% Median	0.000058365	0.002714711	0.000226887
% 75th	0.00414101375	0.03956103725	0.0323174515
% 25th	0.000018444	0.00002060625	0.00001725125
% Max	0.870441613	61.02413128	61.02209588
% Min	0.00000422	0.000003818	0.000003506

% Join-first
\addplot+[
    boxplot prepared={
        upper whisker= 956.015974, %2254.434423
        upper quartile= 0.09942325, %3.4128565
        median= 0.0599375, %%0.07081,
        lower quartile= 0.01872025, %0.01678425, %0.000018444,
        lower whisker= 0.005213 % %0.00389 %0.00000422
    },
] coordinates {};

% General VEO
\addplot+[
    boxplot prepared={
        upper whisker= 367.957499, %60, %61.02209588,
        upper quartile= 0.1437915, %0.0323174515,
        median= 0.049788, %0.000226887,
        lower quartile= 0.01445575, %0.00001725125,
        lower whisker= 0.002988 %0.000003506
    },
] coordinates {};

% Time-first
\addplot+[
    boxplot prepared={
        upper whisker= 600000, %60, % 61.02413128,
        upper quartile= 5.4751015, %0.03956103725,
        median= 0.1849555, %0.002714711,
        lower quartile= 0.01116675, %0.00002060625,
        lower whisker= 0.002687 %0.000003818
    },
] coordinates {};

% % BEST
% \addplot+[
%     boxplot prepared={
%         median=0.0058365,
%         upper quartile=0.38146295,
%         lower quartile=0.0015986,
%         upper whisker=87.0441613,
%         lower whisker=0.0003506
%     },
% ] coordinates {};

\end{axis}
\end{tikzpicture}
%\new{
%$~~$\begin{tabular}{r|rrr}
%\multicolumn{1}{l}{} & JF & $<$TF & $<$LC \\\cline{2-4}
%JF & -- & 805 & 641 \\
%TF & 167 & -- & 287 \\
%LC & 331 & 684 & -- \\
%\hline
%Avg   & 9.43 & 8010 & 1.06 \\
%\end{tabular}}
\hfill\new{
\def\arraystretch{1.2}
$~~$\begin{tabular}{r|rrr}
\multicolumn{1}{l}{} & JF & TF & LC \\\cline{2-4}
JF$<$ & -- & 549 & 299 \\
TF$<$ & 425 & -- & 354 \\
LC$<$ & 675 & 620 & -- \\
 \multicolumn{1}{l}{\cellcolor{gray!15} Avg}   &  \cellcolor{gray!15} 2.18 &  \cellcolor{gray!15} 8,010 & \cellcolor{gray!15} 1.06 \\
\end{tabular}}

\tg{\vspace{2mm}

\begin{tikzpicture}[baseline=(current bounding box.center)]
\begin{axis}[
    ymode=log, % recommended here, since your values range up to 6000
    boxplot/draw direction=y,
    height=3.5cm,
    width=5cm,
    xtick={1,2,3,4},
    xticklabels={
        JF,
        TF,
        LC,
        BEST
    },
    %x tick label style={rotate=20, anchor=east},
    ylabel={Milliseconds}
]

% 	Join-First	Time-First	General VEO
% Median	0.000046042	0.003008529	0.0000408525
% 75th	0.00015940325	0.101319412	0.0011892095
% 25th	0.00001993475	0.00001840275	0.000015224
% Max	0.224171369	61.00326338	61.00846217
% Min	0.000005024	0.000003559	0.000002911

% Join-first
\addplot+[
    boxplot prepared={
        upper whisker= 2296.886522, %0.224171369,
        upper quartile= 0.14016525, %0.00015940325,
        median= 0.047523, %0.000046042,
        lower quartile= 0.016904,%0.00001993475,
        lower whisker= 0.005004%0.000005024
    },
] coordinates {};

% Time-first
\addplot+[
    boxplot prepared={
        upper whisker= 112964.3072, %60, %61.00326338,
        upper quartile= 95.087153,%0.101319412,
        median= 2.927471,%0.003008529,
        lower quartile= 0.0161495, %0.00001840275,
        lower whisker= 0.004449%0.000003559
    },
] coordinates {};

% General VEO
\addplot+[
    boxplot prepared={
        upper whisker= 430.783787, %60, %61.00846217,
        upper quartile= 0.182011,%0.0011892095,
        median= 0.0485145,%0.0000408525,
        lower quartile= 0.016526, %0.000015224,
        lower whisker= 0.005206 %0.000002911
    },
] coordinates {};

\end{axis}
\end{tikzpicture}
%\new{
%$~~$\begin{tabular}{l|rrr}
%\multicolumn{1}{l}{} & JF & $<$TF & $<$LC \\\cline{2-4}
%JF & -- & 575 & 200 \\
%TF & 145 & -- & 129 \\
%LC & 520 & 591 & -- \\
%\hline
%Avg & & & \\
%\end{tabular}}
\new{
$~~$\begin{tabular}{r|rrr}
\multicolumn{1}{l}{} & $<$JF & $<$TF & $<$LC \\\cline{2-4}
JF & -- & 554 & 381 \\
TF & 164 & -- & 155 \\
LC & 337 & 563 & -- \\
\hline
Avg & 6.37 & 723 & 2.19\\
\end{tabular}}}{}

%\vspace*{-3mm}
\caption{Left: Boxplots of runtimes in msec. Right: number of queries for which the approach in the row runs faster than that in the column\new{, with average times at the bottom}. \tg{Top: full graph. Bottom: transient graph.}{}} 
\label{fig:boxplot}
\end{figure}

\begin{figure}[t]
\begin{tikzpicture}[baseline=(current bounding box.center)]
\begin{axis}[
    ymode=log, % recommended here, since your values range up to 6000
    boxplot/draw direction=y,
    height=3.5cm,
    width=3.0cm,
    xtick={1,2,3,4},
    xticklabels={
        JF,
        LC
    },
    %x tick label style={rotate=20, anchor=east},
    ylabel={milliseconds},
    ytick={1e-2,1e0,1e2,1e6},
    yticklabels={$10^{-2}$,$10^{0}$,$10^{2}$,$10^{6}$}
]

% 	Join-first	Time-first	General VEO
% Median	0.000058365	0.002714711	0.000226887
% 75th	0.00414101375	0.03956103725	0.0323174515
% 25th	0.000018444	0.00002060625	0.00001725125
% Max	0.870441613	61.02413128	61.02209588
% Min	0.00000422	0.000003818	0.000003506

% Join-first
\addplot+[
    boxplot prepared={
        upper whisker= 104.967825, %2254.434423
        upper quartile= 0.10547175, %3.4128565
        median= 0.066267, %%0.07081,
        lower quartile= 0.02116325, %0.01678425, %0.000018444,
        lower whisker= 0.005487 % %0.00389 %0.00000422
    },
] coordinates {};

% General VEO
\addplot+[
    boxplot prepared={
        upper whisker= 104.198847, %60, %61.02209588,
        upper quartile= 0.1315925, %0.0323174515,
        median= 0.065423, %0.000226887,
        lower quartile= 0.02016625, %0.00001725125,
        lower whisker= 0.005845 %0.000003506
    },
] coordinates {};

% % BEST
% \addplot+[
%     boxplot prepared={
%         median=0.0058365,
%         upper quartile=0.38146295,
%         lower quartile=0.0015986,
%         upper whisker=87.0441613,
%         lower whisker=0.0003506
%     },
% ] coordinates {};

\end{axis}
\end{tikzpicture}
\begin{tikzpicture}[baseline=(current bounding box.center)]
\begin{axis}[
    ymode=log, % recommended here, since your values range up to 6000
    boxplot/draw direction=y,
    height=3.5cm,
    width=6.0cm,
    xtick={1,2,3,4},
    xticklabels={
        JF,
        LC,
        TF,
        Hu et al.
    },
    %x tick label style={rotate=20, anchor=east},
    %ylabel={milliseconds},
    ytick={1e-3,1e0,1e3,1e6},
    yticklabels={$10^{-3}$,$10^{0}$,$10^{3}$,$10^6$}
]

% 	Join-first	Time-first	General VEO
% Median	0.000058365	0.002714711	0.000226887
% 75th	0.00414101375	0.03956103725	0.0323174515
% 25th	0.000018444	0.00002060625	0.00001725125
% Max	0.870441613	61.02413128	61.02209588
% Min	0.00000422	0.000003818	0.000003506

% Join-first
\addplot+[
    boxplot prepared={
        upper whisker= 4532.127405, 
        upper quartile= 120.182667,
        median= 2.4871375,
        lower quartile = 0.02860425,
        lower whisker= 0.003889
    },
] coordinates {};

% LC
\addplot+[
    boxplot prepared={
        upper whisker= 4576.769716,
        upper quartile= 120.370636,
        median= 5.498376,
        lower quartile= 0.0350075,
        lower whisker= 0.004353
    },
] coordinates {};

% TF
\addplot+[
    boxplot prepared={
        upper whisker= 601003.121086,
        upper quartile= 62714.71236125,
        median= 729.42662,
        lower quartile= 22.95212425,
        lower whisker= 0.00187
    },
] coordinates {};

% Hu et al.
\addplot+[
    boxplot prepared={
        upper whisker= 38102.565021, 
        upper quartile= 14077.96570125,
        median= 1831.8341785,
        lower quartile= 292.4659805,
        lower whisker= 10.067619
    },
] coordinates {};

% % BEST
% \addplot+[
%     boxplot prepared={
%         median=0.0058365,
%         upper quartile=0.38146295,
%         lower quartile=0.0015986,
%         upper whisker=87.0441613,
%         lower whisker=0.0003506
%     },
% ] coordinates {};

\end{axis}
\end{tikzpicture}

%\vspace*{-3mm}
\caption{\new{Left: boxplots of runtimes in msec for the query with clause $t_1 \le t_2$. Right: comparison with Hu et al.~\cite{HSGAY22} for point-in-time queries on WikiT.}}
\label{fig:boxplot2}
\end{figure}
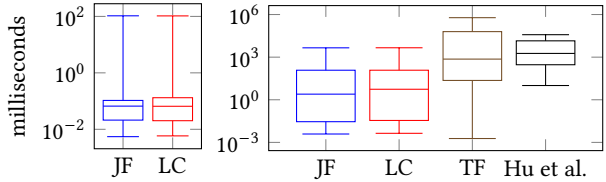

% \begin{table}

% \caption{Number of queries for which the approach in the row gives a lower runtime than the approach in the column} \label{tab:matrix}
% \end{table}

Figure~\ref{fig:boxplot} shows boxplots of query times\tg{ on the full graph}{}, along with \new{averages and} number of queries for which one approach is faster than another. 
% Table \ref{tab:full} provides the median and average times for all queries in our benchmark. 
\new{Our \textbf{LC} strategy is on average twice as fast as {\bf JF} and orders of magnitude faster than {\bf TF}. The boxplots of {\bf LC} and {\bf JF} seem comparable, yet {\bf JF} is 20\% slower in the median. The boxplots and the averages show that {\bf TF} has many bad cases, making it an unstable strategy. By not always binding the temporal variables at the start (per \textbf{TF}) nor at the end (per \textbf{JF}), the \textbf{LC} strategy outperforms \textbf{TF} in 64\% of the queries and \textbf{JF} in 69\%. The average of the query-wise minima of the three times is 0.94 msec, just 11\% faster than \textbf{LC}.}

\tg{The results for \textbf{LC} are not better overall because a relatively small fraction of our graph contains temporal information, while other edges are valid all the time, making the time variable less selective and thus making it more convenient to eliminate it at the end (per \textbf{JF}). To control for this effect, Figure~\ref{fig:boxplot} (bottom) shows the corresponding results for the transient graph, where \textbf{LC} is now the best strategy on more than 70\% of the queries, the best in the median case, though still outperformed by \textbf{JF} at higher percentiles.}{}

%Ya lo explicamos en algún lado, y aquí nos complica un poco el relato.
%One important advantage of the join-first strategy is that, with the time domain being processed at the end, one avoids unnecessary partitions in the time domain, which often results in a lower number of intervals returned. In contrast, the time-first strategy tends to return many more intervals, even though both answers are equivalent: it is just that the answers for time-first are divided into smaller intervals. %This can also be seen in the time statistics for queries that do not return results: in these cases, time-first becomes much more competitive (especially at the median query), and the least-children ordering performs best in the median case. 

% \begin{table}[t]
% \caption{Comparison of median and average times (seconds) for TF, JF, LC, and BEST Over queries without results.}
% \centering
% \begin{tabular}{lrr}
% \hline
%  & \textbf{median} & \textbf{average} \\
% \hline
% TF & 0,021 & 1244 \\
% JF & 0,019 & 0,259 \\
% LC & 0,018 & 1217 \\
% BEST & 0,017 & 0,257 \\
% \hline
% \end{tabular}
% \label{tab:no-results}
% \end{table}

\begin{table*}[t!]
\caption{Average time in msec to solve the different shapes on Divvy, with different windows of the time domain (100\%, 10\%, 1\%) and limiting results to 100,000 \new{(except on the 100\% windows, which set no limit to match Hu et al.~\cite{HSGAY22})}. Times from Zhu et al.~\cite{ZhuFY21} are approximated from their plots. \new{(*): Early termination due to many queries exceeding the 60,000 msec timeout.}}
\vspace*{-2mm}
\begin{tabular}{clrrrrrrrrrrrr}
\toprule
& System & 3-star & 4-star & 5-star & 3-chain & 4-chain & 5-chain & 3-circle & 4-circle & 5-circle & diamond & 4-clique & 5-clique \\
\midrule
\multirow{3}{1mm}{\rotatebox[origin=c]{90}{100\%}}
& \new{Hu et al.~\cite{HSGAY22}} & \new{250} & \new{349} & \new{454} & \new{180} & \new{206} & \new{235} & \new{262} & \new{2,382} & \new{*} & \new{3,614} & \new{1,389} & \new{*} \\
& \new{Ours ({\bf LC})} & \new{100} & \new{116} & \new{138} & \new{57} & \new{85} & \new{1,054} & \new{17} & \new{213} & \new{3,263} & \new{150} & \new{62} & \new{638}\\
& \new{Ours ({\bf TF})} & \new{146} & \new{169} & \new{195} & \new{46} & \new{53} & \new{52} & \new{47} & \new{50} & \new{56} & \new{96} & \new{55} & \new{63}\\
\midrule
\multirow{2}{1mm}{\rotatebox[origin=c]{90}{10\%}}
& Zhu et al.~\cite{ZhuFY21} & 50 & 50 & 50 & 50 & 50 & 70 & 150 & 150 & 70 & 100 & 300 & 700 \\
%& Time \cite{ZhuFY21} & 2,000 & 2,000 & 2,000 & 2,000 & 2,000 & 5,000 & 5,000 & 6,000 & 7,000 & 5,000 & 20,000 & 40,000 \\
& Ours ({\bf TF}) & 12.9 & 23.9 & 18.4 & 4.30 & 6.04 & 4.52 & 4.36 & 4.39 & 4.98 & 8.60 & 4.39 & 6.58 \\
\midrule
\multirow{1}{1mm}{\rotatebox[origin=c]{90}{1\%}}
& \new{Ours ({\bf TF})} & \new{2.22} & \new{4.22} & \new{2.39} & \new{0.97} & \new{1.07} & \new{0.83} & \new{0.80} & \new{0.86} & \new{1.04} & \new{2.06} & \new{1.18} & \new{1.13} \\
\bottomrule
\end{tabular}
%\begin{tabular}{l|rrr}
%Shape & TSRJoin \cite{ZhuFY21} & Time \cite{ZhuFY21} & Ours \\
%\hline
%3-star & 50 & 2,000 & 12.85 \\
%4-star & 50 & 2,000 & 23.93 \\
%5-star & 50 & 2,000 & 18.44 \\
%3-chain & 50 & 2,000 & 4.30 \\
%4-chair & 50 & 2,000 & 6.04 \\
%5-chain & 70 & 5,000 & 4.52 \\
%3-circle & 150 & 5,000 & 4.36 \\
%4-circle & 150 & 6,000 & 4.39 \\
%5-circle & 70 & 7,000 & 4.98 \\
%diamond & 100 & 5,000 & 8.60 \\
%4-clique & 300 & 20,000 & 4.39 \\
%5-clique & 700 & 40,000 & 6.58 \\
%\end{tabular}
\label{tab:divvy}
\end{table*}

\begin{table*}[t!]
\caption{\new{Average time in msec on three collections with our system. We use a 10\% time window and limit the results to 100,000}\label{tab:zhu}}
\vspace*{-2mm}
\begin{tabular}{lrrrrrrrrrrrr}
\toprule
Dataset & 3-star & 4-star & 5-star & 3-chain & 4-chain & 5-chain & 3-circle & 4-circle & 5-circle & diamond & 4-clique & 5-clique \\
\midrule
\new{Divvy} & 12.9 & 23.9 & 18.4 & 4.30 & 6.04 & 4.52 & 4.36 & 4.39 & 4.98 & 8.60 & 4.39 & 6.58 \\
\new{Yellow} & \new{0.23} & \new{0.12} & \new{0.47} & \new{2.80} & \new{20.7} & \new{401} & \new{17.0} & \new{346} & \new{13,913} & \new{197} & \new{204} & \new{34.6} \\
%\new{Yellow ({\bf TBL})} & \new{0.11} & \new{0.11} & \new{0.13} & \new{1.99} & \new{37.6} & \new {1,447} & \new{20.9} & \new{528} & \new{26,086} & \new{459} & \new{327} & \new{11,765} \\
%\new{Yellow ({\bf LC})} & \new{0.23} & \new{0.12} & \new{0.47} & \new{2.80} & \new{20.7} & \new{390} & \new{16.9} & \new{353} & \new{13,875} & \new{196} & \new{204} & \new{34.7} \\
%\hline
\new{Caida} & \new{230} & \new{240} & \new{234} & \new{236} & \new{242} & \new{247} & \new{224} & \new{1{,}416} & \new{3{,}032} & \new{762} & \new{868} & \new{4{,}203}\\
\bottomrule
\end{tabular}
%\begin{tabular}{l|rrr}
%Shape & TSRJoin \cite{ZhuFY21} & Time \cite{ZhuFY21} & Ours \\
%\hline
%3-star & 50 & 2,000 & 12.85 \\
%4-star & 50 & 2,000 & 23.93 \\
%5-star & 50 & 2,000 & 18.44 \\
%3-chain & 50 & 2,000 & 4.30 \\
%4-chair & 50 & 2,000 & 6.04 \\
%5-chain & 70 & 5,000 & 4.52 \\
%3-circle & 150 & 5,000 & 4.36 \\
%4-circle & 150 & 6,000 & 4.39 \\
%5-circle & 70 & 7,000 & 4.98 \\
%diamond & 100 & 5,000 & 8.60 \\
%4-clique & 300 & 20,000 & 4.39 \\
%5-clique & 700 & 40,000 & 6.58 \\
%\end{tabular}
\end{table*}

%\begin{table}[t!]
%\begin{tabular}{l|rrrrrr}
%Query & \multicolumn{5}{c}{Time window selectivity} \\
%         & 0.01\% & 0.1\% & 1\% & 10\% & 20\% & \new{100\%} \\
%\hline
%3-circle & 0.12 & 0.26 & 0.80 & 4.36  & 7.23 \\
%4-star   & 0.29 & 0.70 & 4.22 & 23.9 & 52.1 \\
%5-clique & 0.14 & 0.30 & 1.13 & 6.58  & 10.5 \\
%\hline
% & \multicolumn{5}{c}{Output size limit} \\
%&       $10^3$ & $10^4$ & $10^5$ & $10^6$ & $10^7$ & \new{$\infty$} \\
%\hline
%3-circle & 4.06 & 4.08 & 4.36  & 4.06 & 4.08\\
%4-star   & 15.8 & 23.9 & 23.9 & 24.1 & 24.2 \\
%5-clique & 6.60 & 6.61 &   6.58 & 6.59 & 6.61 \\
%\end{tabular}
%\caption{Average time in msec to solve some shapes with our system, first with a time window from $0.01\%$ to \new{$100\%$} of the domain and limiting results to 100,000, then with a 10\% time window and limiting the results to $10^3$ to \new{$\infty$}. \new{Divvy dataset.}}
\label{tab:divvy}
%\end{table}

%\begin{table}[t!]
%\begin{tabular}{l|rr|rr|rr}
%Query & \multicolumn{2}{c}{Divvy} & \multicolumn{2}{c}{Stack} & \multicolumn{2}{c}{Caida} \\
%      & \cite{HSGAY22} & Ours & \cite{HSGAY22} & Ours & \cite{HSGAY22} & Ours \\
%\hline
%3-star   & & & & & & \\
%4-star   & & & & & & \\
%3-chain  & & & & & & \\
%4-chain  & & & & & & \\
%3-circle & & & & & & \\
%4-circle & & & & & & \\
%\end{tabular}
%\caption{\new{Average time in msec to solve some shapes with Hu et al.'s \cite{HSGAY22} and our system, limiting results to \textcolor{red}{???}.}}
%\label{tab:hu}
%\end{table}

\paragraph{Discussion}

\new{The results show that our index is practical on real-life scenarios: using less than twice the space needed for a non-temporal graph, we answer realistic queries within a millisecond.}
%
%The median times of all strategies for our benchmark support deploying temporal queries in practice. Using our index structure we have shown it is possible to answer realistic queries within a few milliseconds in most cases, and with little space overhead: we need 240.7 bytes per edge, which is less than double that needed for a non-temporal graph. %Assuming that temporal edges in practice have at least two intervals of time, this appears better than storing all edges in a standard database system and annotating them with time intervals. 
%Hence, we answer \textbf{Q1} in a positive way. 
%
%Regarding \textbf{Q2}, we see a substantial impact of the type of strategy used. We see that \textbf{TF} is generally the worst, with \textbf{JF} outperforming \textbf{LC} on the full graph in the median case, and \textbf{LC} likewise outperforming \textbf{JF} on the transient graph. 
\new{We also conclude that it is often highly beneficial to allow temporal variables to be eliminated at any point:
%
%Our strategy \textbf{LC} is the direct extension of one of the standard strategies for ordering variables in classic LTJ algorithms, where one prioritizes trie children with fewer nodes. For the temporal case, this still works well in a significant fraction of the queries.
%in several queries (especially when they do not produce results).
%
%Furthermore, the matrices of Figure~\ref{fig:boxplot} show that being able to evaluate queries in different orderings matters, as for some queries  \textbf{LC} yields a substantial speedup: 
\textbf{LC} is the fastest strategy in the majority of the queries, clearly outperforming \textbf{JF} and \textbf{TF}.}
%when most edges have temporal information, \textbf{LC} outperforms \textbf{JF} and \textbf{TF} over 70\% of the time. %This answers \textbf{Q3} positively: while \textbf{JF} is more robust overall, 
\new{This highlights the relevance of a unique capability of our index, and leaves room for designing better heuristics along the lines of \textbf{LC}.} %Designing better heuristics than \textbf{LC} is an interesting challenge, as now binding time variables early has secondary effects like possibly splitting the time ranges we output into many subintervals. %A promising heuristic could be to incorporate as a factor the estimated cost of splitting intervals when choosing to adopt a strategy other than join-first.

%(they should probably assign more priority to non-time variables than currently done by \textbf{LC}). This is an interesting venue for future work. 

\paragraph{\new{A more complex query}}
\new{Our index solves more than just point-in-time queries. We demonstrate its performance on more complex queries by using now two time variables, $t_1$ and $t_2$, and randomly assigning $t_1$ or $t_2$ (but at least once each) to the BGP triples in order to build the tBGP, and complete the query with a clause $t_1 \le t_2$. We convert the 825 BGPs having more than one triple. Figure~\ref{fig:boxplot2} (left) shows the distribution of times using {\bf JF} and {\bf LC}; note that the {\bf TF} strategy is very inefficient when there are two (or more) time variables, as it would lead to $\Omega(T^2)$ query time.}

\new{In this case, {\bf LC} and {\bf JF} perform similarly, with a distribution close to that of the simpler point-in-time query. For example, the median of {\bf LC} is 65 microseconds, just 30\% higher than the median of the point-in-time queries (50 microseconds). This shows that our index can efficiently solve more complex queries as well.}

\paragraph{\new{Comparison with previous work}}
\new{Figure~\ref{fig:boxplot2} (right) compares the time to solve point-in-time queries using our index and that of Hu et al.~\cite{HSGAY22}. That index implements queries with duration $\delta$, which for $\delta=1$ are point-in-time queries. We chose the strategies that performed best: generic join for cyclic queries and their acyclic baseline for the rest. As they do not limit the output size, we run both systems in that mode over a reduced Wikidata graph having only the triples with non-trivial time annotations; see WikiT in Table~\ref{tab:datasets}. 
We only report on the 76 BGPs that are supported by the code of Hu et al., that is, with variable subject and object. 
%; the graph has 1,161 predicates. 
%We retain the 768 BGPs that give results on this subgraph.} \textcolor{red}{discutir el resultado (brevemente!)
It can be seen that {\bf JF} and {\bf LC} outperform Hu et al. by orders of magnitude, which shows that our index has a competitive advantage over state-of-the-art solutions. {\bf TF} is much slower in this case (still with better median but with worse distribution than Hu et al.); we consider next a scenario where {\bf TF} is more competitive.}
%because the time universe on WikiT is as large as on Wiki but the universally valid edges are absent; trying out all time instants one by one is thus inefficient.}
%Por si hacen falta, los tiempos promedios:
%JF: 420 millisecs
%LC: 424  millisecs
%TF: 136596 millisecs
%Hu et al.: 8273 millisecs

\subsection{\new{Other datasets with synthetic queries}}
%\subsection{Comparison with previous work}

\new{Zhu et al.~\cite{ZhuFY21} study BGPs with fixed topologies (stars, chains, cycles, diamonds, cliques of various sizes) on various datasets. Table~\ref{tab:datasets} describes the datasets we use; Appendix~\ref{app:datasets} (extended version) gives more details. For each topology, they generate 100 point-in-time queries where all nodes are variables and edge labels are assigned a random constant so that the tBGP occurs at least once in the graph. We evaluate our system and related work on this setting.}

\paragraph{\new{Hu et al.}}
\new{Table~\ref{tab:divvy} shows the times of Hu et al.~\cite{HSGAY22} on Divvy, with a 60-second timeout,
% Ya dicho: The shared time variable is simulated with their \emph{durability} parameter set to $1$. 
compared with our system using {\bf LC} and {\bf TF}, and no output size limit. On the acyclic queries (stars and chains), {\bf TF} is 1.7 to 4.5 times faster than Hu et al.; on the cyclic queries it is up to 50 times faster (and more in the shapes where Hu et al.\ have many timeouts). {\bf LC} outperforms {\bf TF} on the stars, but it is (sometimes much) slower in the other shapes. As choosing good VEOs makes a sharper difference on cyclic queries, this suggests that we take advantage from the freedom to choose good query plans. But it also shows that finding the best plan can be challenging.}

\paragraph{\new{Zhu et al.}} 
\new{This paper implements point-in-time queries restricted to a time window $[t_s,t_e]$ given with the query; this corresponds in our language to adding the same temporal variable $t$ to all BGPs and adding constraints $t_s \le t$ and $t \le t_e$. They restrict the results to random time windows that cover a fraction of the time domain, set an output size limit,
%go from 0.01\% to 20\% of the time domain. Further, they consider queries that deliver from $1$ to at most $10^3,\ldots,10^7$ results, 
and set the timeout to 60 seconds. While they do not offer public code, we follow their descriptions to perform similar experiments on Divvy, which is close to their description in size and other parameters. Their TSRJoin index (the one that performs best) takes 483 bytes per tuple (2.5 times our space) and is built 4 times faster than ours in their machine (which is clocked at 2.9 GHz, but its architecture is older).}
Table~\ref{tab:divvy} compares query performance with a time window of 10\% of the domain and limiting the results to 100{,}000, which is the only configuration where we can compare all query shapes. We show their best time (TSRJoin).
%, and the next-best (Time, a time-first approach). 
Our approach, using only {\bf TF} strategy this time, is 2--100 times faster than TSRJoin. % and 3--4 orders of magnitude faster than Time. 
While the 2--4 speedup factors on star shapes can be attributed to our more modern machine\footnote{According to Gemini, our hardware should be from 25\%--35\% faster for pure-CPU calculations, to 3 times faster for computations bounded by RAM transfers.}, we are 1--2 orders of magnitude faster on the other shapes. We conjecture that this is mainly due to our ability to choose arbitrary VEOs for the other variables: we outperform TSRJoin more sharply on the queries where it is not trivial to find the best VEO. On the star queries, instead, there is only one join variable \new{apart from time.} 
%and the query time is dominated by the large output size. %\textcolor{red}{cuál es ese output size vs los que no son tan grandes?}

\paragraph{More on our performance.} 
By looking at the ``Ours'' rows in the areas 100\%, 10\%, and 1\% of Table~\ref{tab:divvy}, we can see that the time window sensitivity has a significant impact on the query time of our index, showing that it filters effectively using the time component.

\new{Table~\ref{tab:zhu} measures our times on other datasets mentioned in Zhu et al.'s paper \cite{ZhuFY21} (but much larger than their versions, so the results are not comparable): Yellow and Caida. We use a 10\% time window and limit the results to 100,000. 
% Eliminado porque sale muy cercano a TF, se ve muy débil
%We compare some VEO strategies on Yellow: since {\bf JF} was too slow, we show only {\bf TF}, {\bf LC}, and a new strategy {\bf TBL}, which puts the time variable just before the lonely variables. To favor the choice of the temporal variable in {\bf LC} due to the 10\% window, we modify the estimation of the number of children at level {\sc t}, by dividing it by $k^4$ if the variable $t$ appears in $k$ triple patterns. While this is too tuned for this situation, we aim to demonstrate that it is possible to design VEO strategies that outperform standard ones like {\bf TF} and {\bf TBL} on the costlier queries. 
The differences in performance are explained mostly by the distribution of predicates: there are a few thousand in Divvy and Yellow, but those in Yellow are considerably skewed, while those in Divvy distribute uniformly. Caida, instead, has just one predicate. Since queries fix a predicate at random and leave all nodes as variables, the times in Caida are high and relatively uniform, growing as larger structures are reported: essentially we are generating all the ways a shape appears in the graph. Predicates help filter edges on Divvy, which yields low and relatively uniform times. On Yellow, with a skewed distribution, filtering is effective for most predicates, but the high times obtained on the larger structures when using a popular predicate dominate the average.} 
%\textcolor{red}{no se me ocurrió otra razón.}
%, are much larger than theirs: 
%Stack records the communication among users in StackOverflow and Caida records the relationships among autonomous systems. In these datasets, edges have no labels (i.e., there is only one label) and span only one time instant; we join successive times of an edge into maximal ranges. \textcolor{red}{comentar lo que corresponda según tipo de query y de dataset.}}

%Finally, Table~\ref{tab:divvy} shows, on a few ``easy'' and a ``hard'' query shapes \new{on Divvy}, how our query times evolve as we vary the selectivity of the window and the limit on the output size. The time window sensitivity has a significant impact on query time, showing that our index filters effectively using the time component. The output size limit has, instead, little impact: only some on the 4-star, whose time is largely related to the number of outputs produced \textcolor{red}{verificar esto (quisiera ver nro de outputs). comentar sobre las nuevas columnas tb. no será simplemente que las queries tienen pocos resultados en divvy, y por eso da lo mismo el límite? de ser así tendría poca gracia el experimento. se podría reemplazar por agregar una fila para winow=1\% a la tabla 2, y ya tendríamos para comparar 1\%, 10\% y 100\%.}.

\section{Conclusions}

%\textcolor{red}{SUMARIZAR}
We have described a data structure that enables processing temporal BGPs in wco time, and provides further guarantees for other key applications such as point-in-time queries or queries with duration. \new{Our experiments show that we solve realistic queries on large graphs within milliseconds, and that our strategy clearly outperforms simple solutions, as well as previous work. Although our index is static, we describe in Appendix~\ref{ap:dynamism} how to support adding new events (i.e., edge insertions and deletions) to the graph.}

One interesting challenge for future work is to obtain beyond-wco query times. For example, we can easily obtain time related to the generalized hypertree width of the query~\cite{emptyheaded}, which is never larger than its fractional hypertree width (the AGM bound). This can be done by optimally partitioning the query into a tree of cyclic queries, using our wco algorithm to solve each cyclic component, and then applying Yannakakis' algorithm \cite{yannakakis} to solve the resulting acyclic join of intermediate results. A problem is that those intermediate results might be obtained with the bindings of the time variables in non-compacted form (akin to representing $\hat{G}$ instead of $G$), and then they might require a lot of space. 
A solution would be to represent the intermediate results optimally using time intervals, yet %. While adapting Yannakakis' semijoins to this representation does not seem to be problematic, efficiently 
finding the minimal set of hyperrectangles covering all the instantiations of the temporal variables seems difficult.
%, in optimal space and time, if the VEO binds them at the end, as they are then output in lexicographical order for each instantiation of the other variables. It does not seem possible to achieve the same in the general case, however.

A related challenge is to choose the most efficient VEOs. As seen in Figure~\ref{fig:boxplot}, which approach works best depends on the query. The least-children (LC) strategy---akin to greedy optimization on cardinality estimates---works well in many cases, but there is room for better heuristics, including adaptive VEOs \cite{cltj}.

% it may split temporal intervals in the output, which increases costs. Better practical heuristics could take such costs into account.

%    \item Can we report the results with maximal time intervals or show it cannot be done? Juan had the idea that, even if we left the time variables to the end, the resolution would force us to hash the intervals.

Another line of future work is to extend the language of the temporal clauses. For example, we may include % relations of the form $w_1 < w_2$ or 
timestamp arithmetic like $w_1 \le w_2+3$ if $\T$ is numeric, which requires representing the original time domain $\T$ (rather than $\T_G$). In this case, the point graph $\hat{G}$ is only bounded by $N \cdot |\T|$, which can even be infinite, and worst-case optimality must be reconsidered. Still, our structures would work without many changes, yet it will be crucial to always bind temporal variables to ranges, as discussed in Section~\ref{sec:inters}.

\new{Finally, wco algorithms can be deployed for standard DBMS~\cite{umbra}, which suggests our techniques could be deployed in a broader setting. This may include temporal path languages~\cite{wu2014path,arenas2022temporal} or even becoming part of temporal DBMS pipelines~\cite{hou2025efficient}.}

%it remains to see how our techniques can be applied in a broader setting where query workloads go beyond tBGPs. Wco algorithms have been deployed in standard DBMS \cite{}, which . For example, it would be interesting to extend our index to support temporal path languages (see e.g. \cite{wu2014path,arenas2022temporal}). And more generally} %understand how to  incorporate our techniques into }%temporal DBMS query plans; this has already been explored for standard relational DBMS \cite{}. }

%interesting to design combined query plans that include our algorithms in a more expressive temporal DBMS pipeline such as that of \cite{hou2025efficient}. There are already systems combining wco algorithms with standard BBMS query plans \cite{}, but it remains to see if this approach can be ported to the temporal setting.}
%    \item Aidan was going to do a search for relevant queries, say in the application of versions. Apparently we support all those he found; we should explain this in the intro...
%    \item I checked Vaisman's paper, but it's mostly about paths. What it has of non-paths is all covered with $t \le t'$ clauses.
%    \item In the paper of conjunctive queries there may be a useful language for time. Or an established language for (time) intervals? Like contained, disjoint, overlaps, etc. All I can imagine is modeled with $t \le t'$...

\begin{acks}
 This work was supported by ANID – Millennium Science Initiative Program – Code ICN17\_002. GN was also supported by Fondecyt Grant 1260080.
 %Aidan: I just have the IMFD
\end{acks}

\bibliographystyle{plain}
\bibliography{paper}

@inproceedings { HSGAY22,
    author = "Xiao Hu and Stavros Sintos and Junyang Gao and Pankaj K. Agarwal and Jun Yang",
    title = "Computing Complex Temporal Join Queries Efficiently",
    booktitle = "Proc. International Conference on Management of Data (SIGMOD)",
    year = 2022,
    pages = "2076--2090",
    !doi = "https://doi.org/10.1145/3514221.3517893"
}

@inproceedings{ CLP11,
  author    = {Timothy M. Chan and Kasper G. Larsen and Mihai P{\u{a}}tra\c{s}cu},
  title     = {Orthogonal range searching on the {RAM}, revisited},
  booktitle = {Proc. 27th ACM Symposium on Computational Geometry (SoCG)},
  pages     = {1--10},
  year      = {2011},
}

@inproceedings{HRRSiswc19,
  author    = {Aidan Hogan and
               Cristian Riveros and
               Carlos Rojas and
               Adri{\'a}n Soto},
  title     = {A worst-case optimal join algorithm for {SPARQL}},
  booktitle = {Proc. 18th International Semantic Web Conference ({ISWC})},
  pages     = {258--275},
  year      = {2019}
}

@inproceedings{MalyshevKGGB18,
  author       = {Stanislav Malyshev and
                  Markus Kr{\"{o}}tzsch and
                  Larry Gonz{\'{a}}lez and
                  Julius Gonsior and
                  Adrian Bielefeldt},
  !editor       = {Denny Vrandecic and
                  Kalina Bontcheva and
                  Mari Carmen Su{\'{a}}rez{-}Figueroa and
                  Valentina Presutti and
                  Irene Celino and
                  Marta Sabou and
                  Lucie{-}Aim{\'{e}}e Kaffee and
                  Elena Simperl},
  title        = {Getting the Most Out of {W}ikidata: Semantic Technology Usage in {W}ikipedia's
                  Knowledge Graph},
  booktitle    = {Proc. 17th International Semantic Web Conference (ISWC)},
  !series       = {Lecture Notes in Computer Science},
  !volume       = {11137},
  pages        = {376--394},
  !publisher    = {Springer},
  year         = {2018},
  url          = {https://doi.org/10.1007/978-3-030-00668-6\_23},
  doi          = {10.1007/978-3-030-00668-6\_23},
  timestamp    = {Tue, 07 Sep 2021 13:47:46 +0200},
  biburl       = {https://dblp.org/rec/conf/semweb/MalyshevKGGB18.bib},
  bibsource    = {dblp computer science bibliography, https://dblp.org}
}

@article{VrandecicK14,
  author       = {Denny Vrandecic and
                  Markus Kr{\"{o}}tzsch},
  title        = {{Wikidata: a free collaborative knowledgebase}},
  journal      = {Commun. {ACM}},
  volume       = {57},
  number       = {10},
  pages        = {78--85},
  year         = {2014},
  url          = {https://doi.org/10.1145/2629489},
  doi          = {10.1145/2629489},
  bibsource    = {dblp computer science bibliography, https://dblp.org}
}

@inproceedings{yannakakis,
  title={Algorithms for acyclic database schemes},
  author={Yannakakis, Mihalis},
  booktitle={Proc. 7th International Conference on Very Large Databases (VLDB)},
  pages={82--94},
  year={1981}
}

@PHDTHESIS {
        Cla96,
        TITLE = "Compact {PAT} Trees",
        AUTHOR = "David R. Clark",
        SCHOOL = "University of Waterloo, Canada",
        YEAR = 1996
        }

@INPROCEEDINGS {
        Mun96,
        AUTHOR = "J. Ian Munro",
        TITLE = "Tables",
        BOOKTITLE = "Proc. 16th Conference on Foundations of Software
                     Technology and Theoretical Computer Science (FSTTCS)",
        YEAR = 1996,
        !SERIES = "LNCS 1180",
        PAGES = "37--42"
        }

@INPROCEEDINGS {
        Jac89,
        AUTHOR = "Guy Jacobson",
        TITLE = "Space-efficient static trees and graphs",
        BOOKTITLE = "Proc. 30th IEEE Symposium on Foundations of Computer
                     Science (FOCS)",
        YEAR = 1989,
        PAGES = "549--554"
        }

@ARTICLE
        { AGHNRRStods24,
          TITLE = "The {R}ing: Worst-Case Optimal Joins in Graph Databases
                   using (Almost) No Extra Space",
          AUTHOR = "Diego Arroyuelo and Adri{\'a}n G{\'o}mez-Brand{\'o}n and Aidan Hogan
                    and Gonzalo Navarro and Juan L. Reutter and Javiel Rojas-Ledesma
                    and Adri{\'a} Soto",
          JOURNAL = "ACM Transactions on Database Systems",
          YEAR = 2024,
          VOLUME = 29,
          NUMBER = 2,
          PAGES = "article 5"
        }

@ARTICLE {
        CNO15,
        AUTHOR = "Francisco Claude and Gonzalo Navarro and Alberto Ord{\'o}{\~n}ez",
        TITLE = "The Wavelet Matrix: An Efficient Wavelet Tree for Large
                 Alphabets",
        JOURNAL = "Information Systems",
        YEAR = 2015,
        VOLUME = 47,
        PAGES = "15--32"
        }

@article{AGM13,
	title        = {Size bounds and query plans for relational joins},
	author       = {Albert Atserias and Martin Grohe and D{\'a}niel Marx},
	year         = 2013,
	journal      = {{SIAM Journal on Computing}},
	volume       = 42,
	number       = 4,
	pages        = {1737--1767}
}

@inproceedings{leapfrog,
	title        = {Triejoin: A simple, worst-case optimal join algorithm},
	author       = {Veldhuizen, Todd L.},
	year         = 2014,
	booktitle    = {Proc.\ International Conference on Database Theory (ICDT)},
	pages        = {96--106}
}

@inproceedings{ BCDMS99,
  author    = {Andrej Brodnik and Svante Carlsson and Erik D. Demaine and
               J. Ian Munro and Robert Sedgewick},
  title     = {Resizable Arrays in Optimal Time and Space},
  booktitle = {Proc. 6th International Symposium on Algorithms and Data
                Structures (WADS)},
  pages     = {37--48},
  year      = {1999},
  !series = "LNCS 1663"
}

@inproceedings{NPRR12,
        title        = {Worst-case optimal join algorithms},
        author       = {Ngo, Hung Q. and Porat, Ely and R{\'e}, Cristopher and Rudra, Atri},
        year         = 2012,
        booktitle    = {Proc. 31st Symposium on Principles of Database Systems (PODS)},
        pages        = {37--48},
        !organization = {ACM}
}

@article{geometric,
        title        = {Joins via geometric resolutions: Worst case and beyond},
        author       = {Khamis, Mahmoud A. and Ngo, Hung Q. and R{\'e}, Cristopher and Rudra, Atri},
        year         = 2016,
        journal      = {ACM Transactions on Database Systems},
        volume       = 41,
        number       = 4,
        pages        = 22,
        !publisher   = {ACM}
}

@inproceedings{tutorialngo,
        title        = {Worst-case optimal join algorithms: Techniques, results, and open problems},
        author       = {Hung Q. Ngo},
        year         = 2018,
        booktitle    = {Proc. 37th Symposium on Principles of Database Systems (PODS)},
        pages        = {111--124}
}

@article { graphflow,
    author = "Amine Mhedhbi and Semih Salihoglu",
    year = 2019,
    title = "Optimizing subgraph queries by combining
binary and worst-case optimal joins",
    journal = "Proceedings of the VLDB Endowment",
    volume = 12,
    number = 11,
       pages = "1692--1704"
    }

@article { adopt, 
    title = "{ADOPT}: Adaptively Optimizing Attribute Orders for Worst-Case Optimal Join Algorithms via Reinforcement Learning",
    author = "Wang, Jialing and Trummer, Immanuel and Kara, Ahmet and Olteanu, Dan",
    journal = "Proceedings of the VLDB Endowment",
    volume = 16,
    number = 11,
    pages = "2805--2817",
    year = 2023
}

@article{freejoin,
  author       = {Wang, Yisu Remy and
                  Willsey, Max and
                  Suciu, Dan},
  title        = {{Free Join}: Unifying Worst-Case Optimal and Traditional Joins},
  journal      = "Proceedings of the ACM on Management of Data (SIGMOD)",
  volume       = {1},
  number       = {2},
  pages        = {150:1--150:23},
  year         = {2023}
}

@article{umbra,
        title        = {Adopting worst-case optimal joins in relational database systems},
        author       = {Freitag, Michael J. and Bandle, Maximilian and Schmidt, Tobias and Kemper, Alfons and Neumann, Thomas},
        year         = 2020,
        journal      = {Proceedings of the VLDB Endowment},
        volume       = 13,
        number       = 11,
        pages        = {1891--1904}
}

@inproceedings{abo2017shannon,
  title={What do {S}hannon-type inequalities, submodular width, and disjunctive datalog have to do with one another?},
  author={Abo Khamis, Mahmoud and Ngo, Hung Q and Suciu, Dan},
  booktitle={Proc. 36th ACM Symposium on Principles of Database Systems (PODS)},
  pages={429--444},  
  year={2017}
}

@inproceedings{ZhuFY21,
  author       = {Kaijie Zhu and
                  George Fletcher and
                  Nikolay Yakovets},
  title        = {Leveraging Temporal and Topological Selectivities in Temporal-clique
                  Subgraph Query Processing},
  booktitle    = {Proc. 37th IEEE International Conference on Data Engineering (ICDE)},
  pages        = {672--683},
  !publisher    = {{IEEE}},
  year         = {2021},
  url          = {https://doi.org/10.1109/ICDE51399.2021.00064},
  doi          = {10.1109/ICDE51399.2021.00064},
  timestamp    = {Tue, 14 Oct 2025 19:36:23 +0200},
  biburl       = {https://dblp.org/rec/conf/icde/ZhuFY21.bib},
  bibsource    = {dblp computer science bibliography, https://dblp.org}
}

@ARTICLE{
        cltj,
        AUTHOR = "Diego Arroyuelo and Daniela Campos and Adri{\'a}n G{\'o}mez-Brand{\'o}n and 
		  Yuval Linker and Gonzalo Navarro and Carlos Rojas and Domagoj Vrgoc",
        JOURNAL = "The Very Large Databases Journal",
	TITLE = "{CompactLTJ}: Space {\&} Time Efficient {L}eapfrog {T}riejoin 
		 on Graph Databases",
	VOLUME = 34,
	PAGES = "article 67",
        YEAR = "2025"
    }

@inproceedings{tsparql,
  title={{T-SPARQL}: A {TSQL2}-like Temporal Query Language for {RDF}},
  author={Grandi, Fabio and others},
  booktitle={ADBIS (local proceedings)},
  volume={639},
  pages={21--30},
  year={2010}
}

@inproceedings{glenda,
  title={{GLENDA}: Querying {RDF} archives with full {SPARQL}},
  author={Pelgrin, Olivier and Taelman, Ruben and Gal{\'a}rraga, Luis and Hose, Katja},
  booktitle={Proc. European Semantic Web Conference (ESWC)},
  pages={75--80},
  year={2023},
  !organization={Springer}
}

@inproceedings{versionedaidan,
  title={Versioned Queries over {RDF} Archives: All You Need is {SPARQL}?},
  author={Cuevas, Ignacio and Hogan, Aidan},
  booktitle={MEPDaW@ ISWC},
  pages={43--52},
  year={2020}
}

@article{sql2011temporal,
  title={Temporal features in {SQL}: 2011},
  author={Kulkarni, Krishna and Michels, Jan-Eike},
  journal={ACM Sigmod Record},
  volume={41},
  number={3},
  pages={34--43},
  year={2012},
  publisher={ACM New York, NY, USA}
}

@inproceedings{timetext1,
  title={A time machine for text search},
  author={Berberich, Klaus and Bedathur, Srikanta and Neumann, Thomas and Weikum, Gerhard},
  booktitle={Proc. 30th Annual International ACM Conference on Research and Development in Information Retrieval (SIGIR)},
  pages={519--526},
  year={2007}
}

@inproceedings{timetext2,
  title={Index maintenance for time-travel text search},
  author={Anand, Avishek and Bedathur, Srikanta and Berberich, Klaus and Schenkel, Ralf},
  booktitle={Proc. 35th International ACM Conference on Research and Development in Information Retrieval (SIGIR)},
  pages={235--244},
  year={2012}
}

@inproceedings{temporalalgebragraph,
  title={Temporal graph algebra},
  author={Moffitt, Vera Zaychik and Stoyanovich, Julia},
  booktitle={Proc. 16th International Symposium on Database Programming Languages},
  pages={1--12},
  year={2017}
}

@inproceedings{temporalalgebrarelational,
  title={Temporal alignment},
  author={Dign{\"o}s, Anton and B{\"o}hlen, Michael H and Gamper, Johann},
  booktitle={Proc. ACM International Conference on Management of Data (SIGMOD)},
  pages={433--444},
  year={2012}
}

@article{statementmodifiers,
  title={Temporal statement modifiers},
  author={B{\"o}hlen, Michael H and Jensen, Christian S and Snodgrass, Richard Thomas},
  journal={ACM Transactions on Database Systems},
  volume={25},
  number={4},
  pages={407--456},
  year={2000},
  publisher={ACM New York, NY, USA}
}

@article{tquel,
  title={The temporal query language {TQuel}},
  author={Snodgrass, Richard},
  journal={ACM Transactions on Database Systems},
  volume={12},
  number={2},
  pages={247--298},
  year={1987},
  publisher={ACM New York, NY, USA}
}

@inproceedings{arenas2022temporal,
  title={Temporal regular path queries},
  author={Arenas, Marcelo and Bahamondes, Pedro and Aghasadeghi, Amir and Stoyanovich, Julia},
  booktitle={2022 IEEE 38th International Conference on Data Engineering (ICDE)},
  pages={2412--2425},
  year={2022},
  organization={IEEE}
}

@inproceedings{stsparql,
  title={Modeling and querying metadata in the semantic sensor web: The model {stRDF} and the query language {stSPARQL}},
  author={Koubarakis, Manolis and Kyzirakos, Kostis},
  booktitle={Proc. Extended Semantic Web Conference},
  pages={425--439},
  year={2010},
  !organization={Springer}
}

@incollection{sparqlst,
  title={{SPARQL}-st: Extending {SPARQL} to support spatiotemporal queries},
  author={Perry, Matthew and Jain, Prateek and Sheth, Amit P},
  booktitle={Geospatial Semantics and the Semantic Web: Foundations, Algorithms, and Applications},
  pages={61--86},
  year={2011},
  publisher={Springer}
}

@inproceedings{temporalontology,
  title={Multi-temporal {RDF} Ontology Versioning.},
  author={Grandi, Fabio and others},
  booktitle={IWOD@ ISWC},
  year={2009}
}

@article{emptyheaded,
  title={Emptyheaded: A relational engine for graph processing},
  author={Aberger, Christopher R and Lamb, Andrew and Tu, Susan and N{\"o}tzli, Andres and Olukotun, Kunle and R{\'e}, Christopher},
  journal={ACM Transactions on Database Systems},
  volume={42},
  number={4},
  pages={1--44},
  year={2017},
  publisher={ACM New York, NY, USA}
}

@article{survey,
  title={Foundations of modern query languages for graph databases},
  author={Angles, Renzo and Arenas, Marcelo and Barcel{\'o}, Pablo and Hogan, Aidan and Reutter, Juan and Vrgo{\v{c}}, Domagoj},
  journal={ACM Computing Surveys},
  volume={50},
  number={5},
  pages={1--40},
  year={2017},
  publisher={ACM New York, NY, USA}
}

@inproceedings{GQL,
  title={A Researcher’s Digest of {GQL}},
  author={Francis, Nadime and Gheerbrant, Am{\'e}lie and Guagliardo, Paolo and Libkin, Leonid and Marsault, Victor and Martens, Wim and Murlak, Filip and Peterfreund, Liat and Rogova, Alexandra and Vrgoc, Domagoj},
  booktitle={Proc. 26th International Conference on Database Theory (ICDT)},
  pages={1:1--1:22},
  year={2023},
  !organization={Schloss Dagstuhl-Leibniz-Zentrum f{\"u}r Informatik}
}

@misc{croker1989completeness,
  title={On completeness of historical relational data models},
  author={Croker, Albert and Clifford, James},
  year={1989},
  howpublished={NYU Working Paper No. IS-89-002}
}

@article{li2022durable,
  title={Durable subgraph matching on temporal graphs},
  author={Li, Faming and Zou, Zhaonian and Li, Jianzhong},
  journal={IEEE Transactions on Knowledge and Data Engineering},
  volume={35},
  number={5},
  pages={4713--4726},
  year={2022},
  publisher={IEEE}
}

@article{semertzidis2018top,
  title={Top-$k$ Durable Graph Pattern Queries on Temporal Graphs},
  author={Semertzidis, Konstantinos and Pitoura, Evaggelia},
  journal={IEEE Transactions on Knowledge and Data Engineering},
  volume={31},
  number={1},
  pages={181--194},
  year={2018},
  publisher={IEEE}
}

@inproceedings{paranjape2017motifs,
  title={Motifs in temporal networks},
  author={Paranjape, Ashwin and Benson, Austin R and Leskovec, Jure},
  booktitle={Proc. 10th ACM International Conference on Web Search and Data Mining (WSDM)},
  pages={601--610},
  year={2017}
}

@article{gao2020time,
  title={Time-respecting flow graph pattern matching on temporal graphs},
  author={Gao, Yunjun and Zhang, Tianming and Qiu, Linshan and Linghu, Qingyuan and Chen, Gang},
  journal={IEEE Transactions on Knowledge and Data Engineering},
  volume={33},
  number={10},
  pages={3453--3467},
  year={2020},
  publisher={IEEE}
}

@inproceedings{huan2025tematch,
  title={{TeMatch}: A Fast Temporal Subgraph Matching Framework with Temporal-Aware Subgraph Matching Algorithms},
  author={Huan, Chengying and Zhang, Heng and Liu, Yongchao and Chen, Likang and Wang, Xuran and Jiang, Yongchun and Ma, Shaonan and Wu, Yanjun},
  booktitle={2025 IEEE 41st International Conference on Data Engineering (ICDE)},
  pages={1029--1042},
  year={2025},
  !organization={IEEE}
}

@inproceedings{abo2022complexity,
  title={The complexity of boolean conjunctive queries with intersection joins},
  author={Abo Khamis, Mahmoud and Chichirim, George and Kormpa, Antonia and Olteanu, Dan},
  booktitle={Proc. 41st ACM SIGMOD-SIGACT-SIGAI Symposium on Principles of Database Systems (PODS)},
  pages={53--65},
  year={2022}
}

@misc { caida,
   title = "{The CAIDA AS Relationships Dataset, 1998--2026}",
   note = "https://www.caida.org/catalog/datasets/as-relationships/",
   year = "downloaded in April 2026"
   }

@article{aitg1,
  title={A survey on temporal knowledge graph: Representation learning and applications},
  author={Cai, Li and Mao, Xin and Zhou, Yuhao and Long, Zhaoguang and Wu, Changxu and Lan, Man},
  journal={arXiv preprint arXiv:2403.04782},
  year={2024}
}

@article{aitg2,
  title={Temporal knowledge graph completion: A survey},
  author={Cai, Borui and Xiang, Yong and Gao, Longxiang and Zhang, He and Li, Yunfeng and Li, Jianxin},
  journal={arXiv preprint arXiv:2201.08236},
  year={2022}
}

@article{cyber,
  title={Temporal Multi-Query Subgraph Matching in Cybersecurity},
  author={Lu, Min and Zhang, Qianzhen and Zhu, Xianqiang},
  journal={Technologies},
  volume={13},
  number={8},
  pages={335},
  year={2025},
  publisher={MDPI}
}

@inproceedings{eventkg,
  title={Eventkg: A multilingual event-centric temporal knowledge graph},
  author={Gottschalk, Simon and Demidova, Elena},
  booktitle={European semantic web conference},
  pages={272--287},
  year={2018},
  organization={Springer}
}

@inproceedings{khurana2016efficient,
  title={Storing and analyzing historical graph data at scale},
  author={Khurana, Udayan and Deshpande, Amol},
  booktitle={International Conference on Extending Database Technology},
  year={2016},
  organization={OpenProceedings. org}
}

@inproceedings{kriegel2000managing,
  title={Managing intervals efficiently in object-relational databases},
  author={Kriegel, Hans-Peter and P{\"o}tke, Marco and Seidl, Thomas},
  booktitle={VLDB},
  volume={20},
  number={0},
  pages={0},
  year={2000}
}

@inproceedings{ceccarello2023indexing,
  title={Indexing temporal relations for range-duration queries},
  author={Ceccarello, Matteo and Dign{\"o}s, Anton and Gamper, Johann and Khnaisser, Christina},
  booktitle={Proceedings of the 35th International Conference on Scientific and Statistical Database Management},
  pages={1--12},
  year={2023}
}

@article{hou2025efficient,
  title={An efficient and scalable graph database with built-in temporal support: J. Hou et al.},
  author={Hou, Jiamin and Zhao, Zhanhao and Lu, Wei and Yang, Shiming and Liu, Shuang and Xu, Quanqing and Yang, Chuanhui and Du, Xiaoyong},
  journal={The VLDB Journal},
  volume={34},
  number={4},
  pages={53},
  year={2025},
  publisher={Springer}
}

@article{wu2014path,
  title={Path problems in temporal graphs},
  author={Wu, Huanhuan and Cheng, James and Huang, Silu and Ke, Yiping and Lu, Yi and Xu, Yanyan},
  journal={Proceedings of the VLDB Endowment},
  volume={7},
  number={9},
  pages={721--732},
  year={2014},
  publisher={VLDB Endowment}
}

\newpage
\appendix
\onecolumn

\section{On the expressive power of tBGPs}\label{app:expressive}

Let us start with the comparison with CQs. To abstract from relational representations of temporal graphs, we consider a relational representation with a single relation $T$ of arity $5$, with the first three positions reserved for triples, the fourth for the start of a time interval, and the fifth for the end of the interval. 
Then, each temporal graph $G$ is directly represented as an instance $I_G$: for each tuple $(s,p,o,[\ti,\tf\,))$ in $G$ we 
add to $I_G$ the atom $(s,p,o,\ti,\tf)$. In this context, one asks whether every tBGP can be \emph{expressed} as a query in $I_G$; we say that a tBPG $Q$ can be expressed as a conjunctive query over the relational representation of graphs if one can find a conjunctive query $Q'$ such that the evaluation of $Q$ over any temporal graph $G$ corresponds to the evaluation of $Q'$ over $I_G$. 

\begin{proposition}
\label{prop-cqs-not-enough}
There is a family of tBGPs that cannot be expressed as conjunctive queries over the relational representation of graphs, nor as conjunctive queries with inequalities. 
\end{proposition}

\begin{proof}
Assume towards contradiction that every tBGP can be expressed as a
conjunctive query over the relational representation of temporal graphs,
possibly extended with inequalities.

Consider, for every $n \ge 1$, the tBGP $Q_n$ consisting of a path of
length $n$ required to hold at the same time instant:
\[
Q_n = \{(x_0,p,x_1,w), (x_1,p,x_2,w), \dots, (x_{n-1},p,x_n,w)\}.
\]
Intuitively, $Q_n$ asks for a path of length $n$ that exists
simultaneously at some time point $w$.

Suppose there exists a conjunctive query $Q_n'$ over the relational
encoding using the relation
$T(s,p,o,\ti,\tf)$ that expresses $Q_n$.

Let $H_n$ be a static graph consisting of a path
\[
a_0 \xrightarrow{p} a_1 \xrightarrow{p} \cdots
\xrightarrow{p} a_n .
\]

We construct two temporal graphs that have the same frozen structure but
different temporal behavior.

\medskip
\noindent
\textbf{Instance $G_1$.}
For every edge $(a_i,p,a_{i+1})$ of $H_n$, include the tuple
\[
(a_i,p,a_{i+1},[t_1,t_3))
\]
with $t_0 < t_1 < t_2 < t_3$.
Additionally, add one extra edge
\[
(b,p',b',[t_0,t_2)).
\]
At time $t_2$, all edges of the path are valid simultaneously, hence
$Q_n(G_1)$ contains a solution witnessing time $w=t_2$.

Since $Q_n'$ expresses $Q_n$, there must be a homomorphism from the body
of $Q_n'$ into $I_{G_1}$ mapping the variable corresponding to $w$ to
the value $t_2$. Consequently, some atom of $Q_n'$ must use $t_2$ as the
value obtained from an interval endpoint appearing in position $5$
(the interval end), because $t_2$ only occurs there.

\medskip
\noindent
\textbf{Instance $G_2$.}
Now consider the temporal graph obtained by replacing the additional
edge with
\[
(b,p',b',[t_2,t_4))
\]
for $t_3 < t_4$, while keeping all path edges unchanged.
Again, the frozen graph is identical to $H_n$, and now the witnessing
time belongs to the start of an interval, namely position $4$.

Hence, correctness of $Q_n'$ implies that there must exist a
homomorphism mapping the same time variable to position $4$ of some atom.

But we also need to map the atom mapping $w$ to the fifth position. As the value $t_2$ is not in the fifth position in any tuple in $I_{G_2}$, it follows that there cannot be a homomorphism from $Q_n'$ to $I_{G_2}$, which is a contradiction. Therefore, no conjunctive query (even with inequalities) can express any of the tBGPs $\{Q_n\}_{n \ge 1}$.
\end{proof}

However, tBGPs do coincide with conjunctive queries (with inequalities) when we expand graphs into their point-based representation, which may be of quadratic size with respect to the original temporal graph. While materializing the point-based representation is therefore not feasible in practice, the following result provides a good justification for our language: by focusing on tBGPs we study the basing block of temporal query languages. To make this more precise, consider a different relational representation $I_G$, which now stores a tuple $(s,p,o,t)$ for each $(s,p,o,t) \in \hat G$. Further, let us say that a tBGP is compatible with a graph $G$ if the time instants mentioned in $Q$ are also in $\mathcal I_G$.

\begin{proposition}
\label{prop-cqs-point-graph}
for every tBGP $Q$ one can construct a conjunctive query $\hat Q$ with inequalities, so that
the answer of $Q$ over a compatible temporal graph $G$ is the same as the answer of $\hat Q$ over $\hat G$, and conversely, as long as the inequalities on the conjunctive query are only on variables used in the fourth position of relations.
\end{proposition}

\begin{proof}
Let $Q$ be a tBGP. Observe that $Q$ can be viewed syntactically as a
conjunctive query of arity $4$, possibly with comparison predicates
between variables occurring in the fourth (temporal) position.
Hence, $Q$ can be directly evaluated over the point-based
representation $\hat G$ as a standard conjunctive query with
inequalities. We show that the answers coincide in both settings.

Throughout the proof we consider only temporal graphs $G$ for which
$Q$ is \emph{well defined}, that is, all constants appearing in $Q$
belong to $\U_G$ and all temporal constants belong to $I_G$.

\medskip
\noindent
\textbf{($\Rightarrow$)}
Let $f$ be a solution of $Q$ over the temporal graph $G$.
Consider any tuple pattern $(x,y,z,w)$ of $Q$.
Since $f$ is a valid assignment, there exists a tuple
\[
(f(x),f(y),f(z),[\ti,\tf))
\in G
\]
such that $\ti \le f(w) < \tf$.
By definition of solutions, we may assume that
$f(w) \in \T_G$.
Then, by construction of the point-based representation,
the tuple
\[
(f(x),f(y),f(z),f(w))
\]
belongs to $\hat G$.
Hence every atom of $Q$ is satisfied in $\hat G$, and all
comparison predicates remain valid.
Therefore $f$ is also a satisfying assignment of $Q$
evaluated as a conjunctive query over $\hat G$.

\medskip
\noindent
\textbf{($\Leftarrow$)}
Conversely, let $f$ be a satisfying assignment of $Q$
over $\hat G$.
Since tuples of $\hat G$ have the form
$(s,p,o,t)$ obtained from intervals of $G$,
for every atom $(x,y,z,w)$ of $Q$ such that
\[
(f(x),f(y),f(z),f(w)) \in \hat G,
\]
there exists an interval $[\ti,\tf)$ with
\[
(f(x),f(y),f(z),[\ti,\tf)) \in G
\quad\text{and}\quad
\ti \le f(w) < \tf .
\]
Hence the atom is satisfied under the semantics of tBGPs.

Moreover, because the domain of values and the temporal domain
are disjoint, variables occurring in the first three positions
cannot be mapped to temporal values and vice versa.
Thus comparison predicates involving temporal variables are
preserved, and $f$ is a valid solution of $Q$ over $G$.
 \end{proof}

From Propositions \ref{prop-cqs-not-enough} and \ref{prop-cqs-point-graph}, it follows that none of the techniques developed for relational conjunctive queries can be directly ported into our temporal graph setting, unless we first expand temporal graphs to their point-based  representation, this, as we explained, is not feasible to do in practice. 

\section{Pseudocodes} \label{ap:pseudo}

We show pseudocode in Algorithms~\ref{alg:leapfrog} to  \ref{alg:wtree}.

\begin{algorithm}[t]
	$x \gets 0$ ; $l \gets 0$ ; \\
        \While{$true$}
           { $x \gets \leap(v_l,x)$ ; \\
	     $i \gets (l+1) \textrm{~mod~} k$ ; \\
             \While{$i \not= l$} 
               { $x' \gets \leap(v_i,x)$ ; \\
                 \lIf{$x' > x$} { $x \gets x'$ ; $l \gets i$ }
                 $i \gets (i+1) \textrm{~mod~} k$ ; \\
	       }
             \lIf{$x = +\infty$} { $\!\!\!\!$ {\bf break} $\!\!$ }
             {\bf report} $x$ ; \\
             $x \gets x+1$ ; 
	   }
\ \\
\caption{The LTJ iterator reports all common children of LTJ trie nodes
$v_0,\ldots,v_{k-1}$. $\leap(v,x)$ returns the next value
$\ge x$ descending from LTJ trie node $v$, or $+\infty$ if none exists.}
\label{alg:leapfrog}
\end{algorithm}

\begin{algorithm}[t]
$\leap(v,x,h)$  \vspace{2mm} \\
  	\lIf{$v = null$} { $\!\!\!\!$ \textbf{return} $+\infty$ }
  	\lIf{$v$ is a leaf} { $\!\!\!\!$ \textbf{return} $0$ }
	\lIf{$\lfloor x / 2^h\rfloor = 1$}
	   { $x \gets \leap(v.r,x-2^h,h-1)$
	   }
	\Else
	   { $x \gets \leap(v.l,x,h-1)$ ; \\ 
	      \lIf{$x \not= +\infty$} { $\!\!\!\!$ \textbf{return} $x$ }
	      $x \gets \mathsf{leftmost}(v.r,h-1)$ ;
           }
        \lIf{$x \not= +\infty$} { $x \gets x+2^h$ }
	\textbf{return} $x$ ; \\

\vspace{5mm}

\textsf{leftmost($v,h$)} \vspace{2mm} \\
    \lIf{$v = null$} { $\!\!\!\!$ \textbf{return} $+\infty$ }
    \lIf{$v$ is a leaf} { $\!\!\!\!$ \textbf{return} $0$ }
    \lIf{$v.l \not= null$} { \textbf{return} $\mathsf{leftmost}(v.l,h-1)$ }
     \textbf{return} $2^h+\mathsf{leftmost}(v.r,h-1)$ ; 

\vspace{5mm}
    
\caption{The implementation of $\leap$ using BTs (or, equivalently, VBTs).
The function receives
the LTJ trie node $v$ (which is identified with the BT root), the minimum
desired value $x$, and the BT height $h=\ell$. The left and right children of 
BT node $v$ are $v.l$ and $v.r$, respectively.}
\label{alg:BTseek}
\end{algorithm}

\begin{algorithm}[t]
        $x \gets 0$ ; $l \gets 0$ ; \\
        \While{$true$}
           { $[x,y) \gets \leap(v_l,x)$ ; \\
             $i \gets (l+1) \textrm{~mod~} k$ ; \\
             \While{$i \not= l$}
               { $[x',y') \gets \leap(v_i,x)$ ; \\
                 \lIf{$x' > x$} { $[x,y) \gets [x',y')$ ; $l \gets i$ } 
		 \lElse { $y \gets \min(y,y')$ }
                 $i \gets (i+1) \textrm{~mod~} k$ ; \\
               }
             \lIf{$x = +\infty$} { $\!\!\!\!$ {\bf break} $\!\!$ }
             {\bf report} $[x,y)\,$ ; \\
             $x \gets y$ ; \\
           }
\ \\
\caption{The modified LTJ iterator to account for intervals when intersecting 
at the time level. The iterator now returns intervals $[x,y)$ where all the 
results will be the same.} 
\label{alg:time-leapfrog}
\end{algorithm}

\begin{algorithm}[t!]
$\leap(d,s,e,p,x,h)$ \\ \vspace{2mm}
  	\lIf{$s+p \not\in [s,e] ~\lor~ E_d[s+p]=0$} { $\!\!\!\!$ \textbf{return} $+\infty$ }
  	\lIf{$h=0$} { $\!\!\!\!$ \textbf{return} $0$ }
  	$r \gets rank(B_d,s,e)$ ;\\
  	$p' \gets rank(B_d,s,s+p)$ ; \\
	\If{$\lfloor x / 2^h\rfloor = 1$}
	   { $x \gets \leap(d+1,e-r+1,e,p'-1,x-2^h,h-1)$ ; \\
	     \lIf{$x \not= +\infty$} { $\!\!\!\!$ \textbf{return} $x+2^h$ }
	     \textbf{return} $+\infty$;
	   }
	\Else
	   { $x \gets \leap(d+1,s,e-r,p-p',x,h-1)$ ; \\ 
	     \lIf{$x \not= +\infty$} { $\!\!\!\!$ \textbf{return} $x$ }
	     \lIf{$r = 0 ~\lor~ E_{d+1}[e-r+p'] = 0$} { $\!\!\!\!$ \textbf{return} $+\infty$ }
	     \textbf{return} $2^h + \mathsf{leftmost}(d+1,e-r+1,e,p'-1,h-1)$ ;
           }

\vspace{5mm}

\textsf{leftmost($d,s,e,p,h$)} \\ \vspace{2mm}
    \lIf{$h = 0$} { $\!\!\!\!$ \textbf{return} $0$ }
    $r \gets rank(B_d,s,e)$ ; \\
 	$p' \gets rank(B_d,s,s+p)$ ; \\
  	\If{$s \le e-r ~\land~ E_{d+1}[s+p-p'] = 1$} { \textbf{return} $\mathsf{leftmost}(d+1,s,e-r,p-p',h-1)$ }
	\textbf{return} $2^h+\mathsf{leftmost}(d+1,e-r+1,e,p'-1,h-1)$ ;

\vspace{5mm}

\caption{The implementation of $\leap$ on our linear-space representation. The function receives
the LTJ trie node that corresponds to the range $[s,e]$ at level $d$ of our structure, the local timestamp of interest $p=p_l-s$, the minimum
desired value $x$, and the BT height $h=\ell$.}
\label{alg:wtree}
\end{algorithm}

\section{Querying the graphs $G_t$ and $G\Xany{t_1}{t_2}$ (Extended version)} \label{ap:point-in-queries}

Definition~\ref{def:otherG} gives some standard labeled graphs that can be derived from a temporal graph $G$, on which it is of interest to answer standard BGPs. We now show how we can use our linear-space data structure on $G$ to answer BGPs on $G_t$ and $G\Xany{t_1}{t_2}$, where $t$, $t_1$, and $t_2$ are given together with the query, in wco time with respect to those graphs. The case of $G\Xall{t_1}{t_2}$ will be discussed later.

We can answer general BGPs $Q$ on $G_t$ by converting them to BGPs on the temporal graph $G$, adding the constant $t$ as the fourth component in all the triples. A particularly interesting result, however, is obtained if we instead descend by the attribute $\textsc{t}=t$ in the 6 LTJ tries that start with attribute {\sc t}. The subtries that descend from those nodes correspond to the LTJ tries for $G_t$, on the attributes {\sc s}, {\sc p}, and {\sc o}. For example, if we descend by $t$ in the trie {\sc tspo}, the subtrie of the node we arrive at is isomorphic to the trie {\sc spo} of $G_t$. We can then run the normal LTJ algorithm for $Q$ using those subtries, exactly as if they were the LTJ tries of $G_t$. Since our $O(N)$ space data structure simulates the operation leap on those subtries in $O(\log N)$ time, the result follows easily. This proves the first part of Thm.~\ref{thm:GtGanyt1t2-intro}.

%\begin{theorem} \label{thm:Gt}
%Let $G$ be a temporal graph with $N$ tuples. Then, there is a data
%structure using $O(N)$ space that can solve BGPs $Q$ with $m$
%tuples on the labeled graph $G_t$ for any time instant $t$ given with $Q$, in time
%$O(Q^*_t\, m \log N)$, where $Q^*_t$ is the maximum number of solutions for $Q$ in
%some labeled graph with $|G_t|$ triples.
%\end{theorem}

To obtain an analogous result for $G\Xany{t_1}{t_2}$, we must extend the navigation of $G_t$ using VBTs, described in Section~\ref{sec:navig}, to intervals $[t_1,t_2)$. The key idea is that, if $t_1$ and $t_2$ are represented by $p_1$ and $p_2+1$ at some node $[s_v,e_v]$, then $rank(E_d,s_v+p_1,s_v+p_2)=0$ iff the node $s_v$ did not exist along the whole period between the local offsets $p_1$ and $p_2$, that is, it did not exist in $G\Xany{t_1}{t_2}$.

We start on the time level, where we find the first and last intervals, $[\ts_a,\te_a)$ and $[\ts_b,\te_b)$, that overlap or are contained in $[t_1,t_2)$ (it might be that $a=b$, but the answer is empty if no such intervals exist). Let these intervals correspond to timestamps $p_a$ and $p_b$ in our linear-space data structure $V$ below the time level.
We start at depth $d:=0$, with interval $[s_v,e_v] := [1,L]$ at the root $v$ of $V$, and with local offsets $p_1 := p_{a-1}-s_v$ and $p_2 := p_b-s_v$. If $rank(E_d,s_v+p_1+1,s_v+p_2)=0$, then the VBT node was null all along the time interval $[t_1,t_2)$; that is to say, there are no elements of $G\Xany{t_1}{t_2}$ descending from $v$. Otherwise, we can go left or right.  
We compute $r := rank(B_d,s_v,e_v)$ and $p_{1/2}' := rank(B_d,s_v,s_v+p_{1/2})$. To descend left, we update $e_v := e_v-r$ and $p_{1/2} := p_{1/2}-p_{1/2}'$; to descend right we update $s_v := e_v-r+1$ and $p_{1/2} := p_{1/2}'-1$.

A special case occurs in this process if $p_1$ and $p_2$ become equal. This means that the subtree of $v$ had no update during $[t_1,t_2)$, and therefore it has tuples below it during that period iff $E_d[s_v+p_1] = 1$, assuming $E_d[0]=0$. Algorithm~\ref{alg:wtreeRange} shows how to modify Algorithm~\ref{alg:wtree} to perform $\leap$ on this simulated trie.

Note that if the algorithm arrives at a node $v$, there exists at least one edge below $v$ during $[t_1,t_2)$; therefore we do not spend any time on nodes that do not exist in the LTJ trie of $G\Xany{t_1}{t_2}$. This yields the second part of Thm.~\ref{thm:GtGanyt1t2-intro}.

 \begin{algorithm}[t!]
$\leap(d,s,e,p_1,p_2,x,h)$ \\ \vspace{2mm}
  	\lIf{$s{+}p \not\in [s,e] \lor rank(E_d,s{+}p_1{+}1,s{+}p_2){=}0$} { $\!\!\!\!$ \textbf{return} $+\infty$ }
  	\lIf{$h=0$} { $\!\!\!\!$ \textbf{return} $0$ }
  	$r \gets rank(B_d,s,e)$ ;\\
  	$p_1' \gets rank(B_d,s,s+p_1)$; $p_2' \gets rank(B_d,s,s+p_2)$ ;\\
	\If{$\lfloor x / 2^h\rfloor = 1$}
	   { $x \gets \leap(d+1,e-r+1,e,p_1'-1,p_2'-1,x-2^h,h-1)$ ; \\
	     \lIf{$x \not= +\infty$} { $\!\!\!\!$ \textbf{return} $x+2^h$ }
	     \textbf{return} $+\infty$ ;
	   }
	\Else
	   { $x \gets \leap(d+1,s,e-r,p_1-p_1',p_2-p_2',x,h-1)$ ; \\ 
	     \lIf{$x \not= +\infty$} { $\!\!\!\!$ \textbf{return} $x$ }
	     \lIf{$r = 0 ~\lor~ rank(E_{d+1},e-r+p_1'+1,e-r+p_2') = 0$} { $\!\!\!\!$ \textbf{return} $+\infty$ }
	     \textbf{return} $2^h + \mathsf{leftmost}(d{+}1,e{-}r{+}1,e,p_1'{-}1,p_2'{-}1,h{-}1)$ ;
           }

\vspace{5mm}

\textsf{leftmost($d,s,e,p_1,p_2,h$)} \\ \vspace{2mm}
    \lIf{$h = 0$} { $\!\!\!\!$ \textbf{return} $0$ }
    $r \gets rank(B_d,s,e)$ ;\\
    $p_1' \gets rank(B_d,s,s+p_1)$ ; $p_2' \gets rank(B_d,s,s+p_2)$ ;\\
  	\If{$s \le e-r ~\land~ rank(E_{d+1},s+p_1-p_1'+1,s+p_2-p_2') > 0$} { \textbf{return} $\mathsf{leftmost}(d+1,s,e-r,p_1-p_1',p_2-p_2',h-1)$ ;}
	\textbf{return} $2^h+\mathsf{leftmost}(d+1,e-r+1,e,p_1'-1,p_2'-1,h-1)$ ;

\vspace{5mm}

\caption{The implementation of $\leap$ on our linear-space representation, simulating the graph $G\Xany{t_1}{t_2}$. The function receives
the LTJ trie node that corresponds to the range $[s,e]$ at level $d$ of our structure, the local timestamps of interest $p_1=p_{a-1}-s$ and $p_2=p_b-s$, the minimum
desired value $x$, and the BT height $h=\ell$. We assume that $rank(E_d,i,i-1)$ returns $E_d[i-1]$, assuming $E_d[0]=0$.}
\label{alg:wtreeRange}
\end{algorithm}

%\begin{theorem} \label{thm:Gany}
%Let $G$ be a temporal graph with $N$ tuples. Then, there is a data
%structure using $O(N)$ space that can solve BGPs $Q$ with $m$
%tuples on the labeled graph $G\Xany{t_1}{t_2}$ for any time interval $[t_1,t_2)$ given with $Q$, in time
%$O(Q^*_{[ t_1,t_2\rangle}\, m \log N)$, where $Q^*_{[t_1,t_2\rangle}$ is the maximum number of solutions for $Q$ in
%some labeled graph with $|G\Xany{t_1}{t_2}|$ triples.
%\end{theorem}

We note that the so-called ``join-first'' strategy, which solves $Q$ on the labeled graph $G_\mathrm{all} = G\Xany{1}{T}$ of all the tuples that ever existed and then filters the results by time, can offer the AGM bound only on $|G_\mathrm{all}|$, whereas our result bounded by $|G\Xany{t_1}{t_2}|$ is the best one can hope for: the algorithm is wco on the labeled graph that has exactly the triples we want to consider.

\paragraph{Triple patterns.}
Consider the case of a BGP formed by a single triple pattern $(s,p,o)$. Our algorithm will first descend in the time level, computing $a$ and $b$ with a binary search, and then descend in $V$ by the constants in $(s,p,o)$. For each node arrived at, there exists at least one triple in $G\Xany{t_1}{t_2}$, and the algorithm will descend by all left and right branches in the VBTs, reporting all the triples. The total time is instance-optimal, $O(\log N)$ per reported triple. The existential query, that is, telling whether the triple pattern has a match or not in $G\Xany{t_1}{t_2}$, takes $O(\log N)$ time. Even this simple query is not handled near-optimally with the ``join-first'' strategy: there could be many matches for $(s,p,o)$ out of $[t_1,t_2)$. It is also not handled well with ``time-first'', which individually considers each time $t \in [t_1,t_2)$ and solves the query on $G_t$: the triple pattern may not appear in many time instants $t$.

\paragraph{Listing time intervals.}
It is also possible to list all the maximal intervals where each match $(s,p,o)$ of the BGP existed during $[t_1,t_2)$, each in $O(\log N)$ time. In the final level $\ell$, each position in $E_\ell[s+p_1+1,s+p_2]$ corresponds to an event where the triple $(s,p,o)$ is successively inserted (1) and deleted (0). Those offsets $p_1 < p \le p_2$ can be mapped to the corresponding time instants as we return from the recursive traversal of the VBT. 
Say that from our range in $E_d[s,e]$ we went to the left child, thereby arriving at the range $E_{d+1}[s,e-r]$ with $r := rank(B_d,s,e)$. Then, a position $E_{d+1}[s-1+p]$ corresponds to the position $E_d[s+select_0(B_d,s,p)]$, where $select_b(B_d,s,p)$ is the position of the $p$th occurrence of bit $b \in \{0,1\}$ in $B_d[s..]$. Analogously, if we went to the right child of $E_d[s,e]$, arriving at the range $E_{d+1}[e-r+1,e]$, a position $E_{d+1}[e-r+p]$ corresponds to the position $E_d[s+select_1(B_d,s,p)]$. Query $select$ is solved in constant time using $o(|B_d|)$ bits on top of $B_d$ \cite{Cla96,Mun96}. After $O(\log N)$ steps we reach the root of $V$, where the timestamp $p$ can be binary searched among the $p_l$ values assigned to the intervals $[\ts_l,\te_l)$ to convert it to time instants in $\T$.

\begin{theorem}
Let $G$ be a temporal graph with $N$ tuples. Then, there is a data
structure using $O(N)$ space that can report all the $occ$ occurrences of a single-triple-pattern query $Q=\{(s,p,o)\}$ in the labeled graph $G\Xany{t_1}{t_2}$, for any time interval $[t_1,t_2)$ given with $Q$, in time $O((1+occ)\log N)$.
It can also list each maximal time interval where each occurrence appears in time $O(\log N)$.
\end{theorem}

Listing the time intervals where a triple pattern existed in $[t_1,t_2)$ also permits tracking the differences between $G_{t_1-1}$ and $G_{t_2-1}$, an important operation on versioned graphs. This can go from tracking a single triple $(s,p,o)$ to listing all the changes in the graph.

\section{Queries with Duration (Extended version)} \label{ap:duration-queries}

In some cases we are interested in requiring that the solutions to the queries last for some time. We consider two cases of such queries. A first one is to establish a time interval where the solution must always hold, so that we are interested only in that time interval. In other words, we must run the query on $G\Xall{t_1}{t_2}$; recall Definition~\ref{def:otherG}. A second one does not fix the exact times, but sets a minimum duration $\delta$ along which the reported solutions must hold \cite{HSGAY22}.

\subsection{Querying the Graph $G\Xall{t_1}{t_2}$}

The idea used to query $G\Xany{t_1}{t_2}$ in optimal time can be extended to query $G\Xall{t_1}{t_2}$: $rank(E_d,s_v+p_1,s_v+p_2)=p_2-p_1+1$ iff the node $s_v$ exists throughout the whole period between the local offsets $p_1$ and $p_2$. Therefore, we can use Algorithm~\ref{alg:wtreeRange} with the only change that we require that all the bits in the area of $E_d$ are 1s, not just one of them. 
This strategy does not yield worst-case time guarantees in terms of the size of $G\Xall{t_1}{t_2}$, however: it is possible that a BT node may have existed all along $[t_1,t_2)$, while none of its children have. We can therefore traverse large parts of the VBT just to find out that there are no matches of $Q$ in $G\Xall{t_1}{t_2}$ (which could even be empty). 
On the other hand, the algorithm is still wco with respect to the size of $G_t$ for {\em any} $t \in [t_1,t_2)$: if the subtree of a node $v$ at depth $d$ does not exist in $G_t$, then the corresponding bit $E_d$ will be zero and the algorithm will not attempt to explore it. Instead of satisfying the AGM bound on $|G\Xall{t_1}{t_2}| = |\cap_{t \in [t_1,t_2)} G_t|$, we satisfy it in terms of $\min_{t \in [t_1,t_2)} |G_t|$. This proves the first part of Thm.~\ref{thm:Gallt1t2duration-intro}.

%\begin{theorem} \label{thm:Gall}
%Let $G$ be a temporal graph with $N$ tuples. Then, there is a data
%structure using $O(N)$ space that can solve BGPs $Q$ with $m$
%tuples on the labeled graph $G\Xall{t_1}{t_2}$ for any time interval $[t_1,t_2)$ given with $Q$, in time
%$O(Q^*_{[t_1,t_2)}\, m \log N)$, where $Q^*_{[t_1,t_2)} = \min_{t_1 \le t < t_2} Q^*_t$. % is the %maximum number of solutions for $Q$ in some labeled graph with $\min_{t \in [t_1,t_2]} |G_t|$ triples.
%\end{theorem}

\subsection{Setting a minimum duration $\delta$}

Consider the problem of solving a BGP on the triples $(s,p,o)$ of a temporal graph $G$ so that the BGP exists in $G$ during an interval of $\delta$ time units or more. We solve this query on $G$, adding the same variable $t$ as the fourth component of all the triple patterns, and use the LTJ tries where the time component {\sc t} is at the end, as in join-first approaches. Note that those tries do not use the linear-space representation of Section~\ref{sec:linear}. Instead, we will use a more sophisticated data structure to store the children of each LTJ trie node $v$ in order to obtain relevant time guarantees. 

Let $u_1,\ldots,u_k$ be the children of node $v$, with increasing values $x_1,\ldots,x_k$. In addition to storing an array with the (values of) the nodes $u_1,\ldots,u_k$, we store $k$ points $(i,\delta_i)$ in a geometric grid, where $\delta_i$ is the longest duration of a tuple of $G$ stored below node $u_i$ (here we refer to the original durations of the tuples, before we mapped them to $[0,T)$). When it comes to solve $\leap(v,x)$, we look for the point $(i,\delta_i)$ with minimum value $i$ such that $x_i \ge x$ and $\delta_i \ge \delta$. This is an orthogonal range geometric query in two dimensions, which a linear-space data structure solves in $O(\log N)$ time \cite{CLP11}, thereby not affecting our complexities. Note that, as we represent times in $\T \setminus \T_G$ with their predecessors in $\T_G$, we work with the maximum possible durations and do not lose any solution.

We leave variable $t$ to be bound at the end. When we reach the last level, {\sc t}, in all the triple patterns of the BGP, we must intersect the $m$ time intervals, aiming again at finding time ranges of length at least $\delta$.
In this level we also convert every node $[\ts_i,\tf_i)$ into a point $(i,\delta_i)$ of a two-dimensional grid, with $\delta_i$ the actual duration of the interval $[\ts_i,\tf_i)$. We can then run the intersection algorithm for the time level as described in Section~\ref{sec:inters} (precisely, the version with intervals of Algorithm~\ref{alg:time-leapfrog}), where $\leap(v,t_0)$ is solved as follows. First we find $i$ such that (1) $\ts_i \le t_0 < \te_i$
or (2) $\te_i \le t_0 < \ts_{i+1}$. In case (1), if the duration of $[t_0,\te_i)$ is at least $\delta$,\footnote{How this is defined depends on how we model durations. For example, if $\T = \mathbb{N}$, we may understand that the interval $[2,3)$ has a duration of $1$ (if we count the number of instants it spans) or less than $1$ (if we count its length on the real line).} we return
$t^* := t_0$. In any other case, the interval $[\ts_i,\te_i)$ is not useful and we search the grid for the point $(j,\delta_j)$ with smallest coordinate $j$ such that $j > i$ and $\delta_j \ge \delta$. This yields the closest time interval to the right that contains an edge whose duration is long enough, and is found in $O(\log N)$ time with the geometric data structure.

It is easy to see that our search algorithm works exactly as if we had run the query on $G$, yet completely ignoring the edges shorter than $\delta$. This proves the second part of Thm.~\ref{thm:Gallt1t2duration-intro}.

\section{Implementation} \label{ap:implem}

Although our data structure uses linear space, the constant is high in practice because of the need to store $4! = 24$ tries with all the possible orderings of $\{ s,p,o,t \}$. Since each trie has $4$ levels and all the tuples are mentioned once in every level before the time level and twice since the time level, a (slightly pessimistic) upper bound on the number of LTJ trie nodes stored is $156N$ for a temporal graph of $N$ tuples. Further, our representation for the trie levels requires two words per node (i.e., two events), and versioned tries after the time level require four words per node (i.e., each event induces $\ell$ in each bitvector $B_d$ and $E_d$), which sets the count to $228N$ words. The space for the trie pointers can in general be dismissed  by using compact topology representations \cite{Jac89}. Considering that we need $5N$ words to represent the $N$ tuples $(s,p,o,[\ti,\tf\,])$ in raw form, the blowup factor in the space is $45.6$. 

\begin{figure}
    \centering
    \includegraphics[width=0.6\textwidth]{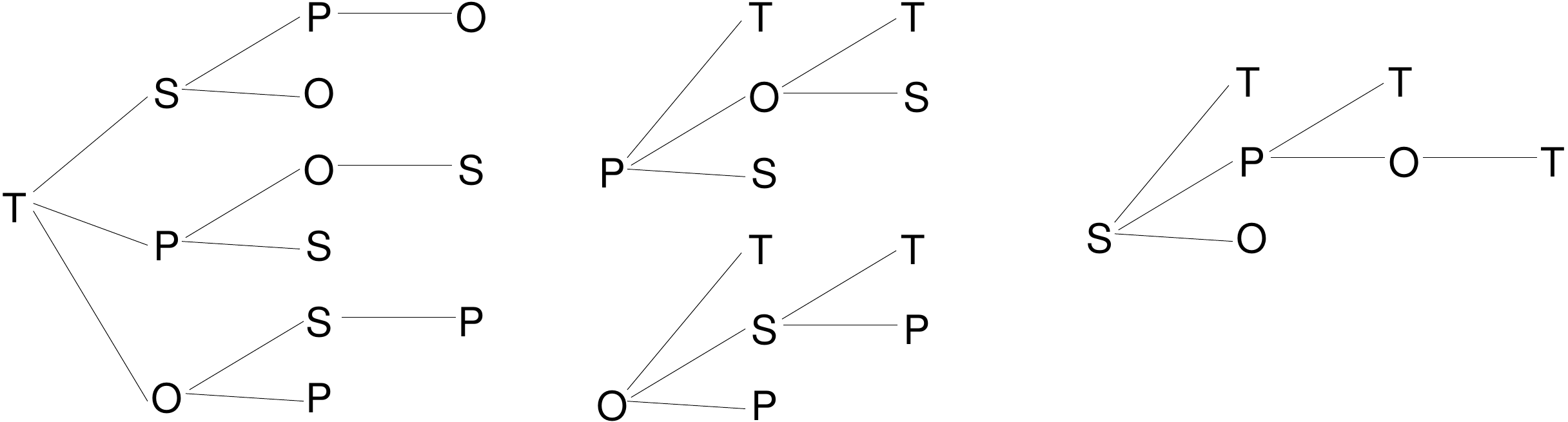}
    \caption{The meta-trie of all the LTJ tries we need to build.}
    \label{fig:metatrie}
\end{figure}

In order to reduce this overhead we make use of {\em trie switching} \cite{AGHNRRStods24} and {\em partial tries}. 
%Trie switching avoids representing some of the $4!$ orders as follows: assume we need to descend with order \textsc{psto}. We might not have that particular trie if we have instead those for \textsc{psot} and \textsc{spto}: once we have descended by $p$ and $s$ in the trie for \textsc{psot}, we can reenter the trie \textsc{spto} from the root, using the current values of $s$ and $p$, in order to continue descending by $t$ and $o$. This can be extended to the idea of partial tries: we can also truncate the trie for \textsc{spot}, to just \textsc{spo}: once we have entered with some $p$, $s$, and $o$, we can return to the trie for \textsc{psot} to reach the final level.
Another space-saving device is the possibility of sharing trie prefixes. For example, the tries for \textsc{spo} and \textsc{spto} can share their first two levels; the nodes of the level \textsc{sp} concatenate two sequence of children: those of the trie that continues with \textsc{o} and those of the trie that continues with \textsc{to}. This can also be emulated with the succinct topology representations that do not use pointers \cite{Jac89}, because we can simulate that they have a first child continuing with \textsc{o} and a second one with \textsc{to}. Our linear-space representation of Section~\ref{sec:linear} can be similarly adapted.
Those fake pointers add up to $O(N)$ bits of space. %\textcolor{red}{Diego puede querer agregar algo más preciso acá.}

With these mechanisms, we can define a {\em meta-trie} that stores all the paths that are stored in the index; each path corresponds to a (possibly partial) LTJ trie. The shared paths in the meta-trie correspond to shared paths in the LTJ tries. Therefore, the number of nodes in the meta-trie correspond to the factor multiplying $N$ in our storage space, if multiplied by the appropriate factor: 1 for levels before time, 2 for the time level, 4 for the levels after time. Figure ~\ref{fig:metatrie} contains a meta-trie that suffices to run LTJ on any instantiation order, and which uses just $76N$, $15.2$ times the space needed to store the raw data (our actual space in the experiments is 22\% less because the higher trie levels have fewer than $N$ elements).

To further reduce space, we use a recent compact trie representation \cite{cltj} for the levels up to the time component; subsequent ones are implemented as described in Section~\ref{sec:linear}.

\section{Benchmarks}
\subsection{Generation of \new{temporal intervals and} tBGPs for Wikidata} \label{app:bgps}

\new{We convert the information of some qualifiers into temporal annotations for the corresponding triples. We distinguish 65 qualifiers of this kind. The ones leading to most timestamps are P580 (``start time'', 10,645,813 annotations), P582 (``end time'', 5,430,256 annotations), P585 (``point in time'', 3,023,126 annotations), and P577 (``publication date'', 1,408,894 annotations); the others produce from a few tens of thousand to as few as 2 annotations. Some qualifiers, like P580 and P582, are paired and used to produce a time interval (open to one side in case only one of the two qualifiers occurs; typically having only P580 implies that the triple is valid up to the present time); the others, like P585 and P577, are used to produce an interval of only one time instant. All the granularities are uniformized to day-level. All the triples without temporal annotations are assumed to be valid for the whole universe of time instants; we say those temporal intervals are ``trivial''.}

To generate tBGPs, we filter disconnected and duplicate BGPs (modulo graph isomorphism). In order to ensure that the BGPs we use for experiments are likely to touch the non-trivial temporal annotations, we define a score to prioritize BGPs with more triple patterns whose predicates have a higher ratio of triples with non-trivial temporal annotations. More specifically, for each predicate $p_i$, in case it is a constant, we define $r_i$ as the ratio of triples with predicate $p_i$ that have non-trivial temporal annotations in the graph; otherwise, if $p_i$ is a variable, we define $r_i$ as 0. For each BGP extracted from the log of the form $Q = \{(s_1,p_1,o_1), \ldots, (s_m,p_m,o_m)\}$, we compute the score $\sum_{i = 1}^m r_i$. We then extract the top 1,000 BGPs from the query logs per this score. Further filtering BGPs with constants not appearing in the graph, we arrive at 972 queries.  
Finally, we extend each triple pattern with a shared temporal variable $v \in \Vt$ to generate $Q' = \{(s_1,p_1,o_1,v), \ldots, (s_m,p_m,o_m,v)\}$. 

\new{\subsection{Other datasets}
\label{app:datasets}}

We downloaded Divvy, Yellow and Caida datasets from Zhu et al.~\cite{ZhuFY21}. We have left the triple sets in our repository for reproducibility.

Divvy records about 21 million trips of shared bikes in Chicago between 2013 and 2019, with start and end time: the source and target locations become the subject and object, and the type of bike is taken as the predicate, leading to 6,525 different predicates (and 633 nodes). We downloaded it from https://divvybikes.com/system-data.

Yellow records about 33 million taxi trips in New York City from 2022, also with start and end time: the source and target locations become again the subject and object, but this time the predicate is chosen to be the distance travelled. This leads to 8,289 predicates (and 263 nodes). We downloaded it from {\em TLC Trip Record Data}, https://www.nyc.gov/site/tlc/about/tlc-trip-record-data.page (the ``Yellow Taxi Trip Records'' of 2022, in PARQUET format).

Caida records about 16 million relations between autonomous systems on the Internet from 1998 to 2026. We downloaded it from {\em The CAIDA AS Relationships Dataset}, https://www.caida.org/catalog/datasets/as-relationships/, subdirectories serial-1/ and serial-2/. Edges have no labels (i.e., there is only one label) and span only one time instant; we join successive times of an edge into maximal ranges. The graph has 110,703 nodes.

\new{\section{Dynamism}
\label{ap:dynamism}}

Our data structure supports efficiently updating the graph with newer events, that is, insertion and removal of edges with progressively larger timestamps. 
The core data structure that must be updated is the linear-space one of Section~\ref{sec:linear}. Adding a new event at time $L+1$ (e.g., at time $7$ in Figure~\ref{fig:wtree}) for some $v_i$ implies appending a bit to the corresponding $b\ell$ bitvectors $v.B[1,L_v]$ and $v.E[1,L_v]$.

Let $v$ be the root node. We first increment $L_v$. Now, if $v_i$ is stored in the left child of $v$ (i.e., the first of the $b\ell$ bits of $v_i$ is a $0$), we set $v.B[L_v] \gets 0$ and continue recursively at the left child; otherwise we set $v.B[L_v] \gets 1$ and continue recursively at the right child. We continue descending, according to the following bits of $v_i$, until reaching the leaves, where nothing is done. At the return of the recursion we set the values $v.E[L_v]$. As explained, a leaf is empty iff it holds an even number of events, whereas an internal node $u$ is empty iff $u.E[L_v] = 0$. We then set $v.E[L_v] \gets 0$ if both children of $v$ are empty or nonexistent, otherwise we set $v.E[L_v] \gets 1$. We must also update the data structures to compute $rank$ (and $select$ if needed, see Appendix~\ref{ap:point-in-queries}), so as to account for the bits appended to $v.B$ and $v.E$. Those structures \cite{Cla96,Mun96} are arrays that summarize information on successive blocks of the bitvectors, so when the bitvectors grow, only a constant amount of work is needed to update the information on the last block or to append a new block. Overall, adding a new event in the linear-space data structure takes time $O(b\ell) = O(\log N)$.

In the dynamic structure, we cannot concatenate all the bitvectors of each level $d$ into a single one, $B_d$ and $E_d$, because appending bits to node bitvectors $v.B$ or $v.E$ would require inserting bits in the middle of some $B_d$ or $E_d$, and that cannot be done in constant time. Instead, we maintain separate bitvectors $v.B$ and $v.E$ for each node $v$. We had avoided this to prevent spending $O(L\log N)$ space for node data, as there can be up to $L \cdot b\ell$ nodes overall. Another bound to the number of nodes, however, is $N$, because there are at most $N$ distinct leaves and thus $O(N)$ internal nodes; the added space is then $O(N)$ (i.e., $O(N\log N)$ bits), which is also linear in the graph size. This works because, as shown in Figure~\ref{fig:metatrie}, we will store only one copy of the linear-space data structure (the leftmost trie in the figure, with this structure at the root). The bitvectors $v.B$ and $v.E$, and their additional $rank$ and $select$ data structures, can be stored as ``extendible arrays'' \cite{BCDMS99}, which support accessing and extending them in constant worst-case time with an additional space overhead of only $O(\sqrt{L\log N} + N\log N)$ bits.

Given the current time range $[0,T)$, a new triple $(s,p,o)$ that is inserted or deleted at time $T$ (which expands the time range to $[0,T+1)$) requires inserting $spo$, $so$, $pos$, $ps$, $osp$, and $op$ in the leftmost trie of Figure~\ref{fig:metatrie}. Each such string represents the $b\ell$-bit descriptions we insert in the linear-space data structure as discussed so far. 

We must also insert the triple in the other tries of Figure~\ref{fig:metatrie}: $pt$, $pot$, $pos$, $ps$, $ot$, $ost$, $osp$, $op$, $st$, $spt$, $spot$, and $so$. Note that some of those paths do not have a time component (e.g., $pos$) and thus correspond to normal tries indicating which strings exist at some time instant. Those tries can be implemented as classic structures, turning the array of children of each node (that we used to binary search for the desired child) into balanced trees to allow insertions and searches in $O(\log N)$ time. We never delete strings (e.g., $pos$) in those tries, as all of them existed at some point; we only add new ones upon insertions. 

The remaining tries have the time component at their last level (e.g., $pot$). In those last levels we store a sequence of time intervals $[ts_i,te_i)$, as described in Section~\ref{sec:timelevel}. In the dynamic case, we must allow writing $te_i=+\infty$ in the last range so that it extends automatically to the current last time when $T$ increases. For example, in the trie $pot$, the time level stores the intervals in which each string $po$ existed in the graph. We must then include $T$ in this last level. When inserting $(s,p,o)$ at time $T$, if the last range is of the form $[ts_i,te_i)$ with $te_i \neq +\infty$, we add a new range $[T,+\infty)$. When deleting $(s,p,o)$ at time $T$, then the last range $[ts_i,+\infty)$ is converted to $[ts_i,T)$.

Overall, we can retain the time complexities of the structure described in the paper, as well as its linear space, and can add new events at increasing timestamps in time $O(\log N)$.

\end{document}